\documentclass[12pt]{article}
\usepackage[top=1.25in, bottom=1.25in, left=1.25in, right=1.25in]{geometry}
\usepackage{amsfonts} 
\usepackage{amsmath,amsthm} 
\usepackage{sectsty} 
\usepackage{graphics,tabularx, multirow,booktabs,MnSymbol}
\usepackage[citestyle=authoryear,hyperref=true,sortcites=false,uniquename=false,mincitenames=1,
maxcitenames =3,uniquelist=false,maxbibnames=10,bibstyle=authoryear-comp,giveninits=true]{biblatex}
\DeclareNameAlias{sortname}{family-given} 
\renewbibmacro{in:}{}
\usepackage{graphicx} 

\usepackage{cancel,soul}
\usepackage[inline]{enumitem}
\usepackage{hyperref}
\usepackage{float, bm}
\hypersetup{
  colorlinks=true,linkcolor=blue, citecolor=red, filecolor=magenta, 
  linkbordercolor={0 1 1}, citebordercolor={1 0 0},
}
\usepackage{setspace} 
\usepackage{chngcntr}

\bibliography{biblio.bib}  

\theoremstyle{definition}

\newtheorem{theorem}{Theorem}
\newtheorem{manualtheoreminner}{Assumption}
\newenvironment{manualassumption}[1]{%
  \renewcommand\themanualtheoreminner{#1}%
  \manualtheoreminner
}{\endmanualtheoreminner}
\newtheorem{corollary}{Corollary}
\newtheorem{lemma}{Lemma}[section]
\newtheorem{remark}{Remark}
\theoremstyle{remark}
\newtheorem{example}{Example}

\title{Continuous permanent unobserved heterogeneity in dynamic discrete choice models
\thanks{I am grateful to Federico Bugni, Adam Rosen, Matt
    Masten, and Arnaud Maurel for their guidance and support. I also thank St\'ephane Bonhomme, Giovanni Compiani, Jeremy Fox, Dave Kaplan, Andriy Norets, Yuya Sasaki, Xun Tang, colleagues at Duke University, as well as various seminar participants
    for useful comments and suggestions.}}
\author{Jackson Bunting\thanks{Department of Economics, University of Washington, \href{mailto:buntingj@uw.edu}{buntingj@uw.edu}.}}
\date{\today}
\begin{document}

\maketitle
 \begin{abstract}

In dynamic discrete choice (DDC) analysis, it is common to use mixture models to control for unobserved heterogeneity. However, consistent estimation typically requires both restrictions on the support of unobserved heterogeneity and a high-level injectivity condition that is difficult to verify. This paper provides primitive conditions for point identification of a broad class of DDC models with multivariate continuous permanent unobserved heterogeneity. The results apply to both finite- and infinite-horizon DDC models, do not require a full support assumption, nor a long panel, and place no parametric restriction on the distribution of unobserved heterogeneity. In addition, I propose a seminonparametric estimator that is computationally attractive and can be implemented using familiar parametric methods.

\begin{description}
	\item[Keywords:] Mixture models, dynamic discrete choice problems, nonparametric identification, unobserved heterogeneity.
	\item[JEL Classification Codes: C14, C61]
\end{description}
\end{abstract}

\newpage

\section{Introduction}

In dynamic discrete choice (DDC) analysis, it is common to use mixture models to control for permanent unobserved heterogeneity. For instance, \textcite{keane1997career,cameron1998life} model the observed distribution of schooling and work decisions as a mixture of individuals with varying unobserved abilities, which differ across occupations.

However, the use of mixture models in DDC analysis has limitations. First, existing identification results restrict the permanent unobserved heterogeneity to be either discrete \parencite{kasahara2009nonparametric} or a scalar random variable  \parencite{hu2012nonparametric}. In the schooling and work example, this limitation may mean the mixture model does not capture the full richness of ability types and patterns of comparative advantage across occupations.

Second, identification of mixture DDC models depends on having `enough variation' in agent behaviour \parencite{kasahara2009nonparametric,hu2012nonparametric}, a condition that is typically assumed at a high level. In the context of the schooling and work example, `enough variation' might require that agents with different unobserved abilities respond adequately differently to changes in wages. Concretely, `enough variation' is an injectivity condition. To express the condition formally, let $P_t(a,x,b)$ represent the model-implied probability that an agent chooses action $a$ in period $t$ given observed covariates $x$ and persistent unobserved heterogeneity $b$. The `enough variation' assumption states for any signed measure $\mu$ on the support of persistent unobserved heterogeneity
\begin{equation}\label{eq:inject_gen}
        \int P_t(a,x,b)d\mu(b)=0~~ \text{for all }(a,x)
        ~ \Longrightarrow ~ \mu=0.
\end{equation}
That is, `enough variation' guarantees that distinct distributions of heterogeneity generate distinct average choice behavior in at least one state. An injectivity condition of this style is imposed in the existing indentification literature.\footnote{Specifically, Equation \eqref{eq:inject_gen} generalizes the rank condition assumed in Proposition 1 \textcite{kasahara2009nonparametric}, and is a specialization of Assumption 2 \textcite{hu2012nonparametric}
.} Yet, despite the crucial role of the injectivity assumption to identification,\footnote{Under some conditions, injectivity is equivalent to identification. See the discussion of Theorem \ref{thm:inject}.} there appear to be few results in the literature on whether it holds in a given DDC model. Gaining insights into the conditions under which injectivity holds is particularly significant given that the assumption, as stated in \textcite[][p. 151]{kasahara2009nonparametric}, ``is not empirically testable from the observed data.'' Moreover, verifying injectivity of an integral operator is known to be a challenging problem in general \parencite[e.g.,][]{andrews2017examples}.

The main contribution of this paper is to propose a general class of DDC models with permanent unobserved heterogeneity that is both continuous and multivariate, and provide low-level conditions for its identification. Applied to the schooling and work example, the class of DDC models in this paper would allow abilities to vary continuously across individuals and to be occupation-specific. I provide sufficient conditions for point identification of all model parameters, including the distribution of agent types (i.e., the distribution of permanent unobserved heterogeneity) and the type-specific choice model. By establishing low-level conditions for identification, the paper provides affirmation of the injectivity assumption for DDC models, demonstrating that it holds at least within one broad class of DDC models.

The paper contains two main results on identification of multinomial DDC models. The first result (Section \ref{sec:model}) pertains to DDC models with random coefficients. The second (Section \ref{sec:model-intercept}) relates to DDC models with random intercepts. I also prove several extensions to these main results, encompassing both stationary (i.e., infinite horizon) and non-stationary (i.e., finite horizon) DDC models. Furthermore, I show an important implication of the results under the additional restriction that permanent unobserved heterogeneity is discrete --- an assumption that is standard in applied work. In this case, a key modeling decision is the number of agent types (i.e., the number of support points of permanent unobserved heterogeneity),\footnote{In general, only a lower bound on the number of mixture components is identified (e.g., \textcite[Proposition 3]{kasahara2009nonparametric}) so identification of finite mixture models requires knowledge of an upper bound (e.g., \textcite[Theorem S.1]{freyberger2018non}). See Section \ref{sec:kwon} for discussion.} which may be a challenging decision if there is no theoretical guidance on the number of agent types. My identification results imply a solution to this problem: namely, that the number of agent types is identified if it is assumed to be finite. 

Within a standard DDC model in the style of \textcite{rust1987optimal,magnac2002identifying}, the low-level conditions for identification can be broadly categorized into two groups. First, I assume a short panel of observations with some continuous variation in the observed covariates, which is a natural prerequisite for nonparametric identification of a continuous latent distribution. Importantly, the results do not require the covariates to have full support, nor place parametric restrictions on the distribution of the permanent unobserved heterogeneity. Second, restrictions on the model primitives are used to ensure injectivity holds. These restrictions have three components: a distributional assumption on the random utility shock, a functional form assumption on the per-period payoffs, and a relevance condition on the covariates.\footnote{This is a key point of departure from the existing identification literature, which allow for more general DDC models at the expense of imposing injectivity at a high level.} The restrictions have the advantage of being low-level and interpretable. For example, the relevance condition can be interpreted as requiring (at least one) covariate to have a non-zero effect on the agent's utility. Moreover, and notably, many of the restrictions are commonly made in the literature. For example, it is common to make distributional assumptions on the random utility shock and functional form assumptions on the per-period payoffs \parencite{aguirregabiria2010dynamic}. In this way, the results of this paper demonstrate that commonly made assumptions impose structure on DDC models that is useful for proving the (otherwise high-level) injectivity condition.

To implement the identification results, I propose a novel estimation method. Existing DDC estimation methods which focus on the parametric case\footnote{In principle, standard DDC models may be semiparametric in the presence of continuous covariates, however, in practice, continuous covariates are often discretized and treated as such for estimation.} \parencite{aguirregabiria2002swapping,arcidiacono2011conditional} do not apply to the model of this paper, as the distribution of unobserved heterogeneity may be an infinite dimensional parameter of interest. Similarly, the computational complexity of DDC models means that immediately available nonparametric methods (such as sieve likelihood estimation) may be impractical. To address these issues I propose a two-step sieve M-estimator, and show it is consistent for the model parameters. I also propose a computationally convenient sieve space based on \textcite{heckman1984method}. Intuitively, the estimator approximates the possibly continuous distribution of permanent unobserved heterogeneity by a discrete distribution. In this setup, the `fixed grid' of support points of the approximating distribution is a tuning parameter of the sieve estimator. Computationally, this estimator is identical to an estimator for a model with finite types, but instead of the number of support points being an identifying assumption, it is simply a tuning parameter.

I illustrate the theory through a simulation exercise and an empirical application based on the labor supply model of \textcite{altuug1998effect}. In this model, agents value consumption and leisure, deciding each period whether to enter the workforce based on expected wages. My identification results allow individual labor productivity to be continuous and consistently estimated from the labor force participation model. The estimates indicate substantial heterogeneity in labor productivity, with a strongly skewed distribution. A counterfactual exercise measures how wages affect labor force participation across the productivity distribution, revealing a highly varied response.

After discussing related literature, I introduce the model and provide one main identification result (Section \ref{sec:model}). Section \ref{sec:extensions} contains the second main identification result (Section \ref{sec:model-intercept}) and other extensions, including to non-stationary DDC problems.  Section \ref{sec:estimation} proposes the two-step sieve M-estimator and shows its consistency. Section \ref{sec:sims} contains the simulation exercise, and Section \ref{sec:applic} the application.

\paragraph{Related literature.}

This paper is closely related to the literature on point identification of DDC models with persistent unobserved heterogeneity \parencite{kasahara2009nonparametric,hu2012nonparametric}. These papers use a short panel to identify type-specific conditional choice probabilities and the distribution of unobserved heterogeneity via an eigendecomposition of the observed data. As mentioned earlier, these papers consider persistent unobserved heterogeneity that is either discrete \parencite{kasahara2009nonparametric} or a scalar random variable  \parencite{hu2012nonparametric}. Relative to these papers, I allow for permanent unobserved heterogeneity that is both continuous and multivariate. As previously mentioned, another important difference is that I provide low-level conditions for the injectivity condition. On the other hand, their approach allows unobserved heterogeneity to enter the model very flexibly, restricted only by certain high-level assumptions.\footnote{However, it is worth noting that \textcite{hu2012nonparametric} do not allow for identification of permanent unobserved heterogeneity from variation in choice behavior alone. Specifically, \textcite{hu2012nonparametric} Assumption 3(ii) requires variation in the state transition by type. To see this, in their notation let
  $W_{t}=(Y_{t},X_{t})$ be observed and $X_{t}^{*}=X^{*}$ latent, then their equation (11) becomes 
\begin{align*}
k\left(w_{t}, \bar{w}_{t}, w_{t-1}, \bar{w}_{t-1}, x^{*}\right)=\frac{f_{X_{t}|X_{t-1},Y_{t-1},X^{*}}(x_{t}|x_{t-1},y_{t-1},x^{*})f_{X_{t}|X_{t-1},Y_{t-1},X^{*}}(\overline{x}_{t}|\overline{x}_{t-1},\overline{y}_{t-1},x^{*})}{f_{X_{t}|X_{t-1},Y_{t-1},X^{*}}(\overline{x}_{t}|x_{t-1},y_{t-1},x^{*})f_{X_{t}|X_{t-1},Y_{t-1},X^{*}}({x}_{t}|\overline{x}_{t-1},\overline{y}_{t-1},x^{*})},
\end{align*}
and thus their Assumption 3(ii) which requires $k\left(w_{t}, \bar{w}_{t}, w_{t-1}, \bar{w}_{t-1}, x^{*}\right)$ to vary in $x^{*}$ fails if the state transition $f_{X_{t}|X_{t-1},Y_{t-1},X^{*}}$ does not depend on $X^{*}$. \textcite{williams2019nonparametric} also makes this point.} For example, my assumptions rule out type-specific transition functions (e.g., \textcite[Section 3.2]{kasahara2009nonparametric}) or unobserved heterogeneity that is first-order Markov (e.g., \textcite{hu2012nonparametric}). See also \textcites{williams2019nonparametric,higgins2023identification} as well as the general review \textcite{compiani2016using}.

Several other papers have analyzed persistent unobserved heterogeneity in DDC models from a partial identification perspective. For instance, \textcite{aguirregabiria2018sufficient} focuses on (point) identification of a subvector of the model parameters, treating permanent unobserved heterogeneity as a nuisance parameter. The related \textcite{aguirregabiria2021identification} considers a DDC model with fixed effects. Some general approaches that allow for set identification include \textcite{chernozhukov2013average,berry2022compiani}. Compared to these papers, I provide conditions for point identification of the DDC model. The paper is also related to the large literature on identification of the distribution of continuous unobserved heterogeneity in binary response models. One stream exploits a linear index and full support covariates, while leaving the distribution of random preference shocks unspecified (\textcite{ichimura1998maximum,lewbel2000semiparametric,gautier2013nonparametric}, among others). Relative to these papers, a DDC model yields a non-linear index with additive parametric preference shocks.

The seminonparametric estimator I propose is based on \textcite{heckman1984method}. Similar `fixed grid' estimators have been analyzed for both the parametric and non-dynamic models \parencite{fox2011simple,fox2016simple}, and are increasingly used in applied work \parencite[e.g.,][]{nevo2016usage,illanes2019retirement}.

\textbf{{Notation:}} For a random variable $X$, $\mathrm{Supp}(X)$ and $f_X$ denote the support and probability density (or mass) function.

\section{Model and identification}\label{sec:model}

\subsection{Model setup}

I consider a standard single-agent dynamic discrete choice structural model as described in \textcite{aguirregabiria2010dynamic}. In each period $t=1,\ldots,T=\infty$, a single agent observes a vector of state variables $(S_{t},\epsilon_t)$ and chooses an action $A_{t}$ from a finite set of actions $A \equiv\{0,1, \ldots,J\}$ (with $J>0$) to maximize expected utility. I assume $\epsilon_t=(\epsilon_{t,a}:a\in A)$ is independent of $\left(\epsilon_\tau, {A}_\tau, S_{\tau+1}\right)$ for $\tau< t$, and is identically distributed according to $d F_{\epsilon}(e)=\prod_a dF_{\epsilon_a}(e_a)$. In addition, conditional on $(A_t,S_t)=\left({a}_t, s_t\right)$, $S_{t+1}$ is independent of $\left(\epsilon_\tau, {A}_{\tau-1}, S_{\tau-1}\right)$ for $\tau \leq t$, with probability distribution $d F_s\left(s_{t+1} \mid {a}_t, s_t\right)$. It then follows that $\left(S_{t+1}, \epsilon_{t+1}\right)$ is a Markov process with a probability density that satisfies
\begin{equation}\label{eq:inf_ci}
d \operatorname{Pr}\left(S_{t+1}=s', \epsilon_{t+1}=e' \mid S_t=s, \epsilon_t=e, A_t={a}\right)=d F_{\epsilon}\left(e'\right) \times d F_s\left(s' \mid {a}, s\right) .
\end{equation}

The agent has a time-separable utility and discounts future payoffs by $\rho \in[0,1)$, where the period $t$ payoff is $u_t(S_t,\epsilon_t,A_t)$. Under these conditions, the agent's choice in time $t$ satisfies
\begin{equation}\label{eq:inf_prob}
a_t=\arg\max_{a\in A}\left\{u_t(s_t,e_t,a)+\rho E[v_{t+1}(S_{t+1})\mid S_t=s_t,A_t=a]\right\},
\end{equation}
where $v_t$ is the so-called integrated value function:
\begin{equation}\label{eq:fixed_point}
    v_t(s_{t})=E\left[\max_{a\in A}\left\{u_t(s_t,\epsilon_t,a)+\rho E[v_{t+1}(S_{t+1})\mid S_t=s_t,A_t=a]\right\}\right].
\end{equation}

In this section I present conditions for identification of the distribution of continuous unobserved heterogeneity within the above model. The first assumption imposes restrictions that are standard for stationary DDC models without permanent unobserved heterogeneity.

\begin{manualassumption}{{I1}}\label{As:inf_bl} (i) $u_t(S_{t},\epsilon_{t},A_{t})=u(S_{t},A_{t})+\sum_{a\in A}\epsilon_{t,a}1[a=A_t]$. (ii)
$\rho\in[0,1)$ is known.
(iii) Equation \eqref{eq:inf_ci}.
 (iv) $u(S_{t},0)=0$.
 (v) $\epsilon_{t,a}$ is independent over
agents, actions and time and distributed extreme value type I. (vi) $\mathrm{Supp}(S_t)$ is bounded.
\end{manualassumption}

Assumption \ref{As:inf_bl} include standard identifying assumptions for DDC models \parencite{magnac2002identifying,aguirregabiria2010dynamic}, including additive separability of the flow utility, that the discount factor is known, a conditional independence assumption, and the outside good. These assumptions are not innocuous --- for example, \textcite{norets2014semiparametric} show that the choice of outside good may affect predicted counterfactual outcomes. Nevertheless, it is standard to assume the unobserved state variables have a known distribution, of which normal and extreme value type I are common choices. It is also common to assume that $S_t$ lies in a compact set, which helps ensure the integrated value function is a bounded function of $S_t$ \parencite{rust1987optimal,kristensen2019solving}.

The next assumption introduces permanent unobserved heterogeneity into the model as an unobserved state variable.

\begin{manualassumption}{I2}\label{As:inf_puh}
\begin{enumerate*}[label=(\roman*)]
    \item $S_t=(X_t^\intercal,\beta^\intercal)^\intercal\in\mathbb{R}^{k+J}$, and $k=J+1$. For each $x\in \mathrm{Supp}(X_1)$, $\beta\mid X_1=x$ admits a bounded density $f_{\beta|X_{1}}$.
    \item $u(s,a)=x^\intercal\left(\beta_{a},\;\gamma_{a}^\intercal\right)^\intercal$.
\item $d\Pr(X_{t+1}=x' \mid A_t=a, X_t=x, \beta=b )= d F_x(x'\mid x,a)$.
\item\label{As:inf_homog}
$\Gamma\equiv(\gamma_{1}\gamma_{2}\cdots\gamma_{J})\in\mathbb{R}^{J\times J}$ is full rank.
\item\label{As:inf_tran}
The probability distribution of $X_{t+1}$ conditional upon $(A_t,X_t)=(a,x)$ has no singular components, and the associated probability density and mass functions are real analytic functions of $x$ with bounded analytic
continuations to $\mathbb{R}^{k}$.
\end{enumerate*} 
\end{manualassumption}

Assumption \ref{As:inf_puh}(i) states that permanent unobserved heterogeneity enters the model as an unobserved state variable. The restrictions placed on its distribution are mild. First, it allows the distribution to have uncountable support. Intuitively this means there may be infinitely many types of agents.\footnote{\label{fn:pmf}One may replace the probability density function in Assumption \ref{As:inf_puh}(i) with probability mass function and the subsequent results go through with minor modification. That is, the results allow for the typical assumption of finitely many types as a special case.} Second, there may be arbitrary dependence between the initial state variable and permanent unobserved heterogeneity.%

Assumption \ref{As:inf_puh}(i) further imposes that the dimension of the permanent unobserved heterogeneity is equal to the size of the choice set minus one (i.e., $\dim(\beta)=J$). It also requires that the dimension of the observed state variable equals the dimension of the permanent unobserved heterogeneity plus one (i.e., $k=\dim(\beta)+1$). Combined with part (ii), this implies that the model has $J$ variables with action-specific but agent-homogeneous effects via $\gamma_a$, and one variable with action- and agent-specific effects. It is straightforward, however, to allow for additional state variables with agent-homogeneous effects (i.e., $k \ge \dim(\beta)+1$ and $\dim(\beta)=J$); see Remark \ref{rem:discrete} for further discussion.

Parts (ii) and (iii) of Assumption \ref{As:inf_puh} control how permanent unobserved heterogeneity enters the model. Part (ii) states that the permanent unobserved heterogeneity enters the model as a random coefficient in the per-period payoff. Importantly, the continuous $\beta$ is vector-valued, allowing its effect to differ across different choice alternatives. By making the unit and time subscripts explicit in part (ii), i.e.,
$$
u(s_{i,t},a_{i,t}) = x_{i,t}^\intercal (\beta_{a,i},\gamma_a^\intercal)^\intercal,
$$
we see that $\beta_i = (\beta_{1,i},\ldots,\beta_{J,i})^\intercal$ can be viewed as an action-specific random effect associated with the first element of the state variable. For example, if $\beta_a$ represents an agent's ability in occupation $a\in A$, some agents may be high ability in all occupations, other agents may be high in some occupations and low in others. Part (iii) requires that the transition of the state variable not depend on the unobserved state variable. As explained below (Remark \ref{rem:transition}), this assumption enables conditions on the model primitives to be used for identification.

The next condition (Assumption I2\ref{As:inf_homog}) imposes that the state variable cannot affect payoffs for each choice in a similar fashion. For example, in the binary choice case ($J=1$), the assumption requires that $\gamma_{1}\neq0\in\mathbb{R}$.

Assumption I2\ref{As:inf_tran} allows the state transition to be a mixture of an absolutely continuous and discrete random variable, but restricts the probability distribution to be a smooth function of the conditioning state variable. In particular, the component probability density and mass functions must be real analytic functions --- that is, functions that have a convergent power series representation. An example of a state transition satisfying Assumption I2\ref{As:inf_tran} is a mixture of a mass point at $x_{t+1}=0$ and a truncated normal: $F_{x}(x';x,a)=\pi{1}(x'=0)+(1-\pi)F_{+}(x';x,a)$, where $F_{+}(x';x,a)$ is a truncated normal whose mean and variance are real analytic functions of $(x,a)$. Other examples of real analytic functions include polynomials, the logistic function, trigonometric functions, the Gaussian function, in addition to compositions, products and linear combinations of these functions. This class of functions is known to include good approximators to square-integrable functions \parencite[e.g.,][Section 2.3]{chen2007large}, and can therefore approximate many density functions arbitrarily well.

\subsection{Injectivity}

Define the conditional choice probability (CCP) function $P(a,x,b)$ to be the model implied probability that $A_t=a$ conditional upon $X_t=x$ and $\beta=b$. The first main theorem states that under the above conditions, the integral operator defined by the CCP function is injective.

\begin{theorem}[Injectivity]\label{thm:inject}
Assume \ref{As:inf_bl} and \ref{As:inf_puh}. Let $\mathcal{X}\subseteq\mathrm{Supp}(X_t)$ be a non-empty open set,  and let $\mu$ be an absolutely continuous finite signed measure on $\mathrm{Supp}(\beta)$. If
\begin{equation*}
\int P(a,x,b)\,d\mu(b) = 0 \quad \text{for almost every } (a,x)\in A\times\mathcal{X},
\end{equation*}
then $\mu=0$, the zero measure.
\end{theorem}

The injectivity condition in Theorem \ref{thm:inject} is fundamental to identification of mixture models. To explain, consider the simple case that $\beta$ is independent of $X_t$ and that the CCP function is known.\footnote{Since the state transition is identified directly from the data, given the model specified in Assumptions \ref{As:inf_bl} and \ref{As:inf_puh}, the CCP function is known if $\gamma=\left\{\gamma_a\in\mathbb{R}^{k-1}:a=1,\ldots,J\right\}$ is known.}. In this case, the data satisfies $\Pr(A_t=a\mid X_t=x)=\int P(a,x,b)dF_\beta(b)$ and the only unknown model parameter is $F_\beta$, the distribution of permanent unobserved heterogeneity. Then, supposing (the interior) of $\mathrm{Supp}(X_t)$ is non-empty and open, the injectivity condition is equivalent to identification of the distribution of unobserved heterogeneity: it states that if two distributions $F_\beta$ and $\tilde{F}_\beta$ are observationally equivalent, i.e.,
\begin{equation*}
    \int P(a,x,b)dF_\beta(b) = \int P(a,x,b)d\tilde{F}_\beta(b)
\end{equation*}
for almost every $(a,x)\in A\times\mathrm{Supp}(X_t)$, then the two distributions are the same, i.e., $F_\beta=\tilde{F}_\beta$. More generally, the injectivity condition in Theorem \ref{thm:inject} is an example of the injectivity assumption in the measurement error literature \parencite[Assumption 3]{hu2008instrumental}, with analogs in the context of DDC models \parencites[Proposition 1]{kasahara2009nonparametric}[Assumption 2]{hu2012nonparametric}.

The proof of Theorem \ref{thm:inject} is provided in Appendix \ref{pf:inj}. Before presenting an outline, a few comments are in order.

\begin{remark}[\textbf{Support of $\bm{X_t}$}]
    Theorem \ref{thm:inject} relies on having continuous variation in the observed state variable: namely that $\mathrm{Supp}(X_t)$ contains a non-empty open set $\mathcal{X}$. Given that injectivity is equivalent to the set $\{ b\mapsto P(a,x,b):(a,x)\in A \times{\mathcal{X}}\}$ being dense in all square integrable functions (see the below overview of the proof), it is natural to require that the set has infinitely many elements. However, importantly, $\mathrm{Supp}(X_t)$ may be arbitrarily small so long as it contains a non-empty open set. As described in the below proof outline, this is an implication of $P$ being real analytic.
\end{remark}

\begin{remark}[\textbf{Discrete state variables}]\label{rem:discrete}
    For notational simplicity, the formal statements in this paper focus on the case that $k=\dim(\beta)+1$ and that each element of $X\in\mathbb{R}^k$ has some continuous component. However, with only notational changes, the results of this paper continue to apply when there are additional observed state variables (i.e., $k\ge \dim(\beta)+1$). In this more general case, there are no limitations on the support of the additional state variables. For instance, they may contain discrete variables such as a constant or indicator functions. See Appendix \ref{sec:discrete_state} for a statement of sufficient conditions for Theorem \ref{thm:inject} in the $k\ge \dim(\beta)+1$ case.
\end{remark}

\begin{remark}[\textbf{Type dependent transitions}]\label{rem:transition}
    In the case that the state transition depends on permanent unobserved heterogeneity (i.e., if Assumption \ref{As:inf_puh}(iii) did not apply), then the kernel of the integral operator useful for identification would depend on both the CCP $P(a,x,b)$ and the state transition $F_x(x';x,a,b)$. In this case, without a behavioral model of $F_x(x';x,a,b)$ it appears to be challenging to provide low level conditions for injectivity of the integral operator. \textcites[Proposition 6]{kasahara2009nonparametric}[Theorem 1]{hu2012nonparametric} provide an identification result for this case, using a high level injectivity assumption.
\end{remark}

\paragraph{Overview of proof of Theorem \ref{thm:inject}.}

Broadly, the argument has two steps: (i) characterizing injectivity in terms of the approximation properties of the CCP function, and (ii) showing that the CCP function satisfies this property.

The characterization of injectivity is developed in two parts. First, I use real analyticity to effectively expand the set of $x$ used to define injectivity. To explain this part, note that the CCP function inherits the smoothness properties of the utility function $u_t$, the state transition $F_x$, and the idiosyncratic shock $F_\epsilon$ (assumed in \ref{As:inf_bl}(i), I2\ref{As:inf_tran}, and \ref{As:inf_bl}(v), respectively). In particular, since these are real analytic, the function $x\mapsto P(a,x,b)$ is also real analytic for each $a\in A$, $b\in\mathrm{Supp}(\beta)$. Under the bounded state variable assumption (Assumption \ref{As:inf_bl}(vi)), this analyticity extends to $\mathbb{R}^{k}$, as shown formally in Lemma \ref{lem:inf_ccp}. This allows us to use a straightforward extension of \textcite{stinchcombe1998consistent} Theorem 3.8 (formalized in Lemma \ref{lem:SW3.8}\footnote{A heuristic justification of Lemma \ref{lem:SW3.8} is as follows: if two mixture distributions generate the same observed moment function $g(x)\equiv E[Y\mid X=x]$ on any small open set and $x\mapsto g(x)$ is real analytic, then they would also yield the same observed moment function on the full Euclidean space (assuming the relevant objects are well defined). Thus, for identification purposes, observing a non-empty open set is as informative as observing the Euclidean space. The idea is related to the properties of neural networks with limited weights, e.g., \textcite{stinchcombe1999neural} Theorem 2.3 and references therein.}) to characterize the injectivity condition in Theorem \ref{thm:inject} as 
\begin{equation}\label{eq:expl_super}
\left(\int P(a,x,b)\,d\mu(b) = 0 \quad \text{for all } (a,x)\in A\times\mathbb{R}^k\right)\; \Longrightarrow\; \mu = 0.
\end{equation}
Relative to the injectivity condition in Theorem \ref{thm:inject}, equation \eqref{eq:expl_super} may be easier to verify since $\mathbb{R}^k\supset\mathcal{X}$.

For the second part, I show in Lemma \ref{lem:inf_ccp} that conditions are satisfied to apply an equivalence result from \textcite{stinchcombe1998consistent} that characterizes condition \eqref{eq:expl_super} in terms of the approximation properties of the set of functions $\{ b\mapsto P(a,x,b):(a,x)\in A \times\mathbb{R}^k\}$. Specifically, that this set is dense in square integrable functions on $\mathrm{Supp}(\beta)$. For intuition of this characterization, consider that in the case that $\beta$ has $R<\infty$ support points, the full (row) rank condition is that the collection of vectors $\{\left(P(a,x,b):b=1,\ldots,R\right):(a,x)\in A \times\mathbb{R}^k\}$ span $\mathbb{R}^{R}$. 

The final step of the proof is to show this property, as summarized in Lemma \ref{lemma_cosapprox}: 
\begin{lemma}[Approximation]\label{lemma_cosapprox}
Under \ref{As:inf_bl} and \ref{As:inf_puh}, the linear span of
\begin{equation*}
\{\, b \mapsto P(a,x,b) : (a,x) \in A \times \mathbb{R}^k \,\}
\end{equation*}
is dense in $\mathcal{L}^2(\mathrm{Supp}(\beta))$, the space of square-integrable functions on $\mathrm{Supp}(\beta)$.
\end{lemma}
To prove Lemma \ref{lemma_cosapprox}, I adapt methods from the classical neural network literature \parencite{hornik1989multilayer,hornik1993some}. Like \textcite{hornik1989multilayer}, the argument is constructive: for a given target function on $\mathrm{Supp}(\beta)$, I find a linear combination of $b\mapsto P(a,x,b)$ that approximates it arbitrarily well. The key part of the construction is to show that for a particular choice of $x\in\mathbb{R}^k$ and $a=0$, $P(a,x,b)$ can approximate the
product of one-dimensional step functions in each component of $b\in\mathbb{R}^J$ (i.e., $\prod_{a=1}^J 1\{b_a > l_a\}$ for $l_1,l_2,\ldots,l_J$). 
It is in this part that the functional form of $u_t$ (Assumption \ref{As:inf_puh}(i)), rank condition on $\gamma$ (Assumption \ref{As:inf_puh}(iv)) and extreme value type I assumption (Assumption \ref{As:inf_bl}(v)) play key roles -- they enable a theoretical guarantee that variation in $x$ can be used to create the step and shift its location in the $\beta$ space. More concretely, Assumption \ref{As:inf_puh}(iv) guarantees that the image of $\Gamma\equiv(\gamma_1\gamma_2\ldots\gamma_J)$ is $\mathbb{R}^{\dim(\beta)}$, and Assumptions \ref{As:inf_puh}(i) and \ref{As:inf_bl}(v) guarantee the linear structure is relevant. A formal proof is in Section \ref{ssec:lemma_cosapprox}.

\subsection{Identification}

To invoke Theorem \ref{thm:inject} for identification of the DDC model, we require the support of the state variable to contain an open set:

\begin{manualassumption}{I3}\label{As:inf_os} For all $x\in\mathrm{Supp}(X_1)$, $\exists ~a\in A$ such that: (i) $\mathrm{Supp}(X_{2}\mid{X}_1=x,A_1=a)$ and $\mathrm{Supp}(X_{3}\mid{X_2}\in{\mathrm{Supp}(X_2\mid X_{1}=x,A_1 =a)},A_2=0)$ contain a non-empty open set; (ii) $S_3\equiv{\mathrm{Supp}}(X_3\mid{X_2}\in{\mathrm{Supp}(X_2\mid X_{1}=x,A_1 =a)},A_2=0)$ and $\cap_{a_3\in{\mathrm{Supp}}(A_3)}\mathrm{Supp}(X_4\mid{X_3}\in{S_3},A_3=a_3)$ span $\mathbb{R}^k$.
\end{manualassumption}

Assumption \ref{As:inf_os} places restrictions on the support of the observed state variable $X_t\in\mathbb{R}^k$. Part (i) requires that the support of the observed state variable contains an open set. Part (ii) requires that the supports contain $k$ linearly independent elements, a mild rank condition which is standard in linear models. As discussed in Example \ref{ex:renewal}, Assumption \ref{As:inf_os} allows for renewal models like \textcite{rust1987optimal}. However, it rules out lagged dependent variables, that is, when $X_t$ contains the lagged choice $A_{t-1}$. This would rule out, for example, a firm entry problem where the current period's entry decision $A_{t}$ depends on whether the firm is currently active ($A_{t-1}$). In particular, lagged dependent variables contradict Assumption \ref{As:inf_os}(ii) since $\mathrm{Supp}(X_4\mid{X_3}=x,A_3=a)$  and $\mathrm{Supp}(X_4\mid{X_3}=x,A_3=\tilde{a})$ are disjoint for $a\neq\tilde{a}$.\footnote{\label{fn:os} Although the open set assumption \ref{As:inf_os}(i) also rules out purely discrete variables, as discussed in Remark \ref{rem:discrete}, these can be allowed with minor notational changes. In this case, Assumption \ref{As:inf_os}(i) is relaxed but Assumption \ref{As:inf_os}(ii) is unchanged. See Section \ref{sec:discrete_state} for a technical statement.} However, unlike some results in the literature, Assumption \ref{As:inf_os} does not require that the support be `rectangular'\footnote{For example, this is Assumption 1(c)-(e) used in  \textcite[]{kasahara2009nonparametric} Propositions 1-9 and subsequently relaxed in Propositions 10 and 11.} --- which requires that, starting from any sequence of choices and past state variables, any state can be reached (i.e., for all  $t$ and $(a,x)\in \mathrm{Supp}(A_t,X_t)$, $\mathrm{Supp}(X_{t+1}|X_t=x,A_t=a)=\mathrm{Supp}(X_{t+1})=\mathrm{Supp}(X_1)$).

\begin{example}[Renewal model]\label{ex:renewal}
Consider a bivariate state variable $X_t\in\mathbb{R}^k$, where action $A_t=0$ `regenerates' the state variable to its baseline as in \textcite[p. 1006]{rust1987optimal}. As in \textcite[Section 6.1]{kristensen2019solving}, the transition kernel may be a mixture of a point mass and a continuous random variable:
\begin{align*}
F_x\left(x_{t+1} ; x_t, a_t \right)&= \pi1\left(x_{t+1}=a_t x_t\right)+(1-\pi)F_+(x_{t+1};x_t,a_t),
\end{align*}
for $\pi\in[0,1]$ and where $F_+(x';x,a)$ has support $\mathrm{Supp}(X_{t+1}|X_t=x_t,A_t=a_t)=\times_{k'=1}^{k}[a_t x_{tk'},K_{k'}]$. When $\pi<1$, $\mathrm{Supp}(X_{t+1}|X_t=x,A_t=0)=\times_{k'=1}^k[0,K_{k'}]$ so Assumptions \ref{As:inf_os}(i)is satisfied. It follows that  \ref{As:inf_os}(ii) is satisified with $\cap_{a_3\in{\mathrm{Supp}}(A_3)}\mathrm{Supp}(X_4\mid{X_3}\in{S_3},A_3=a_3) = \times_{k'=1}^k[0,K_{k'}]$.
\end{example}

The model parameters are ($F_{x},\gamma,f_{\beta|X_{1}}$): the state transition, the homogeneous payoff parameter, and the conditional distribution of permanent unobserved heterogeneity. As the state transition is identified by direct observation, the following result handles the remaining parameters:

\begin{theorem}[Identification]\label{thm:inf}
Assume the distribution of $(X_{t},A_{t})_{t=1}^{T}$ is observed for $T\ge4$, generated from agents solving the model of equation (\ref{eq:inf_prob}) satisfying assumptions \ref{As:inf_bl}-\ref{As:inf_os}. Then $(\gamma,f_{\beta|X_{1}})$ is point identified.
\end{theorem}

Theorem \ref{thm:inf} is established via a decomposition argument \parencite{hu2008instrumental,freyberger2018non}. The model structure imposed by Assumptions \ref{As:inf_bl} and \ref{As:inf_puh} implies the following `factorization equation' representation of the weighted distribution of $(X_{t},A_{t})_{t=1}^{T}$ \parencite{kasahara2009nonparametric}:
\begin{multline*}
\frac{f_{A_{4}A_{3}A_{2}A_{1}X_{4}X_{3}X_{2}|X_{1}}(a_{4},a_{3},a_2,a_{1},x_{4},x_{3},x_{2},x_{1})}{F_{x}(x_{4}|x_{3},a_{3})F_{x}(x_{3}|x_{2},a_2)F_{x}(x_{2}|x_{1},a_{1})}\\=\int{P}(a_{4},x_{4},b)P(a_{3},x_{3},b)P(a_2,x_{2},b)P(a_{1},x_{1},b)dF_{\beta|X_{1}}(b,x_{1}),
\end{multline*}
which is guaranteed to exist under the support condition in Assumption \ref{As:inf_os}(i). In the factorization equation we see the role of Assumption \ref{As:inf_puh}(iii): since the state transition does not depend on $\beta$, it can be passed through the integral.\footnote{Related homogeneity assumptions can also lead to weighting approaches in other models, such as \textcites[Chapter 21]{hernan2020chapman}[Section 6]{bonhomme2023identification}.} Then, by invoking the injectivity result in Theorem \ref{thm:inject}, the representation can be used to express the CCP function $P(a,x,b)$ as the eigenfunction of a particular eigendecomposition.\footnote{This reasoning also suggests that, by directly assuming the injectivity condition in Theorem 1, a related identificaton result may hold under weaker conditions on the model (i.e., weaker versions of Assumptions \ref{As:inf_bl} and \ref{As:inf_puh}). See \textcite{kasahara2009nonparametric}, Remark 2.} I then show that the eigendecomposition is unique, which delivers identification of $\gamma \in \mathbb{R}^k$: The argument is related to identification of dynamic discrete choice models without unobserved heterogeneity (e.g., \textcite{bajari2015identification}), with Assumption \ref{As:inf_os}(ii) playing a central role.  Knowledge of $\gamma$ is then used in combination with the factorization equation and injectivity to identify $f_{\beta|X_1}$. The formal proof is in Section \ref{pf:id}.

\begin{remark}[\textbf{Panel length}]
Theorem \ref{thm:inf} requires at least four observations per individual. In contrast \textcite{kasahara2009nonparametric} require only $T=3$. With three periods, identification of the model in Theorem \ref{thm:inf} is possible under a high-level assumption on the joint distribution of permanent unobserved heterogeneity and the first period state variable.\footnote{For example, \textcite[Proposition 1]{kasahara2009nonparametric} assumes that  for some $x\in\mathrm{Supp}(X_1)$, $\Pr(A_1=1,X_1=x,\beta=b)=\Pr(A_1=1|X_1=x,\beta=b)\Pr(\beta=b|X_1=x)\Pr(X_1=x)>0$ is injective in $b$.} However, the advantage of $T=4$ is to avoid this type of high level condition on the distribution of $(X_1,\beta)$, instead using low level conditions on the choice model.
\end{remark}

\section{Extensions}\label{sec:extensions}

In this section I provide identification results for a number of variations on the model in Section \ref{sec:model}. Sections \ref{sec:model-intercept} and \ref{sec:no_term} consider finite-horizon environments in which the agent's decision rule may vary across periods. Section \ref{sec:model-intercept} focuses on the case where the terminal period is observed, allowing identification of models with random intercepts. Section \ref{sec:no_term} addresses the case where the decision horizon extends beyond the observed sample. It provides two solutions: imposing out-of-sample restrictions or exploiting finite dependence. Section \ref{sec:inf_fe} returns to the infinite-horizon setting and allows for random intercepts under additional assumptions on the transition process. Finally, Section \ref{sec:kwon} shows that the number of agent types is identified in models with discrete unobserved heterogeneity.

\subsection{Non-stationary conditional choice probabilities}\label{sec:model-intercept}

In many contexts, the agent's decision rule may change between periods: for example, if the agent has a finite time-horizon, or if the state variables are subject to structural breaks. In these cases, it is natural to allow the per-period utility function and state transitions to be non-stationary, i.e., to be time-dependent. In this section I consider a finite horizon dynamic discrete choice model in which the terminal decision period is observed. For example, in a model of retirement from the labor force \parencite{rust1997social}, we may eventually observe all individuals retire. Similarly, in a model of educational attainment, we may observe all individuals reach a terminal state \parencite{heckman2018returns}. By definition, the decision-maker has no strategic influence over future utility flows to consider in the terminal period and thus a different proof strategy is adopted. This argument allows for identification of random intercepts, which was not the case in Section \ref{sec:model}.

I begin by adapting Assumptions \ref{As:inf_bl} and \ref{As:inf_puh} to the non-stationary context. In particular, by allowing the flow utility and state transition to be time-dependent.
\begin{manualassumption}{F1}\label{As:f_bl} (i) Assumptions \ref{As:inf_bl} (ii), (iv), (v) and (vi) hold.
(ii) $u_t(S_{t},\epsilon_{t},A_{t})=u_t(S_{t},A_{t})+\sum_{a\in A}\epsilon_{t,a}1[a=A_t]$. (iii) $d\Pr(S_{t+1}=s',\epsilon_{t+1}=e'\mid S_t=s,\epsilon_t=e,A_t=a)=dF_\epsilon(e')\times dF_{s_t}(s'\mid a,s)$.
\end{manualassumption}

\begin{manualassumption}{F2}\label{As:f_puh}
\begin{enumerate*}[label=(\roman*)]
    \item $S_t=(X_t^\intercal,\beta^\intercal)^\intercal\in\mathbb{R}^{k+(1+p)J}$, and $k=p+J\text{ for } p\ge0$. For each $x\in\mathrm{Supp}(X_1)$, $\beta\mid X_1=x$ admits a bounded density $f_{\beta|X_{1}}$. \item For $\gamma_{t,a}\in\mathbb{R}^{k-p}$, $u_t(s,a)=\beta_{a[1]}+x^\intercal\left(\beta_{a[-1]}^{\intercal},\gamma_{t,a}^{\intercal}\right)^{\intercal}$ where $\beta_{a}=(\beta_{a[1]},\beta_{a[-1]}^\intercal)^\intercal\in\mathbb{R}^{1+p}$. \item $d\Pr(X_{t+1}=x_{t+1} \mid A_t=a_t, X_t=x_t, \beta=b )= dF_{x_t}(x_{t+1}\mid x_t,a_t)$. \item\label{As:f_homog} $\Gamma_{T}\equiv(\gamma_{T,1}\gamma_{T,2}\cdots\gamma_{T,J})\in\mathbb{R}^{J\times J}$ is full rank.
\end{enumerate*} 
\end{manualassumption}

Assumption \ref{As:f_puh} states that permanent unobserved heterogeneity enters the model as a state variable. The restrictions are weaker than those in the infinite horizon model (Assumption \ref{As:inf_puh}). First, the permanent unobserved heterogeneity can include a random intercept. Second, there may be multiple random coefficients for each option, whereas in Section \ref{sec:model} the model was limited one action-specific random coefficient (i.e., $p=1$). This relaxation is possible due to the relatively simple structure of the terminal period CCP function. As was the case for the infinite horizon model, the support of permanent unobserved heterogeneity may be finite, but it need not be (see footnote \ref{fn:pmf}). Like Assumption I2\ref{As:inf_homog}, Assumption F2\ref{As:f_homog} imposes that the state variable cannot affect payoffs for each choice in a similar fashion. Since identification is attained from the terminal period, we place weaker restrictions on the transition $F_{x_t}$ relative to Assumption I2\ref{As:inf_tran}.

To describe the injectivity result for the finite horizon model, denote the CCP function $P_t(a,x,b)$ and let $T$ denote the decision horizon of the agent.
\begin{theorem}[Injectivity]\label{thm:inject_finite}
Assume \ref{As:f_bl} and \ref{As:f_puh}. Let $\mathcal{X}\subseteq\mathrm{Supp}(X_T)$ be a non-empty open set and let $\mu$ be a finite signed measure on $\mathrm{Supp}(\beta)$. If
    \begin{equation*}
        \int P_T(a,x,b)d\mu(b)=0 \quad \text{for almost every } (a,x)\in A\times \mathcal{X},
    \end{equation*}
    then $\mu=0$, the zero measure.
\end{theorem}

The proof of Theorem \ref{thm:inject_finite} is contained in Section \ref{sec:f_proofs}. The proof logic is rather different to Theorem \ref{thm:inject}: to show Theorem \ref{thm:inject_finite}, I show the implication directly by demonstrating that $\int P_T(a,x,b)d\mu(b)=0$ implies that the induced measure of $P_T(a,x,\beta)$ is zero. The linear utility function and distributional assumption on $F_\epsilon$ are particularly useful for this. The result then follows from \textcite{masten2018random}, Lemma 1.

As for the time stationary model, we require further restrictions on the state variable $X_t$ for identification of the DDC model. First, Assumption \ref{As:f_os} requires there be some continuous variation in $X_T$ after conditioning upon each history of actions and state variables.

\begin{manualassumption}{F3}\label{As:f_os}
For each $x_{1}\in\mathrm{Supp}(X_1)$ and ($a_{1},a_2,\ldots,a_{T-1})\in A^{T-1}$, there is $(x_{2},x_3,\ldots,x_{T-1})\in \times_{t=2}^{T-1} \mathrm{Supp}(X_t)$ such that
\begin{equation*}
\mathrm{Supp}\left(X_{T} \mid A_{T-1}=a_{t-1},X_{T-1}=x_{t-1},\dots,A_1=a_1,X_1=x_1\right)
\end{equation*}
contains a non-empty open set. Moreover, for each $t$, $\mathrm{Supp}((1,X_t))$ spans $\mathbb{R}^{k+1}$.
\end{manualassumption}

To introduce the final assumption, let  $\gamma_t=\left\{\gamma_{t,a}:a=1,\ldots,J\right\}$ and define $S_T\equiv\mathrm{Supp}\left(X_{T} \mid A_{T-1}=a_{t-1},X_{T-1}=x_{t-1},\dots,A_1=a_1,X_1=x_1\right)$ and let $E\subset S_{T}\times {A}$, $P_T(a;x,b,\gamma)$ be the model implied probability of $A_T=a$ conditional upon $X_T=x$ evaluated at $\beta=b$ and $\gamma_T=\gamma$, and $\mathcal{L}_\mathcal{A}$ be the set of bounded functions on $\mathcal{A}$. Then define the operator
\begin{equation*}
L^{E,\gamma}_{T,\beta}:\mathcal{L}_{\mathrm{Supp}(\beta)}\rightarrow\mathcal{L}_{E}\qquad[L^{E,\gamma}_{T,\beta}m](x,a)=\int P_T(a;x,b,\gamma)m(b)db.
\end{equation*} 
Denote $(L^{E,\gamma}_{T,\beta})^{-1}$ as the left inverse of $L^{E,\gamma}_{T,\beta}$.

\begin{manualassumption}{F4}\label{As:f_rank}
For every $\gamma\neq\tilde{\gamma}$, there exists $E,\tilde{E}\subseteq{S}_{T}\times{A}$ containing non-empty open sets such that the operator defined in equation \eqref{eq:fpfh_op} is injective.
\begin{equation}\label{eq:fpfh_op}
L^{E,\gamma,\tilde{E},\tilde{\gamma}}_{T,\beta}:\mathcal{L}_{\mathrm{Supp}(\beta)}\rightarrow\mathcal{L}_{\mathrm{Supp}(\beta)}\qquad[L^{E,\gamma,\tilde{E},\tilde{\gamma}}_{T,\beta}m](b)=\left[\left((L^{E,\gamma}_{T,\beta})^{-1}L^{E,\tilde\gamma}_{T,\beta}-(L^{\tilde{E},\gamma}_{T,\beta})^{-1}L^{\tilde{E},\tilde\gamma}_{T,\beta}\right)m\right](b).
\end{equation}
\end{manualassumption}

This high-level condition ensures that the parameter $\gamma_{T}$ can be identified without knowledge of the distribution of unobserved heterogeneity. A few comments on Assumption \ref{As:f_rank} are in order. First, given Theorem \ref{thm:inject_finite}, Assumptions \ref{As:f_bl}-\ref{As:f_os} imply that, for any $E$ containing a non-empty open set, $L^{E,\gamma}_{T,\beta}$ is injective so that $L^{E,\gamma,\tilde{E},\tilde{\gamma}}_{T,\beta}$ exists.  Second, the condition is stated in terms of observed objects, and thus the operator defined in Assumption \ref{As:f_rank} is identified by direct observation. Third, should Assumption \ref{As:f_rank} not hold, I show in an appendix (Lemma \ref{thm:f_norank}) that under Assumptions \ref{As:f_bl}-\ref{As:f_os} and a scale restriction on $\gamma_{T}$, that $\gamma_{T}$ and the distribution of unobserved heterogeneity are identified. 

Finally, the condition can be related to the high-level necessary conditions for identification of a common parameter in discrete choice panel data given in \textcite{johnson2004identification,chamberlain2010binary}. To describe their result, fix $x\equiv(x_{1},x_{2},\dots,x_{T})$ and for convenience let $A=\{0,1\}$ and $\gamma$ be time-invariant. Let $p(b;x,\gamma)$ be the length $2^{T}$ vector of choice probabilities $\left\{\prod_{t=1}^{T}P_{t}(a_{t},x_{t},b;\gamma):(a_{t})_{t=1}^{T}\in \{0,1\}^{T}\setminus\{0_{T}\}\right\}$ in the $(2^{T}-1)$-dimensional hypercube.  \textcite[Theorem 2.2]{johnson2004identification} states that the common parameter $\gamma$ will not be identified if the set $\{p(b;x,\gamma):b\in{S}_{\beta}\}$ does not lie in a hyperplane for some $x$. For the \textit{static} binary choice model with $T=2$, \textcite{chamberlain2010binary} shows that the hyperplane restriction is satisfied if and only if the unobserved state variables are i.i.d. extreme-value type I. Given the remarkable result of \textcite{chamberlain2010binary}, one may conjecture that the $T=2$ \textit{dynamic} binary choice model does not satisfy \textcite{johnson2004identification}'s condition and therefore $\gamma$ is not identified. If this is the case, then $\forall{x_{2}}\in\mathrm{Supp}{(X_{2})}$ and $\gamma\neq\tilde{\gamma}$, there exist some $f_{\beta|X_{1}X_{2}}\neq\tilde{f}_{\beta|X_{1}X_{2}}$ such that
\begin{equation*}
\left[L^{\mathrm{Supp}{(X_{2})},\gamma}_{2,\beta}f_{\beta|X_{1}X_{2}}(\cdot,x_{1},x_{2})\right](x_{2})=\left[L^{\mathrm{Supp}{(X_{2})},\tilde\gamma}_{2,\beta}\tilde{f}_{\beta|X_{1}X_{2}}(\cdot,x_{1},x_{2})\right](x_{2}),
\end{equation*}
where the distribution of unobserved heterogeneity $f_{\beta|X_{1}X_{2}}$ is allowed to depend on $x_{2}$ as in \textcite{johnson2004identification,chamberlain2010binary}. If the distribution is restricted to be the same for all $x_{2}\in\mathrm{Supp}{(X_{2})}$, the above condition implies that for each $\gamma\neq\tilde{\gamma}$, ${x_{2}}\in\mathrm{Supp}{(X_{2})}$, then there are some $f_{\beta|X_{1}},\tilde{f}_{\beta|X_{1}}$ that satisfy 
\begin{equation*}
\left[L^{\mathrm{Supp}{(X_{2})},\gamma}_{2,\beta}f_{\beta|X_{1}}(\cdot,x_{1})\right](x_{2})=\left[L^{\mathrm{Supp}{(X_{2})},\tilde\gamma}_{2,\beta}\tilde{f}_{\beta|X_{1}}(\cdot,x_{1})\right](x_{2}).
\end{equation*}
However, since the distribution of unobserved heterogeneity is required to be the same for all $x_{2}$, there may be some other $\tilde{x}_{2}\in\mathrm{Supp}{(X_{2})}$ such that 
\begin{equation*}
\left[L^{\mathrm{Supp}{(X_{2})},\gamma}_{2,\beta}f_{\beta|X_{1}}(\cdot,x_{1})\right](\tilde{x}_{2})\neq\left[L^{\mathrm{Supp}{(X_{2})},\tilde\gamma}_{2,\beta}\tilde{f}_{\beta|X_{1}}(\cdot,x_{1})\right](\tilde{x}_{2}).
\end{equation*}
Let $E,\tilde{E}$ be neighborhoods of $(x_{2},\tilde{x}_{2})$, respectively. In the proof to Theorem \ref{thm:f} it is shown that, without knowledge of $f_{\beta|X_{1}}$ or $\tilde{f}_{\beta|X_{1}}$, there does exist such an $\tilde{x}_{2}$ if the operator defined in equation (\ref{eq:fpfh_op}) is injective. This can be viewed as a partial converse to \textcite{johnson2004identification}'s high-level condition: in that case, without knowledge of $f_{\beta|X_{1}}$ or $\tilde{f}_{\beta|X_{1}}$, one can show there does \textit{not} exist such an $\tilde{x}_{2}$ if their `rank' condition does not apply. In principle, the logic of Assumption \ref{As:f_rank} can be extended to the general discrete choice panel model of \textcite{johnson2004identification}, if the distribution of unobserved heterogeneity is required to be independent of covariates. To state the theorem denote $\gamma=\{\gamma_t:t=1,\ldots,T\}$.

\begin{theorem}[Identification]\label{thm:f}
Assume the distribution of $(X_{t},A_{t})_{t=1}^{T}$ is observed for $T\ge2$, generated from agents solving the model of equation (\ref{eq:inf_prob}) satisfying assumptions \ref{As:f_bl}-\ref{As:f_rank}. Then $(\gamma,f_{\beta|X_{1}})$ is point identified.
\end{theorem}
Section \ref{sec:f_proofs} contains the proof of Theorem \ref{thm:f}.

\subsection{Non-stationary conditional choice probabilities without the terminal period}\label{sec:no_term}

In many empirical settings, the decision horizon of the agent extends beyond the period of observation. For example, a worker's labor force participation decisions may not be observed for their entire working life. This poses an issue for identification since in-sample decisions reflect payoff parameters for both in- and out-of-sample time periods. This section provides two solutions for this issue. The first approach is to impose restrictions on out-of-sample payoffs. Section \ref{sec:f_nt} adopts this approach and shows that the model without random intercepts is identified. 

The second approach is to use a property of the state transition known as `finite dependence', which occurs if multiple sequences of actions leads to the same distribution of the state variable \parencite{arcidiacono2011practical}. Finite dependence limits the number of out-of-sample time periods that affect in-sample decisions. Section \ref{sec:fh_fd} considers a model that exhibits finite dependence, and shows a binary choice model with random coefficients is identified.

For both approaches, I consider a model that satisfies the following condition:
\begin{manualassumption}{F2${}^\prime$}\label{As:f_nt_puh}
(i) Assumptions \ref{As:inf_puh}(i) and (v) hold. (ii) For each $t$, ${u_t(s,a)=x^\intercal\left(\beta_{a},\;\gamma_{t,a}^\intercal\right)^\intercal},$ for $\gamma_{a,t}\in\mathbb{R}^{J}$. 
(iii) $d\Pr(X_{t+1}=x' \mid A_t=a, X_t=x, \beta=b )= dF_{x_t}(x_{t+1}\mid x_t,a_t)$.
(iv)
$\Gamma_{t}\equiv(\gamma_{t,1}\gamma_{t,2}\cdots\gamma_{t,J})\in\mathbb{R}^{J\times J}$ is full rank.
\end{manualassumption}
Analagously to the Sections \ref{sec:model} and \ref{sec:model-intercept}, Assumptions \ref{As:f_bl} and \ref{As:f_nt_puh} are sufficient for injectivity of the integral operator with kernel function $P_t(a,x,b)$.

\subsubsection{Out of sample restrictions}\label{sec:f_nt}

Let $T$ denote the final observed period and $T_1>T$ denote the final decision period of the agent. Since we do not observe behavior
in periods $(T+1,\dots,T_{1})$, the following restriction is placed on out-of-sample behavior:

\begin{manualassumption}{F5}\label{As:f_nt_oos}
For all $t\in(T+1,\dots,T_{1})$, $\gamma_{t}=\gamma_{T}$ and $dF_{x_{t-1}}(x'|x,a)=dF_{x_{T-1}}(x'|x,a)$.
\end{manualassumption}

With these assumptions and a support condition on $X_{t}$ related to Assumption \ref{As:inf_os}, identification results follows as a Corollary of Theorem
\ref{thm:inf}. The proof is found in Section
\ref{ssec:f_proofs_nt}.

\begin{corollary}\label{thm:f_nt}
Assume the distribution of $(X_{t},A_{t})_{t=1}^{T}$ is observed for $T=4$, generated from agents solving the model of equation \eqref{eq:inf_prob} satisfying Assumptions \ref{As:f_bl}, \ref{As:f_nt_puh}, \ref{As:f_nt_os} and \ref{As:f_nt_oos}. Then $(\gamma,f_{\beta|X_{1}})$ is point identified.
\end{corollary}

\subsubsection{Finite dependence}\label{sec:fh_fd}

A DDC model exhibits finite dependence if there are multiple sequences of actions that yield the same distribution over the state variable. Finite dependence is useful for estimation as it allows the continuation value to be expressed in terms of CCPs \parencite{arcidiacono2011practical}. This fact also makes finite dependence useful for identification in models without permanent unobserved heterogeneity, as it reduces the number of periods of out-of-sample behavior that must be assumed known \parencite[Section 3.3]{arcidiacono2020identifying}. 

In this section I show a similar feature is present for models with continuous permanent unobserved heterogeneity. In particular, I assume the transition function exhibits a special case of finite-dependence: the renewal action. The canonical example of renewal is machine replacement, but models of turnover and job matching also display this pattern \parencite{arcidiacono2020identifying}. This idea is formalized in the next assumption, which, in addition to a support condition, is sufficent for identification.

\begin{manualassumption}{F6}\label{ass:fh_fd}
For each $t$, $\exists\; a\in \mathrm{Supp}(A_{t})$ such that $dF_{x_t}(x'|x,a)=dF_{x_t}(x'|\tilde{x},a)$ for all $x'$ and $x,\tilde{x}\in\mathrm{Supp}(X_t)$.
\end{manualassumption}

\begin{corollary}\label{thm:fh_fd}
Assume the distribution of $(X_{t},A_{t})_{t=1}^{4}$ is observed, generated from agents solving the model of equation (\ref{eq:inf_prob}) with $J=1$ and satisfying assumptions \ref{As:f_bl}, \ref{As:f_nt_puh}, \ref{As:f_nt_os_fd}, and \ref{ass:fh_fd}. Then $(\gamma,f_{\beta|X_{1}})$ is point identified.
\end{corollary}
Section \ref{sec:pf_fd} contains the proof to Corollary \ref{thm:fh_fd}, whose substance is adapted from the proof of Theorem \ref{thm:inf}.

\subsection{Random intercepts in a stationary model}\label{sec:inf_fe} 

This section considers identification of an infinite-horizon DDC model with random intercepts. It shows point identification can be attained under an additional restriction on the state transition. Specifically, there must be some point in the support of $X_{t}$ for which the state transition is not choice dependent. For instance, the machine replacement model of \textcite[Example 9]{kasahara2009nonparametric} displays this property. Before introducing the restriction on the state transition, the next assumption states that the permanent unobserved heterogeneity enters the model as a random intercept:

\begin{manualassumption}{I2${}^\prime$}\label{As:inf_puh_fe}
\begin{enumerate*}[label=(\roman*)]
\item Assumptions \ref{As:inf_puh} (i), (iii) and (iv) hold.
\item $u(s,a)=\beta_a + x^\intercal\gamma_{a}.$
\end{enumerate*} 
\end{manualassumption}

The next assumption strengthens Assumption \ref{As:inf_os} by requiring the state transition to be constant across choices:

\begin{manualassumption}{I3${}^\prime$}\label{As:inf_os_fe} For all $x_1\in\mathrm{Supp}(X_1)$, $\exists ~a_1\in\mathrm{Supp}(A_1)$ such that: (i) $\mathrm{Supp}(X_{2}\mid{X}_1=x_1,A_1=a_1)$ and $\mathrm{Supp}(X_{3}\mid{X_2}\in{\mathrm{Supp}(X_2\mid X_{1}=x_1,A_1 =a_1)},A_2=0)$ contain non-empty open sets for which all elements $x$ satisfy $dF_x(x'\mid \tilde{a},x)=dF_x(x'\mid a, x)$ for all $x'$, $a$ and $\tilde{a}$; (ii) $S_3\equiv{\mathrm{Supp}}\left((1,X_3)\mid{X_2}\in{\mathrm{Supp}(X_2\mid X_{1}=x_1,A_1 =a_1)},A_2=0\right)$ and $\cap_{a_3\in{\mathrm{Supp}}(A_3)}\mathrm{Supp}\left((1,X_4)\mid{X_3}\in{S_3},A_3=a_3\right)$ span $\mathbb{R}^{k+1}$.
\end{manualassumption}

\begin{corollary}\label{thm:inf_fe}
Assume the distribution of $(X_{t},A_{t})_{t=1}^{T}$ is observed for $T\ge4$, generated from agents solving the model of equation (\ref{eq:inf_prob}) satisfying assumptions \ref{As:inf_bl}, \ref{As:inf_puh_fe} and \ref{As:inf_os_fe}. Then $(\gamma,f_{\beta|X_{1}})$ is point identified.
\end{corollary}

The proof to Corollary \ref{thm:inf_fe} is contained in Section \ref{pf:inf_fe}. It follows from the proofs of Theorems \ref{thm:inf} and \ref{thm:inject_finite}.

\subsection{Identifying the number of mixture components}\label{sec:kwon}

In the existing DDC literature, it is common to assume permanent unobserved heterogeneity is discrete. When this assumption is made, a key parameter is the number of support points of permanent unobserved heterogeneity. In practice, it is common to assume the number of support points is known, although there are methods to identify a lower bound on the number of support points \parencite{kasahara2009nonparametric,kasahara2014non,kwon2019estimation} which have been applied in economics \parencite{igami2016unobserved}. However, in general, these methods can only identify the number of support points if an upper bound is known. This is because there is no guarantee \textit{a priori} that there is enough variation in the data and structure on the model to to identify any arbitrarily large number of types. Intuitively, the population likelihood may be flat as a mixture component is added, but this may be because the initial likelihood had the true number of mixture components \textit{or} because the models with and without an additional mixture component are observationally equivalent. Technically, this issue can be resolved by imposing an injectivity condition, i.e., a rank assumption on an unobserved matrix \parencites[Proposition 3]{kasahara2009nonparametric}[Assumption 2.1]{kwon2019estimation}.

The purpose of this section is to show the models of Theorem \ref{thm:inf} and Corollary \ref{thm:f_nt} satisfy a condition equivalent to \textcite[Assumption 2.1]{kwon2019estimation} when the distribution of unobserved heterogeneity is discrete. This means the number of types is identified, without knowledge of an upper bound on the number of types.

\begin{corollary}\label{thm:kwon}
Assume the distribution of $Y=(X_{t},A_{t})_{t=1}^{T}$ is observed for $T\ge3$, generated from the DDC model satisfying either Assumptions \ref{As:inf_bl}-\ref{As:inf_os} or Assumptions \ref{As:f_bl}, \ref{As:f_nt_puh}, \ref{As:f_nt_os} and \ref{As:f_nt_oos}. In addition, suppose that the support of $\beta|X_1$ has $R<\infty$ points of support. Then, for any fixed $x_1\in \mathrm{Supp}(X_1)$, $R$ is identified as the rank (defined as the dimension of the range) of the operator
\begin{equation*}
[Lu](x_{3})= \int{u}(x_{2})\frac{f_{A_{3}A_{2}A_{1}X_{3}X_{2}|X_{1}}(0,0,0,x_{3},x_{2},x_{1})}{F_{x_{3}}(x_{3}|x_{2},0)F_{x_{2}}(x_{2}|x_{1},0)}dx_{2}.
\end{equation*}
\end{corollary}

The proof to Corollary \ref{thm:kwon} is found in Section \ref{pf:kwon}. The result means that the techniques of \textcite{kasahara2014non,kwon2019estimation} can be used to consistently estimate the number of types should the applied econometrician wish to maintain the standard assumption that permanent unobserved heterogeneity is discrete.\footnote{The model in Corollary \ref{thm:kwon} can be directly adapted to the general frameworks of \textcite{kasahara2014non,kwon2019estimation}. See, in particular, \textcite{kwon2019estimation} Equation 2.1 and \textcite{kasahara2014non} Equation 2.} These techniques also give rise to valid hypothesis tests regarding the number of types, including testing the null of type degeneracy (that is, $R=1$). Broadly speaking, these estimators consist of forming a matrix of observed choice probabilities with values of $X_{3}$ varying over the rows, and $X_{2}$ over the columns. Corollary \ref{thm:kwon} means that, at the population level, the rank of the matrix equals the true number of types.

\section{Estimation}\label{sec:estimation}

This section considers consistent estimation of the model parameters in a short panel. The distribution of $Y\equiv(A_{t},X_{t})_{t=1}^{T}$ can be written as
\begin{equation*}
\int\prod_{t=2}^{T}\left(P_{t}(a_{t},x_{t},b;\gamma,F_{x})F_{x_{t}}(x_{t}|x_{t-1},a_{t-1})\right)P_{1}(a_{1},x_{1},b;\gamma,F_{x})F_{x_{1}}(x_{1})dF_{\beta|X_{1}}(b,x_{1}),
\end{equation*}
where $F_{\beta|X_{1}}(b,x_{1})$ is the cumulative distribution function of $\beta$ conditional upon $X_{1}=x_{1}$, $F_{x_{1}}$ is the marginal distribution of $X_{1}$ and the dependence of the CCPs on $(\gamma,F_x)$ is made explicit. I propose two-step sieve M-estimation based on the above expression. The first step consists of estimating the state transitions and marginal distribution of the initial state, $F_{x}=\{F_{x_t}:t=1,\ldots,T\}$. The second step consists of forming the pseudo-likelihood function using the fact that the CCPs $P_{t}$ are known up to the state transition and payoff parameter $(F_{x},\gamma)$, and using sieve M-estimation methods to estimate ($\gamma,F_{\beta|X_{1}}$).

It is of course possible to estimate the model in a single step as a sieve maximum likelihood problem. The advantage of the proposed two-step approach is computational: by treating $F_{x}$ as fixed in the second step, computationally advantageous methods for approximating the value function may be used, such as \textcite{kristensen2019solving}.

Although I show consistency for a general sieve space (Section \ref{sssec:gss}), this may be  computationally burdensome to implement, since estimation requires computing the CCPs for every point in the support of the sieve. To circumvent this issue, I suggest a `fixed grid' estimator \parencite{heckman1984method} which reduces the computational burden by having a finite number of support points (Section \ref{sssec:foms}). Given these results, the practioner's decision to approximate $F_{\beta|X_1}$ by a continuous function or by the `fixed grid' can be viewed as a choice of tuning parameter, rather than an identifying assumption.

In this section, I focus on estimating the cumulative distribution function of $\beta$. While it would be possible to present conditions for consistent estimation of the density function, smoothness restrictions would rule out the possibility that the type distribution has discrete support, which is the standard assumption in the literature. Moreover, focusing on the distribution function of $\beta$ enables the choice of the piecewise constant sieve space described in Section \ref{sssec:foms}, which has particular computational advantages.

As a final comment, in practice there will be an approximation error in the evaluation of the CCPs. This problem is inherent to dynamic discrete processes with large state spaces, and has received significant attention in the recent literature \parencite{rust2008dynamic,kristensen2019solving}. I assume away the effect of these errors on estimation --- that is, that the approximation error is negligible relative to sampling error. In principle, the results of \textcite{kristensen2019solving} could be used to explicitly consider the effect of value function approximation error on estimation, though I do not pursue this here. Of course, the approximation error can be made arbitrarily small at increased computational cost.

\subsection{A general two-step seminonparametric estimator}\label{sssec:gss}

In this section, I briefly outline the two-step sieve M-estimator and present the general consistency result. Denote the true parameters as $\theta_{0}=(F_{x},\gamma,F_{\beta|X_{1}})\in\Theta=\mathcal{F}\times\Gamma\times\mathcal{M}$, where $\mathcal{F}$ is the space of state transitions, $\Gamma\subseteq\mathbb{R}^{\dim\gamma}$, and $\mathcal{M}$ is the space of distribution functions on $\mathrm{Supp}(\beta)$ conditional upon $x\in\mathrm{Supp}(X_1)$. The first step consists of forming a consistent estimator $\hat{F}_{x}$ for the state transition $F_{x}$. Since the state transition is directly observed, standard non-parametric methods are available. For the second step, the log-likelihood contribution of the $i$th observation is
\begin{equation*}
\psi(y_{i},\hat{F}_{x},\gamma,F_{\beta|X_{1}})\equiv\log\int\prod_{t=1}^{T}P_{t}(a_{i,t},x_{i,t},b;\hat{F}_{x},\gamma)dF_{\beta|X_{1}}(b,x_{i1}),
\end{equation*}
where $P_{t}(a,x,b;\hat{F}_{x},\gamma)$ is the model implied probability of observing choice $a$ in period $t$ conditional upon state $x$ and permanent unobserved heterogeneity $b$, evaluated at the first-step estimate $\hat{F}_{x}$ and candidate parameter $\gamma$. Given a sieve space $\mathcal{M}_{n}$, which approximates $\mathcal{M}$ arbitrarily well for large $n$, the second step estimator is defined as
\begin{equation}\label{eq:smle}
\frac{1}{n}\sum_{i=1}^{n}\psi(y_{i},\hat{F}_{x},\hat{\gamma},\hat{F}_{\beta|X_{1}})\ge\sup_{(\gamma,F)\in\Gamma\times\mathcal{M}_{n}}\frac{1}{n}\sum_{i=1}^{n}\psi(y_{i},\hat{F}_{x},\gamma,F)-o_{p}(1/n)
\end{equation}
The following result states that under standard regularity conditions, the estimator is consistent.
\begin{theorem}\label{thm:est}
Let $(A_{i,t},X_{i,t}:{t}=1,\dots,T)_{i=1}^{n}$ be i.i.d. data generated from the DDC model satisfying either Assumptions \ref{As:inf_bl}-\ref{As:inf_os} or Assumptions \ref{As:f_bl}-\ref{As:f_rank}. If Assumptions \ref{As:est_fs}-\ref{As:est_uc} hold, then the estimator $(\hat{\gamma},\hat{F}_{\beta|X_{1}})$ defined in equation \eqref{eq:smle} is consistent for $(\gamma,F_{\beta|X_{1}})$.
\end{theorem}
The full statement of Theorem \ref{thm:est} and its proof are contained in Appendix \ref{ssec:est}.

\subsection{Fixed grid estimation}\label{sssec:foms}

In this section I propose a particular choice of sieve which has the advantage of being simple to implement: the first-order monotone spline sieve. This is a popular choice of sieve for seminonparametric models, see for example \textcite{heckman1984method,chen2007large,fox2016simple}. To define the sieve, let $\mathcal{B}_{n}=\{b_j:j=1,\ldots,B(n)\}$ be a set of knots that partition $\mathrm{Supp}(\beta)$ and $\mathcal{X}_{n}=\{\mathcal{X}_{k,n} :k=1,\ldots X(n)\}$ be a partition of $\mathrm{Supp}(X_1)$. The sieve space $\mathcal{M}_{n}$ is defined as follows:
{\small\begin{equation}\label{eq:fomss}
\left\{F\colon\mathrm{Supp}((\beta,X_1))\rightarrow[0,1]:F(b,x_{1})=\sum_{j=1}^{B(n)}\sum_{k=1}^{X(n)}P_{j,k}1(b_{j}\le{b})1(x_{1}\in\mathcal{X}_{k,n}),~P_{j,k}\ge0,\sum_{j=1}^{B(n)}P_{j,k}=1\right\},
\end{equation}}
where the sets $(\mathcal{B}_{n},\mathcal{X}_{n})$ are tuning parameters. For a given choice of tuning parameters, an element of $\mathcal{M}_{n}$ consists of $X(n)$ piecewise constant (step) functions in $b$, indexed by the partition cells $\mathcal{X}_n$, each such function having jumps of size $P_{j,k}$ at point $b_{j}$. The computational advantages of this sieve are clear: to find the supremum in (\ref{eq:smle}), for each $x_{1}$, the CCP functions need only be evaluated for the values $b_{j}\in\mathcal{B}_{n}$. This would not be the case if the sieve space consisted of functions that were continuous in $b$.

A theoretical advantage of this sieve space is that many of the high-level conditions for consistency are attained as long as the number of knots does not grow too fast. See Appendix \ref{ssec:fge} for details.

\begin{theorem}\label{thm:est_fomss}
Let $(A_{i,t},X_{i,t}:{t}=1,\dots,T)_{i=1}^{n}$ be  i.i.d. data generated from the DDC model satisfying either Assumptions \ref{As:inf_bl}-\ref{As:inf_os} or Assumptions \ref{As:f_bl}-\ref{As:f_rank}. If Assumptions \ref{As:est_fs}, \ref{As:est_fomss} and \ref{As:est_compl} hold, then the estimator $(\hat{\gamma},\hat{F}_{\beta|X_{1}})$ defined in equation \eqref{eq:smle} is consistent for $(\gamma,F_{\beta|X_{1}})$.
\end{theorem}

To implement the estimator, the number and location of grid points must be chosen. For consistency, it is enough that $B(n)X(n)\log(B(n)X(n))=o(n)$ and that the grid points become dense in the support of $(\beta,X_1)$. In principle, convergence rates for this estimator could be derived to determine optimal growth rates for $B(n),X(n)$.

For computation, it may be attractive to use profiling. In particular, to form  ($\hat{\gamma},\hat{F}_{\beta|X_{1}}$), fix $\gamma$ and let
\begin{equation*}
\hat{F}_{\beta|X_{1}}(\gamma)=\arg\sup_{F\in\mathcal{M}_{n}}\frac{1}{n}\sum_{i=1}^{n}\psi(y_{i},\hat{F}_{x},\gamma,F).
\end{equation*}
For $\mathcal{M}_n$ as in equation  \eqref{eq:fomss}, this is a convex optimization problem, with a unique global optimum that can be computed efficiently (e.g., \textcite{koenker2014convex}). The profile estimator is formed as
\begin{equation*}
\frac{1}{n}\sum_{i=1}^{n}\psi(y_{i},\hat{F}_{x},\hat{\gamma},\hat{F}_{\beta|X_{1}}(\gamma))\ge\sup_{\gamma\in\Gamma}\frac{1}{n}\sum_{i=1}^{n}\psi(y_{i},\hat{F}_{x},\gamma,\hat{F}_{\beta|X_{1}}(\gamma))-o_{p}(1/n).
\end{equation*}

\section{Simulations}\label{sec:sims}

This section investigates the estimator of Section \ref{sssec:foms} in a Monte Carlo simulation. The main goals of this section are twofold: first, to explore the finite sample performance of the estimator; and, second, to provide empirical support for the asymptotic results of Section \ref{sec:estimation}. I simulate data using a simple labor force participation model based on \textcite[Section 6]{altuug1998effect}, which also acts as a basis for the empirical illustration in Section \ref{sec:applic}.

In each period, each individual decides whether or not to enter the labor force, upon observation of the state variable. Thus $A=\{0,1\}$, with $a_{t}=1$ representing an individual decision to enter the labor force at time $t$. The period payoff from entering the labor market depends on the observed state variable $x_t=(x_{t,1},x_{t,2})^\intercal\in\mathbb{R}^2$, the entry-specific shock $\epsilon_{t,1}$, and  individual-specific labor productivity $\beta$ as follows:
\begin{equation*}
\beta x_{t,1}+\gamma{x}_{t,2}+\epsilon_{t,1}
\end{equation*}
Following the model of \textcite{altuug1998effect}, $x_{t,1}$ can be interpreted as an average consumption value (see Section \ref{sec:applic} for details) and $x_{t,2}$ is equal to the income of the primary earner in the household. The period payoff from not entering is $\epsilon_{t,0}$. The random preference shock $\epsilon_{t,a}$ is assumed to be distributed  extreme value type I and independent across time, choices and agents. Further, the agents' time horizon is assumed to be infinite with exponential discount factor 0.9. In addition, I assume that $\beta$ is independent of $X_1$ and consider three different choices for its distribution. In DGP 1, $\beta$ follows a mixture of three truncated normal distributions:
\begin{equation*}
\beta\sim
\begin{cases}
\mathcal{N}_{tr}(1.5,1)&\text{with prob. }1/3\\\mathcal{N}_{tr}(2.5,0.25)&\text{with prob. }1/3\\\mathcal{N}_{tr}(3.5,1)&\text{with prob. }1/3
\end{cases},
\end{equation*}
where $\mathcal{N}_{tr}(\mu,\sigma)$ is the truncated normal distribution with parameters ($\mu,\sigma$), minimum value $0$ and maximum value $50$. In DGPs 2 and 3, I assume $\beta$ follows a uniform distribution on $[0,5]$ and $\{1,2.5,4\}$, respectively. I assume that the first period observed state variable is drawn independently from the uniform distribution on $[0,4]\times[0,4]$, and that $F_x(x'|x,a)=F_1(x'_1|x,a)F_2(x'_2|x,a)$,
where $F_1$ and $F_2$ are truncated normal distributions with means $x_1/(a+2)$ and $(x_1+x_2)/(a+2)$ respectively, unit standard deviations and truncated to the interval $[0,4]$. I set $\gamma=2$.


The simulation results are the average of $1{,}000$ i.i.d. datasets $(a_{i,t},x_{i,t}:t=1,\dots,8)_{i=1}^{n}$ drawn from this model.\footnote{In practice,
the state space $[0,4]\times[0,4]$ and support of $\beta$ are discretized to solve the model. The discrete state space and support of $\beta$ have 400 and 1{,}000 points of support respectively.} Results are presented for four sample sizes: $n=100,500,1{,}000$, and $10{,}000$. For estimation I choose the number of grid points equal to $4n^{1/4}$ (i.e., $13, 19, 23$ and $40$), which satisfies the rate conditions required for Theorem \ref{thm:est_fomss}, and consider a grid of equally spaced points between 0 and 6. For estimation, I assume knowledge of the discount factor, the state transition $F_x$, and impose that the initial state is independent of $\beta$, leaving the unknown parameters as $(\gamma,F_\beta)$, the homogeneous effect of spousal income and the distribution of labor productivity.

\begin{table}[ht]
\centering
\begin{tabular}{ll|ccc|c|cc|ccc}
  \hline \hline& & \multicolumn{3}{c|}{$\gamma$} &  &  &  &  \multicolumn{3}{c}{No. types} \\
  & $n$ & Bias & Std & RMSE & Time & MISE & MIAE & Mean & Min & Max \\ \hline \hline  \multirow{4}{*}{DGP 1} & $100$ & -0.323 &  1.645 &  1.677 &  20 & 0.075 & 0.458 &  5.2 &  2 &  9 \\ 
   & $500$ & -0.222 &  1.694 &  1.708 &  22 & 0.039 & 0.333 &  6.8 &  4 & 10 \\ 
   & $1{,}000$ & -0.096 &  1.688 &  1.691 &  25 & 0.032 & 0.301 &  7.6 &  5 & 11 \\ 
   & $10{,}000$ &  0.070 &  1.646 &  1.647 & 143 & 0.020 & 0.240 & 10.1 &  6 & 21 \\ 
   \hline
\multirow{4}{*}{DGP 2} & $100$ & -0.347 &  1.679 &  1.715 &  21 & 0.070 & 0.479 &  5.5 &  3 &  8 \\ 
   & $500$ & -0.191 &  1.779 &  1.789 &  22 & 0.036 & 0.350 &  7.1 &  4 & 10 \\ 
   & $1{,}000$ & -0.121 &  1.751 &  1.755 &  26 & 0.029 & 0.312 &  7.8 &  4 & 11 \\ 
   & $10{,}000$ &  0.027 &  1.663 &  1.663 & 168 & 0.018 & 0.246 & 10.3 &  7 & 23 \\ 
   \hline
\multirow{4}{*}{DGP 3} & $100$ & -0.408 &  1.811 &  1.857 &  22 & 0.110 & 0.534 &  5.1 &  2 &  9 \\ 
   & $500$ & -0.332 &  1.822 &  1.852 &  23 & 0.062 & 0.361 &  6.1 &  3 & 10 \\ 
   & $1{,}000$ & -0.183 &  1.802 &  1.811 &  28 & 0.046 & 0.291 &  6.5 &  3 & 10 \\ 
   & $10{,}000$ & -0.206 &  1.639 &  1.652 & 145 & 0.018 & 0.136 &  7.3 &  4 & 14 \\ 
   \hline
\end{tabular}
\caption{\small {Simulation results for estimation of $\gamma$ and $F_{\beta}$ for each DGP and sample size. ``$\gamma$'' denotes results for estimation of $\gamma$, which includes $\sqrt{n}$ scaled average empirical bias (``Bias''), standard deviation (``Std'') and root mean-squared error (``RMSE''). ``Time'' denotes median computation time in seconds. ``MISE'' denotes empirical mean integrated squared error, ``MIAE'' denotes empirical mean integrated absolute error, and ``No. types'' denotes the number of support points.}} 
\label{table:sim2}
\end{table}
Table \ref{table:sim2} presents results for the estimator of $(\gamma,F_\beta)$, in addition to computation times. First consider results for $\gamma$. Here, empirical variance is significantly larger than empirical bias, which diminishes with sample size. Scaled empirical mean squared error is largely flat across sample sizes. In terms of computational burden, the fixed grid estimator takes around 30 seconds to run for the smaller sample sizes, though it takes around 2 minutes for $n=10{,}000$.

Turning to results for the estimation of $F_{\beta}$, both measures of integrated error diminish with sample size.\footnote{Integrated absolute and squared error for simulation run $m$ with estimate $\hat{F}_{\beta,m}$ is  $\int|\hat{F}_{\beta,m}(b)-F_{\beta}(b)|db$ and $\int\left(\hat{F}_{\beta,m}(b)-F_{\beta}(b)\right)^{2}db$, respectively.} The number of grid points increases slowly with sample size --- indeed slower than the growth of the number of support points selected by the estimator. For example, in DGP 1 for $n=100$, on average 5.2 points are selected. This increases to 10.1 for the large sample size. This pattern is broadly similar to previous simulation results for a parametric variant of this estimator \parencite{fox2011simple}. The number of support points chosen is similar between DGP 1 and DGP 2, but fewer points are chosen in the DGP with discrete types (DGP 3). Additional simulation results are presented in Appendix \ref{sec:sims_addl}.

\section{Empirical illustration}\label{sec:applic}

This section revisits the female labor supply model of \textcite{altuug1998effect}. I combine the life-cycle model of \textcite{altuug1998effect} with the identification results of Section \ref{sec:model} to estimate the distribution of labor productivity from data on labor force participation and perform a counterfactual exercise to measure how the response to a wage increase varies across the productivity distribution.

\subsection{Framework}

\textcite{altuug1998effect} introduces a framework to understand female labor supply that takes into account aggregate shocks and time non-separable preferences. In their model, agents gain utility from consumption and leisure. Under their specification of consumption and Pareto optimality, individual $i$ at time $t$ generates utility from consumption as:
\begin{equation}\label{amut:cons}
\eta_{i}\lambda_{t}\beta_{i}\omega_{t}\exp(\gamma_{3}^\intercal x_{Wi,t})l_{i,t}.
\end{equation}
The term $(\eta_{i}\lambda_{t})$ is the shadow value of consumption, which is estimated from data on consumption. The term ($\beta_{i}\omega_{t}\exp(\gamma_{3}'x_{Wit})l_{i,t}$) represents an individual's predicted earnings,\footnote{For clarity, in this section I will denote permanent unobserved heterogeneity as $\beta_i$.} which is equal to the amount of time they spend working conditional on participating, $l_{i,t}$, multiplied by their marginal product. The individual-specific marginal product of labor consists of unobserved aggregate and individual productivity effects ($\omega_{t},\beta_{i}$) in addition to a component that depends on covariates $x_{Wi,t}$. These terms are estimated from the wage equation, which is as follows:
\begin{equation*}
\tilde{w}_{i,t}=\omega_{t}\beta_{i}\exp(\gamma_{3}^\intercal x_{Wi,t})\exp(\tilde{\epsilon}_{i,t}).
\end{equation*}

\textcite{altuug1998effect} consider two estimators for the individual-specific productivity $\beta_{i}$. First, they use the fixed effects estimator from the wage equation above. Of course, in the asymptotic framework considered in this paper where $n$ is large but $T$ is fixed, this estimator is subject to the incidental parameters problem and is not consistent in general. For the second estimator, the authors assume that the fixed effect is an unknown function of observables, and then estimate that function non-parametrically. The observed variables consists of demographic data such as race, marital status and education levels. This estimator will be inconsistent if the set of observed variables is misspecified---that is, if individual productivity cannot be written as a function of observed data. The identification results of Section \ref{sec:model} obviate the need to estimate individual-specific productivity from the wage equation. Instead, $\beta_{i}$ can be interpreted as a random coefficient in the discrete choice model of labor force participation elaborated below.

\subsection{Model}

Suppose the per-period payoff from entering the labor market for individual of type $\beta_i$ is:
\begin{equation}\label{eq:myaltuug}
x_{i,t}^\intercal\left(\beta_i,\gamma^{\intercal}\right)^{\intercal}+\epsilon_{i,t,1}
\end{equation}
with $x_{i,t}=({z}_{i,t},1,\text{hinc}_{i,t},\text{age}_{i,t},\text{kids}_{i,t},\text{educ}_{i,t})$. Here ${z}_{i,t}$ is constructed following the approach of \textcite{altuug1998effect}, that is ${z}_{i,t}=\eta_{i}\lambda_{t}\omega_{t}\exp(\gamma_{3}^\intercal x_{Wi,t}){l}_{i,t}$ where each component is estimated from the consumption/wage regressions described above (see Appendix \ref{app:data_construct} for details). The remaining components of $x_{i,t}$ are, respectively, a constant term, annual head-of-household income, an age variable, whether there is a child in the household, and an education variable.\footnote{For simplicity, the age and education variable are dummies indicating whether the individual is over 35 year old and whether they have completed a college degree, respectively. In the DDC model, I assume that college degree status is constant over time (which is true for 97.5\% of individuals).}

Relative to the DDC model of participation in \textcite[Equation 6.7]{altuug1998effect}, $\beta_{i}$ is treated as an unobserved random variable. In their model $\beta_{i}$ is replaced by fixed effect estimates and treated as a known constant in their DDC model. Like \textcite{altuug1998effect}, I make the outside good assumption and assume that $\epsilon_{i,t,a}$ is distributed extreme value type I and independent across agents, time and actions. For simplicity, I assume that the agents' time horizon is infinite and that the exponential discount factor is 0.9 and known to the econometrician.

\subsection{Data and estimator}

As in \textcite{altuug1998effect}, the labor force participation model is estimated using a subset of data from the PSID. The data construction is described in Appendix \ref{app:data_construct}, and closely follows the details in \textcite[Appendix B]{altuug1998effect}.  The final data set contains 3084 individuals, each of whom have between four and ten panel observations, with an average close to eight.

I estimate the model using the two-step estimator described in Section \ref{sssec:foms}. The first step consists of estimating the state transition $F_x(x'|x,a)$. To simplify this step, I 
assume that, conditional upon $A=a$, (i) $X'-X$ is independent of $X$ and (ii) the components of $X'-X$ are mutually independent. Then, I estimate the densities $Z'-Z \mid A=a$ and $Hinc'-Hinc \mid A=a$ for each $a=0,1$ via the kernel density estimator with the Gaussian kernel and rule-of-thumb bandwidth.\footnote{The bandwidth is $1.06 \mathrm{std}\left[\sum_{i=1}^{n}\sum_{t=1}^{T_i-1}1\{A_{i,t}=a\}(y_{i,t+1}-y_{i,t})\right] (\sum_{i=1}^{n}\sum_{t=1}^{T_i-1}1\{A_{i,t}=a\})^{-1/5}$, for ${y=Z,Hinc}$ where $\mathrm{std}$ denotes standard deviation and $T_i$ is the panel length of observation $i$.} I note that these restrictions satisfy the real analyticity requirement Assumption \ref{As:inf_puh}(iv) (more precisely, its generalization in Section \ref{sec:discrete_state}).\footnote{This model has additional state variables with homogeneous effects (i.e., $k>\dim(\beta)+1$ where $k=\dim(X_t)=6$); as discussed in Remark \ref{rem:discrete}, the conditions of Section \ref{sec:model} must be adapted accordingly. A formal statement of these conditions is provided in Section \ref{sec:discrete_state}.} To see this, observe that, under the above specifications, the state transition cumulative distribution function is
$$
\Pr(X'\le x \mid X=x,A=a)=\Phi\left(\frac{z'-z-\mu_{1,a}}{\sigma_{1,a}}\right)\Phi\left(\frac{hinc'-hinc-\mu_{2,a}}{\sigma_{2,a}}\right)h(d',d,a),
$$
where $\Phi$ is the standard normal cumulative distribution function, $d=(educ,age,kids)$ are the discrete variables, and $h,\mu_{1,a},\sigma_{1,a},\mu_{2,a},\sigma_{2,a}$ are unknown parameters to be estimated. Thus, for each fixed $(x',d,a)$, the state transition is a bounded real analytic function of $(z,hinc)$ that is supported on $\mathbb{R}^2$. Given this discussion and model assumptions described above, the two sufficient conditions for injectivity are satisfied; then, for identification, I impose the required support condition, which appears plausible given both $Z_{i,t}$ and $Hinc_{i,t}$ are continuous random variables.

The second step requires specifying a sieve space for $\beta_i$. The step-wise constant sieve space of Section \ref{sssec:foms} is adopted, with the number and location of the knots as tuning parameters. For simplicity, $\beta_i$ is assumed independent of $X_{i,1}$. Consistent with the simulation design, the number of knots is set to $4n^{1/4} \approx 30$, placed uniformly between 0 and 15. The lower bound of 0 reflects a natural restriction on labor productivity, while the upper bound of 15 is sufficiently large that, for reasonable parameter values, the conditional choice probability is close to 1.

I implement the estimator using the profiling approach described in Section \ref{sssec:foms}.\footnote{The remaining tuning parameter is the starting value of $\gamma$, which is set as the estimates from the same estimator with five knots, equally spaced between 0 and 15. That estimator is itself initialized with the estimates from the parametric model (i.e., under the assumption that $\beta_i$ is degenerate with unknown support).} The model solutions required in the second step are obtained following \textcite{kristensen2019solving}. Inference is conducted using the standard bootstrap, see Appendix \ref{sec:sims_addl} for evidence on its performance in a simulation exercise. Additional results on the fit of the estimated model are provided in Appendix \ref{sec:model_fit}.

\subsection{Results}

Table \ref{table:altug} presents point estimates of the finite dimensional parameter $\gamma$ alongside bootstrapped standard errors. Estimates indicate that utility from working increases with education, but decreases with head-of-household income and age. Having children in the household is estimated to have a negligible effect on utility from working.

\begin{table}[ht]
\centering
\begin{tabular}{ccccc}
  \hline\hline Intercept & $hinc_{i,t}$ & $kids_{i,t}$ & $age_{i,t}$ & $educ_{i,t}$   \\  \hline
-2.527 & -0.312 &  0.054 & -0.610 &  0.331 \\ 
  {\small(0.1279)} & {\small(0.0276)} & {\small(0.0779)} & {\small(0.0758)} & {\small(0.0874)} \\ 
   \hline
\hline
\end{tabular}
\caption{Point estimates of $\gamma$ for the participation model of Section \ref{sec:applic}. Standard errors are in parentheses, calculated as the standard deviation of the estimator over 1{,}000 bootstrap samples.} 
\label{table:altug}
\end{table}

Figure \ref{fig:altug} presents the estimated distribution of $\beta_i$ from the fixed grid estimator. The estimated distribution has 21 points of support, with mean 3.11, median 3.11, standard deviation 1.35, skewness 2.39 and kurtosis 15.83, indicating substantial heterogeneity in labor productivity.\footnote{For comparison, in a model where $\beta_i$ is assumed to have three unknown points of support and estimated using the method of \textcite{arcidiacono2003finite}, the estimated distribution has mean 2.93, median 2.56, standard deviation 0.91, skewness -0.57 and kurtosis  2.29. See Appendix \ref{app:appl_AJ}.}

\begin{figure}[H] \centering \includegraphics[width=\textwidth]{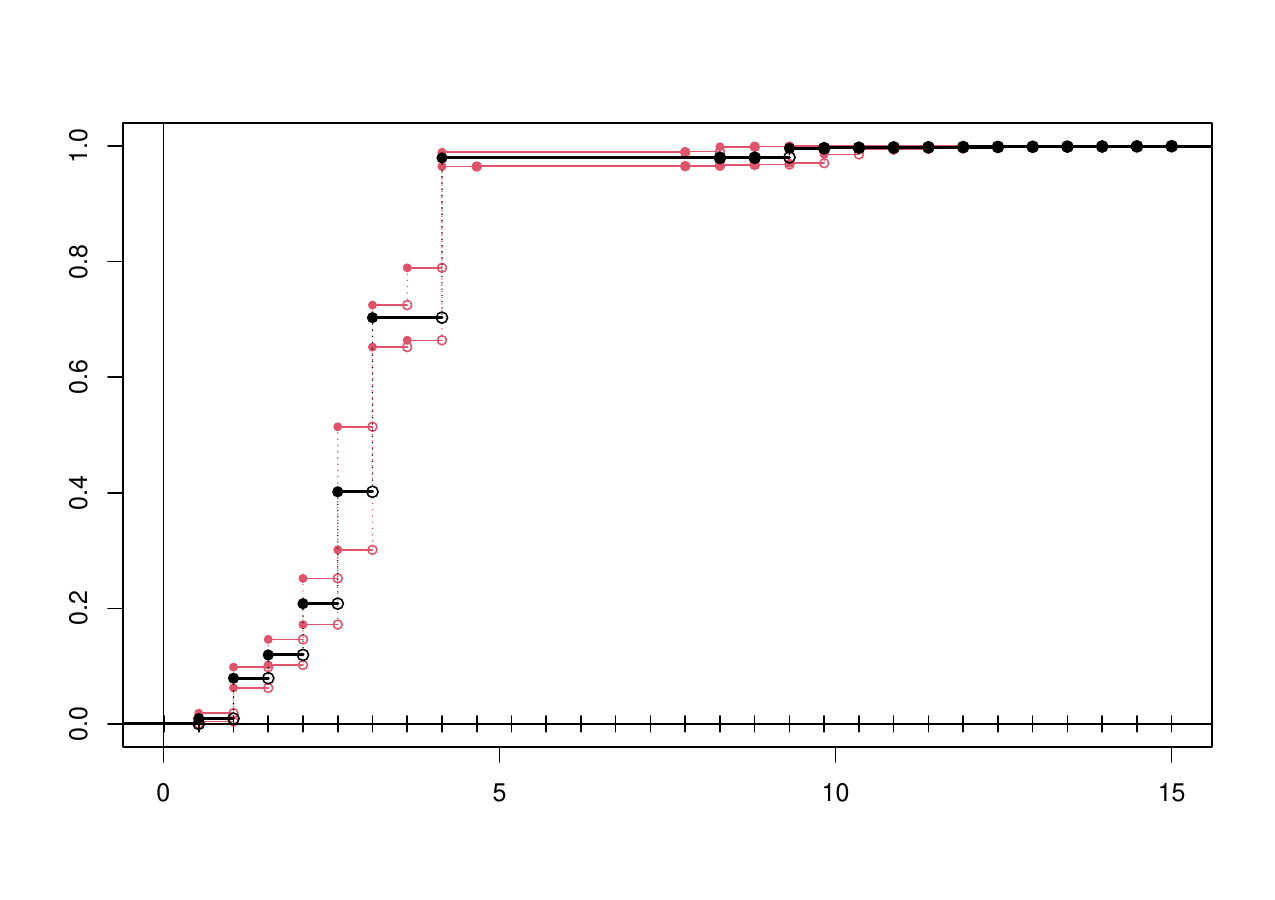}
\vspace*{-20mm} 	\caption{\small{Estimated distribution of
  $\beta_{i}$ for the participation model of Section \ref{sec:applic}. The black curve represents the point estimate, the red curves represent
  bootstrapped 95\% pointwise confidence intervals. The ticks on
          the x-axis represent the grid points.}}\label{fig:altug}
\end{figure}

\subsection{Counterfactual analysis}\label{sec:cf}

In this section, I conduct a counterfactual exercise to measure how wages affect labor market participation across the skill distribution. The counterfactual considered is where the agent's expected wage received from working (i.e., under $A_{i,t}=1$) is increased by $x\%$ over its status quo value, for $x=5,10,15,20,25$, holding all else fixed.\footnote{In the model described above, agent $i$'s expected wage from working in period $t$ is $\omega_t\beta_i\exp(\gamma_3^\intercal x_{Wi,t})$.} For each counterfactual wage change of $x\%$, I draw $(\beta^{(x)}_m,X^{(x)}_{m,t},A^{(x)}_{m,t}:t=1,\ldots,T)_{m=1}^{M}$ for $M=1{,}000{,}000$ and $T=5$ from the estimated model,\footnote{\label{cf_desc}Each simulated panel $m=1,2,\ldots,M$  is drawn independently as follows. First, $\beta_m$ is drawn from the estimated distribution $\hat{F}_\beta$ and $X_{m,1}$ is drawn from the empirical distribution of $X_{i,1}$. Then the conditional choice probability $P(1,X_{m,1},\beta_m;\hat{F}_x,\hat{\gamma})$ is computed and used to draw $A_{m,1}$. Next, $X_{m,2}$ is set as $X_{m,1}+\xi_{A_{m,1}}$ where $\xi_a$ is drawn uniformly from the empirical distribution of $X'-X \mid A=a$, with the draw truncated to respect the empirical supports. $A_{m,2}$ and $(X_{m,t},A_{m,t})$ for $t=3,\ldots,T$ are drawn analogously.} and report the average labor market participation rate for six different quantiles of $\beta$.

Table \ref{table:cf_data} displays the results of this counterfactual exercise. Each cell displays the average labor market participation for the counterfactual wage increase conditional upon a particular quantile of $\beta$. Specifically, for a $x\%$ wage increase and quantile $q_\alpha \equiv \inf \{c:\hat{F}_\beta(c)\ge \alpha\}$, the table reports
$$
\frac{\sum_{m=1}^M\sum_{t=1}^T1\{A^{(x)}_{m,t}=1,\beta^{(x)}_m=q_\alpha\}}{T\sum_{m=1}^M1\{\beta^{(x)}_m=q_\alpha\}}.
$$
The table also displays the implied elasticity of quantile-specific labor force participation with respect to wages, based upon the 25\% wage increase.\footnote{Specifically, the elasticity is calculated as $\frac{\left(\frac{\sum_{m=1}^M\sum_{t=1}^T1\{A^{(25)}_{m,t}=1,\beta^{(25)}_m=q_\alpha\}}{T\sum_{m=1}^M1\{\beta^{(25)}_m=q_\alpha\}}\right)-\left(\frac{\sum_{m=1}^M\sum_{t=1}^T1\{A^{(0)}_{m,t}=1,\beta^{(0)}_m=q_\alpha\}}{T\sum_{m=1}^M1\{\beta^{(0)}_m=q_\alpha\}}\right)}{\frac{\sum_{m=1}^M\sum_{t=1}^T1\{A^{(0)}_{m,t}=1,\beta^{(0)}_m=q_\alpha\}}{T\sum_{m=1}^M1\{\beta^{(0)}_m=q_\alpha\}}.}.$} For comparison, total (i.e., unconditional) labor force participation is 0.6496, and its elasticity with respect to wages is estimated to be approximately 0.11. Standard errors for the counterfactual estimates are in Table \ref{table:cf_std_errors}.

\begin{table}[H]
\centering
\begin{tabular}{l|cccccc}
  \hline\hline & \multicolumn{6}{c}{Quantile of labor productivity $\beta$}  \\ Wage increase & $q_{0.01}$ & $q_{0.2}$ & $q_{0.4}$ & $q_{0.6}$ & $q_{0.8}$ & $q_{0.99}$ \\ 
   \hline
0\% & 0.1312 & 0.4127 & 0.5597 & 0.7168 & 0.8992 & 0.9998 \\ 
  5\% & 0.1360 & 0.4192 & 0.5649 & 0.7207 & 0.9010 & 0.9999 \\ 
  10\%  & 0.1408 & 0.4257 & 0.5699 & 0.7245 & 0.9027 & 0.9999 \\ 
  15\%  & 0.1457 & 0.4319 & 0.5748 & 0.7282 & 0.9043 & 0.9999 \\ 
  20\%  & 0.1502 & 0.4378 & 0.5796 & 0.7317 & 0.9058 & 0.9999 \\ 
  25\%  & 0.1546 & 0.4439 & 0.5840 & 0.7350 & 0.9073 & 0.9999 \\ 
   \hline
Elasticity: & 0.7129 & 0.3022 & 0.1737 & 0.1015 & 0.0357 & 0.0001 \\ 
   \hline
\hline
\end{tabular}
\caption{Counterfactual labor force participation rates. Each cell represents estimated labor force participation rates under a counterfactual $x\%$ increase in wages (for $x=0,5\ldots,25$) among those with labor productivity $q_{\alpha}$, which denotes  the $\alpha$'th percentile (for $\alpha=0.01,0.2,\ldots,0.99$) of the estimated distribution of $\beta$. The estimates are based on $1{,}000{,}000$ draws from the model evaluated at the estimated parameter values and counterfactual wages. ``Elasticity'' is the implied percent change in labor force participation from a 1\% increase in counterfactual wages (calculated using the 25\% counterfactual wage increase).} 
\label{table:cf_data}
\end{table}

Several observations can be made from this counterfactual exercise. First, average labor force participation varies greatly across the distribution of productivity. For instance, it increases from 13\% at the first percentile to almost 100\% at the 99th percentile.\footnote{In the data, around 14.6\% of individuals never work. This percentage is 10.6\% in the simulated data with no wage change.} Second, the supply response to a wage increase is much larger at lower skill quantiles: the implied elasticity is 0.30 at the 20th percentile, but only 0.036 at the 80th percentile.

\section{Conclusion}\label{sec:conclusion}

In this paper I show point identification of a broad class of multinomial dynamic discrete choice models with multivariate continuous permanent unobserved heterogeneity. Relative to the existing literature, I allow for permanent unobserved heterogeneity that is both multivariate and continuous, and provide low-level conditions for point identification.  My results encompass both finite and infinite horizon models, and do not rely on a full support condition, nor parametric assumptions on the distribution on permanent unobserved heterogeneity.

I propose a seminonparametric estimator for the distribution of continuous permanent unobserved heterogeneity in the style of \textcite{heckman1984method}. The estimator is computationally simple, and coincides with the estimator for a semiparametric model. As a result, the applied econometrician can proceed as they would for discrete permanent unobserved heterogeneity, providing they commit to increasing the number of support points as the sample size grows.

\addcontentsline{toc}{section}{\refname}
\printbibliography

\appendix

\section{Proofs}

Throughout this appendix I use the following notations: $S_\beta=\mathrm{Supp}(\beta)$; for $\lambda$ the Lebesgue measure, $\mathcal{L}^2_{A}$ is the usual $L^2$ space $\mathcal{L}^2(A,\lambda)$  and $\mathcal{L}_{A}$ is the usual $L^\infty$ space $\mathcal{L}^\infty(A,\lambda)$; $\text{sp}{A}$ indicates the linear span of set $A$, and $\overline{\text{sp}}{A}$ indicate its closure in the $L^2$ norm.

\subsection{Proof of results in Section \ref{sec:model}}\label{pf:inj}

\subsubsection{Proof of Theorem \ref{thm:inject}}

\begin{proof}\label{pf:thm_inj}
For $\mathcal{V}\subset\mathbb{R}^k$, define
\begin{align*}
L^*_{\beta,\mathcal{V}}:\mathcal{L}_{S_{\beta}}\rightarrow\mathcal{L}_{\mathcal{V}}&\qquad[L^{*}_{\beta,\mathcal{V}}m](x)=\int{P}(0,x,b)m(b)db.
\end{align*}
Note that any absolutely continuous measure $\mu$ on $S_\beta$ with bounded density $m$ satisfies $m \in \mathcal{L}_{S_\beta}$. Let $\mathcal{X}\subseteq\mathrm{Supp}(X_t)$ be a non-empty open set. By Lemma \ref{lem:inf_ccp}, for each fixed $b\in S_\beta$, the map $x\mapsto P(0,x,b)$ is real analytic on $\mathbb{R}^{k}$. Now, by Lemma \ref{lem:SW3.8}, if $[L^*_{\beta,\mathcal{X}} m](x)=0$ almost everywhere on $\mathcal{X}$, it follows that $[L^*_{\beta,\mathcal{X}} m](x) = 0$ for all $x \in \mathcal{X}$. Therefore, to prove the theorem it suffices to show injectivity of $\mathcal{L}_{\beta,\mathcal{X}}$, which in turn follows from injectivity of  $L^*_{\beta,\mathbb{R}^k}$ by Lemma \ref{lem:SW3.8}. To show this, define 
\begin{equation}\label{functions0}
  \tilde{\mathcal{H}}=\left\{b\mapsto P(0,x,b)\colon \;x\in\mathbb{R}^{k}\right\}.
\end{equation}
By Lemma \ref{lem:inf_ccp} and Theorem 3.1 in \textcite{stinchcombe1998consistent}, $L^*_{\beta,\mathbb{R}^{k}}$ is injective if ${\mathrm{sp}}\tilde{\mathcal{H}}$ is dense in $\mathcal{L}^{2}_{S_\beta}$. The result follows from Lemma \ref{lemma_cosapprox}.
\end{proof}

\subsubsection{Proof of Theorem \ref{thm:inf}}\label{pf:id}

\begin{proof}[Proof of Theorem \ref{thm:inf}]\label{pf:inf}
By Assumptions \ref{As:inf_bl} and \ref{As:inf_puh},
\begin{multline*}
f_{A_{4}A_{3}A_{2}A_{1}X_{4}X_{3}X_{2}|X_{1}}(a_{4},a_{3},0,a_{1},x_{4},x_{3},x_{2},x_{1})=\int{P}(a_{4},x_{4},b)F_{x}(x_{4}|x_{3},a_{3})P(a_{3},x_{3},b)\\\times{F_{x}(x_{3}|x_{2},0)}P(0,x_{2},b)F_{x}(x_{2}|x_{1},a_{1})P(a_{1},x_{2},b)f_{\beta|X_{1}}(b,x_{1})db.
\end{multline*}
Where the transition kernel has positive measure, we can write
\begin{multline*}
\frac{f_{A_{4}A_{3}A_{2}A_{1}X_{4}X_{3}X_{2}|X_{1}}(a_{4},a_{3},0,a_{1},x_{4},x_{3},x_{2},x_{1})}{F_{x}(x_{4}|x_{3},a_{3})F_{x}(x_{3}|x_{2},0)F_{x}(x_{2}|x_{1},a_{1})}\\=\int{P}(a_{4},x_{4},b)P(a_{3},x_{3},b)P(0,x_{2},b)P(a_{1},x_{1},b)f_{\beta|X_{1}}(b,x_{1})db.
\end{multline*}

Fix $x_1\in\mathrm{Supp}(X_1)$ and let $a_1\in\mathrm{Supp}(A_1)$ satisfy Assumption \ref{As:inf_os}. Let $S_2=\mathrm{Supp}(X_2|X_1=x_1,A_1=a_1)$ and $S_4=\cap_{a_3\in A}\mathrm{Supp}(X_4 \mid X_3 \in S_3,A_3 = a_3)$ and define the operators $L_{3,4,2}:\mathcal{L}_{S_{2}}\rightarrow{A}\times\mathcal{L}_{S_{3}}$
and
$L_{3,2}:\mathcal{L}_{S_{{2}}}\rightarrow{A}\times\mathcal{L}_{S_{3}}$
as follows:
\begin{align*}&[L_{3,4,2}m](a_{3},x_{3})=\int\frac{f_{A_{4}A_{3}A_{2}A_{1}X_{4}X_{3}X_{2}|X_{1}}(a_{4},a_{3},0,a_{1},x_{4},x_{3},x_{2},x_{1})}{F_{x}(x_{4}|x_{3},a_{3})F_{x}(x_{3}|x_{2},0)F_{x}(x_{2}|x_{1},a_{1})}m(x_{2})dx_{2},\\
&[L_{3,2}m](a_{3},x_{3})=\int\frac{f_{A_{3}A_{2}A_{1}X_{3}X_{2}|X_{1}}(a_{3},0,a_{1},x_{3},x_{2},x_{1})}{F_{x}(x_{3}|x_{2},0)F_{x}(x_{2}|x_{1},a_{1})}m(x_{2})dx_{2}.
\end{align*}
Under Assumption \ref{As:inf_os} the above operators are observed and well-defined for some fixed $(x_4,a_4)$. The operators can be decomposed into constituent parts. For this purpose define
\begin{align*}
L_{3,\beta}:\mathcal{L}_{S_{\beta}}\rightarrow{A}\times\mathcal{L}_{S_3}\qquad&[L_{3,\beta}m](a_{3},x_{3})=\int{P}(a_{3},x_{3},b)m(b)db,\\
D^{4}_{\beta}:\mathcal{L}_{S_{\beta}}\rightarrow\mathcal{L}_{S_{\beta}}\qquad&[D^{4}_{\beta}m](b)=P(a_{4},x_{4},b)m(b),\\
D_{\beta}:\mathcal{L}_{S_{\beta}}\rightarrow\mathcal{L}_{S_{\beta}}\qquad&[D_{\beta}m](b)=P(a_{1},x_{1},b)f_{\beta|X_{1}}(b,x_{1})m(b),\\
L_{\beta,2}:\mathcal{L}_{S_2}\rightarrow\mathcal{L}_{S_{\beta}}\qquad&[L_{\beta,2}m](b)=\int{P}(0,x_{2},b)m(x_{2})dx_{2}.
\end{align*}
It is straightforward to derive that
$L_{3,4,2}=L_{3,\beta}D_{\beta}^{4}D_{\beta}L_{\beta,2}$ and  $L_{3,2}=L_{3,\beta}D_{\beta}L_{\beta,2}$.

By Theorem $\ref{thm:inject}$, $L_{3,\beta}$ and $L_{\beta,2}^*$ are injective where $L_{\beta,2}^{*}$ is the adjoint\footnote{The adjoint of a linear operator between Hilbert Spaces $L:U\rightarrow{V}$ is the operator $L^{*}:V\rightarrow{U}$ that satisfies $\langle Lu,v\rangle _{V}=\langle u,L^{*}v\rangle _{U}$ where $\langle \cdot,\cdot\rangle _{W}$ is the inner product on $W$. See \textcite{carrasco2007linear} for further discussion.} of $L_{\beta,2}$. Then, since $D_{\beta}$ is invertible (as $P(a_{1},x_{1},b)f_{\beta|X_{1}}(b,x_{1})>0$ almost surely-$\mathrm{Supp}(\beta|X_1=x_1)$) and $L_{3,\beta}$ and $L_{\beta,2}^{*}$ are injective, $L_{3,2}$ has a right inverse,\footnote{Following \textcite{hu2008identification}, by `right inverse' we mean the existence of an operator $L_{3,2}^{-1}$ such that $L_{3,2}L_{3,2}^{-1}:\mathcal{L}_{S_2}\rightarrow\mathcal{L}_{S_2}$ is the identity operator.} the equivalence
\begin{equation}\label{eq:decomp}
L_{4,3,2}L_{3,2}^{-1}=L_{3,\beta}D_{\beta}^{4}L_{3,\beta}^{-1}
\end{equation}
holds and $L_{3,\beta}D_{\beta}^{4}L_{3,\beta}^{-1}$ is the eigendecomposition of the known operator $L_{4,3,2}L_{3,2}^{-1}$ \parencite[Lemma A.1]{williams2019nonparametric}. Each $b$ indexes an eigenvalue $P(a_4,x_4,b)$ of $L_{4,3,2}L_{3,2}^{-1}$, with corresponding eigenfunction $(a_{3},x_{3})\mapsto{P(a_{3},x_{3},b)}$. As in \textcite{hu2008instrumental}, the decomposition is unique up to (1) scaling of the eigenfunctions, (2) uniqueness of the eigenvalues, and (3) reindexing of the eigenvalues (``ordering'').

First, the scale of the eigenfunctions $(a_{3},x_{3})\mapsto{P(a_{3},x_{3},b)}$ is fixed since they are probabilities that must satisfy $\sum_{a_3\in A}P(a_{3},x_{3},b)=1$. Second, for eigenvalue uniqueness, as shown in \textcite[][p. 213]{hu2008instrumental}, it is sufficient that for each $b\neq\tilde{b}\in{S}_{\beta}$, there exist some $(a_{4},x_{4})\in{A}\times S_4$ such that $P(a_{4},x_{4},b)\neq{P}(a_{4},x_{4},\tilde{b})$. To show this, suppose for all $(a_{4},x_{4})\in{A}\times S_4$, $P(a_{4},x_{4},b)={P}(a_{4},x_{4},\tilde{b})$. Then, by standard arguments for identification of homogenous parameters in DDC models \parencite[e.g.,][Section 3.5]{bajari2015identification}, it follows that for each $a\in A$
\begin{equation*}
\begin{pmatrix}
\tilde{b}_{a}&\tilde\gamma_{a}^\intercal
\end{pmatrix}^{\intercal}x_{4}
=\begin{pmatrix}
b_{a}&\gamma_{a}^\intercal
\end{pmatrix}^{\intercal}x_{4}.
\end{equation*}
Then, since $S_4$ contains $k$ linearly independent elements, $\tilde{b}_{a}=b_{a}$ and thus $\tilde{b}=b$ as required.

Finally, the problem of ordering arises because any injective function $R$ may be used to redefine the latent variable $\beta=R(\tilde{\beta})$ while satisfying $L_{3,\beta}D_{\beta}^{4}L_{3,\beta}^{-1}=L_{3,\tilde\beta}D_{\tilde\beta}^{4}L_{3,\tilde\beta}^{-1}$\footnote{This equality is shown explicitly in \textcite[Supplement S.3]{hu2008instrumental}.} where
\begin{align*}
L_{3,\tilde{\beta}}:\mathcal{L}_{S_{\tilde\beta}}\rightarrow{A}\times\mathcal{L}_{S_{{3}}}\qquad&[L_{3,\tilde{\beta}}m](a,x)=\int\Pr(A_{3}=a\mid{X}_{3}=x,\tilde\beta=b)m(b)db,\\
D^{4}_{\tilde\beta}:\mathcal{L}_{S_{\tilde\beta}}\rightarrow\mathcal{L}_{S_{\tilde\beta}}\qquad&[D^{4}_{\tilde{\beta}}m](b)=\Pr(A_{4}=a_{4}\mid{X}_{4}=x_{4},\tilde\beta=b)m(b).
\end{align*}
Notice that $\Pr(A_{3}=a\mid{X}_{3}=x,\tilde\beta=b)=\Pr(A_{3}=a\mid{X}_{3}=x,\beta=R(b))=P(a,x,R(b))$. I show the only admissible reordering function is identity. For this purpose, suppose that for all $(a_{3},x_{3})\in{A}\times{S}_{{3}}$, $P(a_{3},x_{3},R(b))=P(a_{3},x_{3},b)$. By standard arguments for identification of homogenous parameters in DDC models \parencite[e.g.,][Section 3.5]{bajari2015identification}, it follows that for each $a\in{A}$,
\begin{equation*}
\begin{pmatrix}
R(b_{a})&\tilde\gamma_{a}^\intercal
\end{pmatrix}^{\intercal}x_{3}
=\begin{pmatrix}
b_{a}&\gamma_{a}^\intercal
\end{pmatrix}^{\intercal}x_{3}.
\end{equation*}
Under Assumption \ref{As:inf_os}(ii) $S_{{3}}$ contains $k$ linearly independent vectors, so it follows that $\left(R(b_{a}),\tilde{\gamma}^{\intercal}_{a}\right)^{\intercal}=\left(b_{a},\gamma^{\intercal}_{a}\right)^{\intercal}$ and thus $R(b)=b$. Thus $P(a,x,b)$ is identified as the unique eigenfunction of $L_{4,3,2}L_{3,2}^{-1}$, yielding identification of $\gamma$ under Assumption \ref{As:inf_os}(ii).

To identify $f_{\beta|X_{1}}$, notice that
\begin{equation*}
\frac{f_{a_{2}a_{1}x_{2}|x_{1}}(0,a_{1},x_{2},x_{1})}{F_{x}(x_{2}|x_{1},a_{1})}=\left[L_{\beta,2}^{*}(P(a_{1},x_{1},\cdot)f_{\beta|X_{1}}(\cdot\,x_{1}))\right](x_{2}).
\end{equation*}
$L_{\beta,2}^{*}$ is injective and identified, since its kernel is identified. Applying the left inverse of $L_{\beta,2}^{*}$, $P(a_{1},x_{1},b)f_{\beta|X_{1}}(b,x_{1})$ and thus $f_{\beta|X_{1}}(b,x_{1})$ is identified.
\end{proof}

\subsubsection{Proof of Lemma \ref{lemma_cosapprox}}\label{ssec:lemma_cosapprox}

\begin{proof}
Under Assumptions \ref{As:inf_bl} and \ref{As:inf_puh}, 
\begin{equation}\label{eq:ccp2}
P(a,x,b)=\frac{\exp\left(x^\intercal(b_{a},\;\gamma_{a}^{\intercal})^{\intercal}+\rho\int{v}(x';b)dF_{x}(x'|x,a)\right)}{\sum_{\tilde{a}\in A}\exp\left(x^\intercal(b_{\tilde{a}},\;\gamma_{\tilde{a}}^{\intercal})^{\intercal}+\rho\int{v}(x';b)dF_{x}(x'|x,\tilde{a})\right)},
\end{equation}
and define $\tilde{\mathcal{H}}\equiv\left\{b\mapsto P(0,x,b)\colon\;x\in  \mathbb{R}^{k}\right\}$. First, I show that for any $l=(l_1,l_2,\ldots,l_J)^\intercal\in\mathbb{R}^{J}$ there is a sequence in $\tilde{\mathcal{H}}$ whose limit is $1\{b\in\times_{a=1}^{J}(l_a,\infty)\}$. Given $l\in\mathbb{R}^{J}$ and $n\in\mathbb{N}$, let $\tilde{x}_n=n\Gamma^{-1}l$, which exists due to Assumption I2\ref{As:inf_homog}. Denote $x_n=(-n,\tilde{x}_n^\intercal)^\intercal$. If
\begin{equation}\label{eq:value_limit}
\lim_{n\rightarrow\infty}\frac{x_n^\intercal(b_{a},\;\gamma_{a}^{\intercal})^{\intercal}+\rho\int{v}(x';b)dF_{x}(x'|x,a)}{x_n^\intercal(b_{a},\;\gamma_{a}^{\intercal})^{\intercal}}=1
\end{equation}
then, for any $b\in S_\beta$, $P(0,x_n,b)\rightarrow {1}\{b \in \times_{j=1}^J (l_j,\infty)\}$ as $n\rightarrow\infty$. Since $x_n^\intercal(b_{a},\;\gamma_{a}^{\intercal})^{\intercal}=-n(b_a-l_a)$ diverges when $b_a\neq l_a$, for equation \eqref{eq:value_limit} it is sufficient that $\int{v}(x';b)dF_{x}(x'|x,a)$ is uniformly bounded in $(a,x,b)\in A\times \mathbb{R}^k\times S_\beta$. Denote $S_{X'}$ as the support of the state transition kernel and consider that
\begin{align*}
  \left|\int{v}(x';b)dF_{x}(x'|x,a)\right|
  \le&\int\left|{v}(x';b)\right|\left|dF_{x}(x'|x,a)\right|\\
=&\int_{x'\in{S}_{X'}}\left|{v}(x';b)\right|\left|dF_{x}(x'|x,a)\right|+\int_{x'\not\in{S}_{X'}}\left|{v}(x';b)\right|\left|dF_{x}(x'|x,a)\right|\\
=&\int_{x'\in{S}_{X'}}\left|{v}(x';b)\right|\left|dF_{x}(x'|x,a)\right|\\
<&\;M
\end{align*}
for some $M<\infty$. The second equality is because $dF_x(x'|x,a)=0$ for any $x'\not\in S_{X'}$, the final inequality follows since (i) $v(x;b)$ is bounded on the compact set $S_{X'}\times S_\beta$ \parencite{kristensen2019solving}, and (ii) $dF_x(x'|x,a)$ is a bounded function of $x$ (Assumption I2\ref{As:inf_tran}) and $S_{X'}\times A$ is compact.

Next, it follows that, for any $u=(u_1,u_2,\ldots,u_J)^\intercal\in\mathbb{R}^J$ there is a sequence $(h_n)_{n\in\mathbb{N}}\subset\mathrm{sp}\tilde{\mathcal{H}}$, each element formed by adding and subtracting $2^{J}$ elements of $\tilde{\mathcal{H}}$, such that, as $n\rightarrow\infty$, $h_n(b)\rightarrow1\{b\in\times_{a=1}^{J} (l_a,u_a]\},$ which implies 
\begin{equation*}
 \overline{\mathrm{sp}}\tilde{\mathcal{H}}\supset\left\{b\mapsto1\{b\in\times_{a=1}^{J} (l_a,u_a]\}\colon\; l,u\in\mathbb{R}^J\right\}.
\end{equation*}
 
To conclude we show ${\mathrm{sp}}\tilde{\mathcal{H}}$ is dense in simple functions on $S_\beta$. Let $E\subset S_\beta$ be Lebesgue measurable and let $\epsilon>0$, and denote $\chi_E(b)=1\{b\in E\}$. From \textcite{rudin1987real} Theorem 2.17(a), there is a set $\mathcal{O}=\cup_{i=1}^n\times_{j=1}^J(l_{j,i},u_{j,i}]\subset S_\beta$ such that the Lebesgue measure of $E\Delta\mathcal{O}\equiv(E\setminus\mathcal{O})\cup(\mathcal{O}\setminus E)$ is at most $\epsilon$. Note that $\chi_{\mathcal{O}}(b)\in\overline{\mathrm{sp}}\tilde{\mathcal{H}}$ and that $\chi_{\mathcal{O}}$ and $\chi_{E}$ agree on $S_\beta\setminus(E\Delta\mathcal{O})$. Then since $|\chi_E(b)-\chi_{\mathcal{O}}(b)|\le 1$,
\begin{align*}
    \int_{S_\beta}|\chi_E(b)-\chi_{\mathcal{O}}(b)|^2db&=\int_{E\Delta\mathcal{O}}|\chi_E(b) -\chi_{\mathcal{O}}(b)|^2db+\int_{S_\beta \setminus (E\Delta\mathcal{O})}|\chi_E(b)-\chi_{\mathcal{O}}(b)|^2 db\\
    &<\epsilon + 0. \qedhere
\end{align*}
\end{proof}

\subsubsection{Supporting lemmas}\label{ssec:lemmas}

\begin{lemma}[{Properties of the CCP function}]\label{lem:inf_ccp}
Assume \ref{As:inf_bl} and \ref{As:inf_puh}. If $\mathrm{Supp}(X_t)$ contains a non-empty open set, then $\tilde{\mathcal{H}}=\left\{b\mapsto P(0,x,b)\colon\;x\in\mathbb{R}^{k}\right\}$ is a norm bounded subset of $\mathcal{L}^2_{{S}_{\beta}}$. Moreover, $x\mapsto P(a,x,b)$ are real analytic functions on $\mathbb{R}^{k}$ for any fixed $(a,b)$.
\end{lemma}
\begin{proof}[Proof of Lemma \ref{lem:inf_ccp}]

Under \ref{As:inf_bl} and \ref{As:inf_puh}, for any $(a,x,b)\in A\times\mathrm{Supp}(X_t)\times S_\beta$,\footnote{Recall that the integrated value function was defined in equation \eqref{eq:fixed_point} as $v_t(s)$. I change the notation to $v(x;b)$ since Assumption \ref{As:inf_bl} implies time invariance and that $S_t=(X_t,\beta)$.}
 \begin{equation}\label{eq:ccp}
P(a,x,b)=\frac{\exp\left(x^\intercal(b_{a},\;\gamma_{a}^{\intercal})^{\intercal}+\rho\int{v}(x';b)dF_{x}(x'|x,a)\right)}{\sum_{\tilde{a}\in{A}}\exp\left(x^\intercal(b_{\tilde{a}},\;\gamma_{\tilde{a}}^{\intercal})^{\intercal}+\rho\int{v}(x';b)dF_{x}(x'|x,\tilde{a})\right)}.
\end{equation}
Since $\mathrm{Supp}(X_t)$ contains an open set and the analytic continuation of a vanishing function on an open set is vanishing everywhere, the analytic continuation of $x\mapsto F_x(x'|x,a)$ to $\mathbb{R}^{k}$ satisfies $\{x':\exists x\in\mathrm{Supp}(X_t),~dF_x(x'|x,a)>0\}=\{x':\exists x\in\mathbb{R}^{k},~dF_x(x'|x,a)>0\}$. Therefore $P$ in equation \eqref{eq:ccp} is well-defined on $A\times\mathbb{R}^{k}\times S_\beta$.

Since the set ${S}_{\beta}$ is a compact subset of $\mathbb{R}^{J}$ and $|P(a,x,b)|\le1$ for all $(a,x,b)\in A\times\mathbb{R}^{k}\times S_\beta$, 
\begin{equation*}
\|P(a,x,\cdot)\|^{2}_{2}=\int_{{S}_{\beta}}P(a,x,b)^{2}d\lambda(b)\le\int_{{S}_{\beta}}d\lambda(b)<\infty,
\end{equation*}
and thus $b\mapsto P(a,x,b)$ is an element of $\mathcal{L}^{2}_{{S}_{\beta}}$.

To show $x\mapsto P(a,x,b)$ is real analytic, consider that since the sum, composition and ratio of strictly positive real analytic functions are real analytic \parencite{krantz2002primer} it is sufficient to show $x\mapsto\int{v}(x';b)dF(x'|x,a)$ is real analytic. By Assumption I2\ref{As:inf_tran},
\begin{equation*}
\int{v}(x';b)dF(x'|x,a)=\int{v}(x';b)f_{c}(x'|x,a)dx'+\sum_{i=1}^{N}v(i;b)f_{d}(i|x,a)
\end{equation*}
where $f_c(\cdot|x,a)$ is a probability density function and $f_d(\cdot|x,a)$ is a probability mass function with $N$ points of support. Since $f_{d}$ is a real analytic function of $x$, it is sufficient to show $\int{v}(x';b)f_{c}(x'|x,a)dx'$ is real analytic. By assumption I2\ref{As:inf_tran}, $f_{c}(x'|x,a)$ is real analytic on $x\in\mathbb{R}^{k}$. That is, for each fixed $(a,x')$, there is a unique power series representation, such that for all $x\in\mathbb{R}^{k}$,
\begin{equation*}
f_{c}(x'|x,a)=\sum_{n\in\mathbb{N}^{J+1}}\alpha_{n}(a,x')x^{n}.
\end{equation*}
For any $x'$ outside its bounded support and any $a$, since $f_{c}(x'|x,a)=0$ for $x\in\mathrm{Supp}(X_t)$, $f_{c}(x'|x,a)=0$ for ${x}\in\mathbb{R}^{k}$ since $\mathrm{Supp}(X_t)$ contains an open set. We
are now in a position to show the result.
\begin{align*}
\int{v}(x';b)f_{c}(x'|x,a)dx'=&\int{v}(x';b)\sum_{n\in\mathbb{N}^{J+1}}\alpha_{n}(a,x')x^{n}dx'\\
=&\int\sum_{n\in\mathbb{N}^{J+1}}\tilde\alpha_{n}(a,x')x^{n}dx'\\
=&\sum_{n\in\mathbb{N}^{J+1}}\left(\int\tilde\alpha_{n}(a,x')dx'\right)x^{n}=\sum_{n\in\mathbb{N}^{J+1}}\breve\alpha_{n}x^{n}
\end{align*}

The first equality holds by definition. The second holds from defining $\tilde{\alpha}_{n}(a,x')={v}(x';b)\alpha_{n}(a,x')$. The third equality holds from the bounded convergence theorem because, the integral being supported on a bounded set, $\tilde{\alpha}_{n}(a,x')$ is dominated by its supremum taken over its bounded support. The final equality is by definition of $\breve\alpha_{n}=\int\tilde\alpha_{n}(a,x')dx'$, which exists since the defining integral is supported on a bounded set.
\end{proof}

Lemma \ref{lem:SW3.8} is a straightforward generalization of \textcite[Theorem 3.8]{stinchcombe1998consistent} that allows for non-linear kernel functions and the domain of the functions in the image of the integral transform to be a strict subset of the Euclidean space.

\begin{lemma}\label{lem:SW3.8}
Let $F$ be a signed measure with compact support $\mathcal{Y}\subseteq\mathbb{R}^{d_Y}$, and let $\mathcal{X}\subseteq\mathbb{R}^{d_X}$. Suppose $x\mapsto f(x,y)$ is real analytic on $\mathcal{X}$ for each $y\in\mathcal{Y}$, and that
\begin{equation}\label{eq:sw}
\int f(x,y)\,dF(y) = 0 \quad \text{for all } x\in \mathcal{X} \quad \implies \quad F = 0.
\end{equation}
Then for any non-empty open set $T\subseteq\mathcal{X}$, if
\begin{equation*}
\int f(x,y)\,dF(y) = 0 \quad \text{for almost every } x\in T,
\end{equation*}
it follows that (i) $\int f(x,y)\,dF(y) = 0$ for all $x\in T$ and (ii) $F = 0$ (the zero measure).
\end{lemma}

\begin{proof}[Proof of Lemma \ref{lem:SW3.8}]
Suppose that equation (\ref{eq:sw}) holds and that for almost every ${x}\in{T},\int{f}(x,y)dF(y)=0$, for some $T\subseteq\mathbb{R}^{d_X}$ open and non-empty. Since $f$ is real analytic for each $y$ and $\mathcal{Y}$ is bounded, $\int{f}(x,y)dF(y)$ is a real analytic function of $x$ \parencite{mattner1999complex}. A real analytic function that vanishes on a subset of an open set with positive Lebesgue measure must vanish identically on that open set. Then, since $\int{f}(x,y)dF(y)$ is zero on an open set, it is zero on the Euclidean space (e.g., \textcite{krantz2002primer}, Corollary 1.2.6) and by equation (\ref{eq:sw}), $F=0$.
\end{proof}

\subsection{Proof of results in Section \ref{sec:model-intercept}}\label{sec:f_proofs}

\textbf{Notation:} $\tilde{A}=\{1,2,\ldots,J\}$. For a vector $x$, let $x_{[k]}$ denote the $k$th element and $x_{[-k]}$ the vector excluding the $k$th element.

\subsubsection{Proof of Theorem \ref{thm:inject_finite}}
\begin{proof}
Under Assumptions \ref{As:f_bl} and \ref{As:f_puh},
\begin{equation*}
{P}_{T}(a,x,b)=\frac{\exp\left(b_{a[1]}+x^\intercal(b_{a[-1]}^{\intercal},\;\gamma_{T,a}^{\intercal})^{\intercal}\right)}{\sum_{\tilde{a}\in {A}}\exp\left(b_{\tilde{a}[1]}+x^\intercal(b_{\tilde{a}[-1]}^{\intercal},\;\gamma_{T,\tilde{a}}^{\intercal})\right)}.
\end{equation*}
Denote $x=(z^\intercal,w^\intercal)^\intercal$ for $z\in\mathbb{R}^p$ and $w\in\mathbb{R}^{J}$, and observe $(z,w)\mapsto P_T(a,x,b)$ is real analytic. Since $S_\beta$ is compact, Lemma \ref{lem:SW3.8} applies and the result holds if, for any signed measure $\mu$,
\begin{equation*}
\int P_T(a,x,b)d\mu(b)=0 \quad\text{for all }(a,x)\in A \times\mathbb{R}^{k}, ~ \Longrightarrow ~ \mu=0
\end{equation*}
I show this condition directly. Assume $\mu$ is a finite signed measure satisfying
\begin{equation}\label{eq:f_mc_hyp}
\forall(a,{z})\in{A}\times\mathbb{R}^{p},\;  \int{P}_{{T}}(a,x,b)d\mu(b)=0
\end{equation}
for any fixed $w$. Viewed as a function of a $w\in\mathbb{R}^{J}$ this object is infinitely differentiable and since it is identically zero, all of its derivatives are zero. Furthermore, since both $P_{T}$ and $\mu$ are bounded, we can exchange the order of differentiation and integration, so that for any $1\le i\le J$,
\begin{equation*}
\forall{n}\in\mathbb{N}_{+}\; ,\forall(a,z)\in{A}\times\mathbb{R}^{p},\;  \int \frac{\partial^{n}}{\partial{w_{[i]}^{n}}}{P}_{{T}}(a,x,b)d\mu(b)=0.
\end{equation*}
Fix $a$ and consider the first-order partial derivative ($n=1$) with respect to $w_{i}$:
\begin{equation*}
\forall{z}\in\mathbb{R}^{p},\;\gamma_{T,a[i]}\int{P}_{T}(a,x,b)d\mu(b)-\sum_{j\in\tilde{A}}\gamma_{T,j[i]}\int{P}_{T}(a,x,b)P_T(j,x,b)d\mu(b)=0.
\end{equation*}
From equation \eqref{eq:f_mc_hyp}, it follows that,
\begin{equation*}
\forall(a,{z})\in{A}\times\mathbb{R}^{p},\;\sum_{j\in\tilde{A}}\gamma_{T,j[i]}\int{P}_{T}(a,x,b){P}_{T}(j,x,b)d\mu(b)=0.
\end{equation*}
Repeating the argument for all $i\in\tilde{A}$ yields the system of linear equations
\begin{equation*}
\Gamma_{T}^\intercal \int{P}_{T}(a,x,b)\otimes\tilde{P}_{T}^\intercal(x,b)d\mu(b)=0^\intercal_{J}
\end{equation*}
where $\tilde{P}_{T}(x,b)$ is the vector $({P}_{T}(a,x,b)\colon{a}\in\tilde{A})$, $\otimes$ is the Kronecker product and $0_{J}\in\mathbb{R}^{J}$ is the zero vector. By Assumption F2\ref{As:f_homog}, $\Gamma_{T}$ is full rank and thus $\int{P}_{T}(a,x,b)\otimes\tilde{P}^\intercal_{T}(x,b)d\mu(b)=0^\intercal_{J}$. Repeating the argument for each $a$,
\begin{equation*}
\forall{{z}}\in\mathbb{R}^{p},\; \int\tilde{P}_{T}(x,b)^{\alpha}d\mu(b)=0
\end{equation*}
for multi-indices $\alpha\in\{1,2\}^{J}$. Repeating the argument for higher order derivatives,
\begin{equation}\label{eq:four}
\forall{z}\in\mathbb{R}^{p},\; \int\tilde{P}_{T}(x,b)^{\alpha}d\mu(b)=0
\end{equation}
for all $\alpha\in\mathbb{N}^{J}$. Let $\mu_{{z}}$ be the signed measure induced by $\beta\mapsto\tilde{P}_{T}(x,\beta)$, i.e.,
\begin{equation*}
\mu_{z}(B)=\int1\{\tilde{P}_{T}(x,b)\in{B}\}d\mu(b).
\end{equation*}
That is, $\mu_{z}$ is the measure of $\tilde{P}_{T}(x,\beta)$. Thus from equation (\ref{eq:four}),
\begin{equation*}
\forall{{z}}\in\mathbb{R}^{p},\int{x}^{\alpha}d\mu_{{z}}(x)=0
\end{equation*}
for all $\alpha\in\mathbb{N}^{J}$. It follows that the Fourier transform of $\tilde{P}_{T}(x,\beta)$ is identically zero, and thus the measure $\mu_{z}$ is zero for each ${z}\in\mathbb{R}^{p}$ \parencite[Theorem 1 Proof]{hornik1993some}. Since $\tilde{P}_{T}(x,\beta)=\tilde{P}_{T}(x,\tilde\beta)$ implies $\beta_{a[1]}+x^\intercal(\beta_{a[-1]}^{\intercal},\;\gamma_{T,a}^{\intercal})^{\intercal}=\tilde\beta_{a[1]}+x^\intercal(\tilde\beta_{a[-1]}^{\intercal},\;\gamma_{T,a}^{\intercal})^{\intercal}$ for all $a\in A$, $\mu_{{z}}=0$ implies for all $z\in\mathbb{R}^p$,
\begin{equation*}
\int1\{b:\{b_{a[1]}+x^\intercal(b_{a[-1]}^{\intercal},\;\gamma_{T,a}^{\intercal})^{\intercal}:a\in\tilde{A}\}\in{B}\}d\mu(b)=0.
\end{equation*}
From here standard arguments \parencite[Lemma 1]{masten2018random} give that the characteristic function of $\beta$ is zero and thus the signed measure $\mu=0$.
\end{proof}

\subsubsection{Proof of Theorem \ref{thm:f}}
\begin{proof}
Let $Y=((A_{t},X_{t})_{t=2}^{T},A_{1})$. By Assumption \ref{As:f_bl}, the
distribution of $Y$ conditional upon $X_{1}=x$ is
\begin{equation*}
f_{y|x_1}(y,x_{1})=\int\prod_{t=2}^{T}\left(P_{t}(a_{t},x_{t},b)F_{x_{t}}(x_{t}|x_{t-1},a_{t-1})\right)P_{1}(a_{1},x_{1},b)f_{\beta|X_{1}}(b,x_{1})db.
\end{equation*}
Fix $x_1\in\mathrm{Supp}(X_1)$ and $(a_{t})_{t=1}^{T-1}\in A^{T-1}$. By Assumption \ref{As:f_os},
\begin{equation*}
\frac{f_{y|x_1}(y,x_{1})}{\prod_{t=2}^{T}F_{x_{t}}(x_{t}|x_{t-1},a_{t-1})}=\int\prod_{t=1}^{T}P_{t}(a_{t},x_{t},b)f_{\beta|X_{1}}(b,x_{1})db.
\end{equation*}
Let $g(b;(a_{t})_{t=1}^{T-1})=\prod_{t=1}^{T-1}P_{t}(a_{t},x_{t},b)f_{\beta|X_{1}}(b,x_{1})$, and define the operator
\begin{equation*}
{L}_{T,\beta}:L_{S_{\beta}}\rightarrow{A}\times\mathcal{L}_{S_{T}}\qquad[{L}_{T,\beta}m](a_{T},{x}_{T})=\int{P}_{T}(a_{T},x_{T},b)m(b)db.
\end{equation*}
Under Assumption \ref{As:f_bl}-\ref{As:f_os}, Theorem \ref{thm:inject_finite} implies $L_{T,\beta}$ is injective and that the operator defined in \ref{As:f_rank} exists. Suppose $\gamma_{T},\tilde{\gamma}_{T}$ are observationally equivalent, i.e.,
\begin{equation*}
    (x_{T},a_T)\in S_T\times A, \; \int P_T(a_T,x_T,b;\gamma_T)g(b;(a_{t})_{t=1}^{T-1})db = \int P_T(a_T,x_T,b;\tilde\gamma_T)\tilde{g}(b;(a_{t})_{t=1}^{T-1})db.
\end{equation*}
In particular for $E$ as in Assumption \ref{As:f_rank}, $[L_{T,\beta}^{E,\gamma_{T}}g](a_T,x_{T})=[L_{T,\beta}^{E,\tilde\gamma_{T}}\tilde{g}](a_T,x_{T})$ for all $(x_{T},a_T)\in{E}$. Since $L_{T,\beta}^{E,\gamma_{T}}$ is injective, $g(b;(a_{t})_{t=1}^{T-1})=[(L_{T,\beta}^{E,{\gamma}_{T}})^{-1}L_{T,\beta}^{E,\tilde\gamma_{T}}\tilde{g}](b)$. Similarly, by Assumption \ref{As:f_rank}, for some $\tilde{E}$, $g(b;(a_{t})_{t=1}^{T-1})=[(L_{T,\beta}^{\tilde{E},{\gamma}_{T}})^{-1}L_{T,\beta}^{\tilde{E},\tilde\gamma_{T}}\tilde{g}](b)$. It
follows that
\begin{equation*}
0=\left[\left((L_{T,\beta}^{E,{\gamma}_{T}})^{-1}L_{T,\beta}^{E,\tilde\gamma_{T}}-(L_{T,\beta}^{\tilde{E},{\gamma}_{T}})^{-1}L_{T,\beta}^{\tilde{E},\tilde\gamma_{T}}\right)\tilde{g}\right](b),
\end{equation*}
but $\tilde{g}(b;(a_{t})_{t=1}^{T-1})\neq0$. Under Assumption \ref{As:f_rank}, if $\gamma_T\neq\tilde{\gamma}_T$ then $L_{T,\beta}^{E,\gamma_{T},\tilde{E},\tilde{\gamma}_{T}}\equiv(L_{T,\beta}^{E,{\gamma}_{T}})^{-1}L_{T,\beta}^{E,\tilde\gamma_{T}}-(L_{T,\beta}^{\tilde{E},{\gamma}_{T}})^{-1}L_{T,\beta}^{\tilde{E},\tilde\gamma_{T}}$ is injective, so we conclude $\gamma_T=\tilde{\gamma}_T$. Next, $ g(b;(a_{t})_{t=1}^{T-1})$ is identified as
\begin{equation*}
    g(b;(a_{t})_{t=1}^{T-1}) = \left[L_{T,\beta}^{-1}  \frac{f_{y|x_1}(y,x_{1})}{\prod_{t=2}^{T}F_{x_{t}}(x_{t}|x_{t-1},a_{t-1})}\right](b),
\end{equation*}
which is possible since $L_{T,\beta}$ is identified and injective. Repeating this argument for each choice sequence ($a_{t})_{t=1}^{T}$, $f_{\beta|X_{1}}(b,x_{1})$ is identified as $\sum_{{a}\in{A}^{(T-2)}}g(b;{a})$. Similarly, $P_t(a_t,x_t,b)$ is identified as the sum of $g(b,(a_{t})_{t=1}^{T-1}\div f_{\beta|X_{1}}(b,x_{1})$ over the support of $(a_{t})_{t=1}^{T-1}$ for all periods except the $t$th. Finally, given identification of $\gamma_{t+1}, \gamma_{t+2},\ldots \gamma_T$, under Assumption \ref{As:f_os}, $\gamma_t$ is identified.
\end{proof}

\subsubsection{Proof without rank condition}

We consider the more general case that $k \ge p+J$.

\begin{manualassumption}{F2$^{\bm{add}}$}\label{As:f_puh_add}
\begin{enumerate*}[label=(\roman*)]
    \item $S_t=(X_t^\intercal,\beta^\intercal)^\intercal\in\mathbb{R}^{k+(1+p)J}$, and $k \ge p+J\text{ for } p\ge0$. For each $x\in\mathrm{Supp}(X_1)$, $\beta\mid X_1=x$ admits a bounded density $f_{\beta|X_{1}}$. \item Let $\delta_{T,a}$ be the first $J$ elements of $\gamma_{T,a}$. Then $\Gamma_{T}\equiv(\delta_{T,1}\delta_{T,2}\cdots\delta_{T,J})\in\mathbb{R}^{J\times J}$ is full rank.
    \item Assumptions \ref{As:f_puh} (ii) and (iii).
\end{enumerate*} 
\end{manualassumption}

\begin{manualassumption}{F3$^{\bm{add}}$}\label{As:f_os_add}
Let $Z_t$ denote the first $p+J$ elements of $X_T$. For each $x_{1}\in\mathrm{Supp}(X_1)$ and ($a_{1},a_2,\ldots,a_{T-1})\in A^{T-1}$, there is $(x_{2},x_3,\ldots,x_{T-1})\in \times_{t=2}^{T-1} \mathrm{Supp}(X_t)$ such that
\begin{equation*}
\mathrm{Supp}\left(Z_{T} \mid A_{T-1}=a_{t-1},X_{T-1}=x_{t-1},\dots,A_1=a_1,X_1=x_1\right)
\end{equation*}
contains a non-empty open set. Moreover, for each $t$, $\mathrm{Supp}((1,X_t))$ spans $\mathbb{R}^{k+1}$.
\end{manualassumption}

\begin{lemma}[Result without rank condition]\label{thm:f_norank}
Assume the distribution of $(X_{t},A_{t})_{t=1}^{T}$ is observed for $T\ge2$, generated from agents solving the model of equation (\ref{eq:inf_prob}) with $J=1$ satisfying assumptions \ref{As:f_bl}, \ref{As:f_puh_add} and \ref{As:f_os_add}. Furthermore $S_T$ contains no isolated points. If $\gamma_{T[1]}=1$, then $f_{\beta|X_{1}}$ is point identified.
\end{lemma}

\begin{proof}
Proceed as in the proof to Theorem \ref{thm:f}. For identification of $\gamma_T$, suppose for all $x\in{S}_{T}$,  
\begin{equation*}
  \int \Lambda\left(  b_{[1]}+x^\intercal\left(b_{[-1]}^{\intercal},\;\gamma_{T}^{\intercal}\right)^{\intercal}\right){g}(b;a_1)db=  \int \Lambda\left(  b_{[1]}+x^\intercal\left(b_{[-1]}^{\intercal},\;\tilde\gamma_{T}^{\intercal}\right)^{\intercal}\right)\tilde{g}(b;a_1)db.
\end{equation*}
Since $S_T$ contains no isolated points, we can differentiate the above equation with respect to $x\in S_T$. Furthermore, as both $\Lambda$ and $g$ are bounded, the limits defining differentiation and integration may be exchanged, so that for all $x\in {S}_{T}$ and $p< k'\le k$,
\begin{equation*}
  \int \frac{\partial}{\partial{x_{k'}}}\Lambda\left(   b_{[1]}+x^\intercal\left(b_{[-1]}^\intercal,\gamma_{T}^\intercal\right)^\intercal  \right){g}(b;a_1)db= \int \frac{\partial}{\partial{x_{k'}}}\Lambda\left(   b_{[1]}+x^\intercal\left(b_{[-1]}^\intercal,\tilde\gamma_{T}^\intercal\right)^\intercal  \right)\tilde{g}(b;a_1)db.
\end{equation*}
Since the derivative of $\Lambda(x)$ is $\Lambda(x)(1-\Lambda(x))$, the above display is equivalent to
\begin{align*}
&\gamma_{T[k']}\int[\Lambda(1-\Lambda)]\left(   b_{[1]}+x^\intercal\left(b_{[-1]}^\intercal,\gamma_{T}^\intercal\right)^\intercal  \right){g}(b;a_1)db\\&=\tilde\gamma_{T[k']}\int[\Lambda(1-\Lambda)]\left(   b_{[1]}+x^\intercal\left(b_{[-1]}^\intercal,\tilde\gamma_{T}^\intercal\right)^\intercal  \right)\tilde{g}(b;a_1)db.
\end{align*}
By assumption $\gamma_{T[1]}=\tilde{\gamma}_{T[1]}=1$, so for all $x\in {S}_{T}$,
\begin{equation*}
\int[\Lambda(1-\Lambda)]\left(   b_{[1]}+x^\intercal\left(b_{[-1]}^\intercal,\gamma_{T}^\intercal\right)^\intercal  \right){g}(b;a_1)db=\int[\Lambda(1-\Lambda)]\left(   b_{[1]}+x^\intercal\left(b_{[-1]}^\intercal,\tilde\gamma_{T}^\intercal\right)^\intercal  \right)\tilde{g}(b;a_1)db.
\end{equation*}
Therefore, for any $k'$
\begin{equation*}\left(\gamma_{T[k']}-\tilde\gamma_{T[k']}\right)\int[\Lambda(1-\Lambda)]\left(   b_{[1]}+x^\intercal\left(b_{[-1]}^\intercal,\gamma_{T}^\intercal\right)^\intercal  \right){g}(b;a_1)db=0,
\end{equation*}
and since the logistic function takes values in $(0,1)$, $\gamma_{T[k']}=\tilde\gamma_{T[k']}$ and $\gamma_T$ is identified. Given identification of $\gamma_T$, $f_{\beta|X_1}$ is identified by the argument in the proof to Theorem \ref{thm:inject_finite}.
\end{proof}

\newpage
\begin{refsection}
\setcounter{table}{0}
\renewcommand{\thetable}{B\arabic{table}}
\setcounter{figure}{0}
\renewcommand{\thefigure}{B\arabic{figure}}
\section{Online appendix}
\setcounter{page}{1}

Throughout this appendix I use the following notations: $S_\beta=\mathrm{Supp}(\beta)$; $\mathcal{L}_{A}$ denotes the usual $L^\infty$ space $\mathcal{L}^\infty(A,\lambda)$ where $\lambda$ the Lebesgue measure.

\subsection{Additional proofs for Section \ref{sec:extensions}}

\subsubsection{Proof of Corollary \ref{thm:f_nt}}\label{ssec:f_proofs_nt}

Before stating the proof, we introduce the support assumption used in the statement of Corollary \ref{thm:f_nt}.
\begin{manualassumption}{F3${}^\prime$}\label{As:f_nt_os}
For each $x\in\mathrm{Supp}(X_1)$, $\exists a\in A$ such that $\forall a_2,a_3\in A$ (i) $\mathrm{Supp}(X_1)$, $S_2\equiv\mathrm{Supp}(X_2\mid X_{1}=x,A_1 =a)$, $S_3=\mathrm{Supp}(X_{3}\mid{X_2}\in{S_2},A_2=a_2)$ and $\cap_{a_3\in A}\mathrm{Supp}(X_4\mid X_3\in S_3, A_3=a_3)$  span $\mathbb{R}^k$. (ii) $\mathrm{Supp}(X_{3}\mid{X_2}\in{S_2},A_2=a_2)$ and $\mathrm{Supp}(X_{4}\mid X_3\in S_3, A_3=a_3)$ contain a non-empty open set.
\end{manualassumption}

\begin{proof}
Fix $x_1\in\mathrm{Supp}(X_1)$ and denote $S_4=\mathrm{Supp}(X_4\mid X_3\in S_3, A_3=a_3)$ which satisfies Assumption \ref{As:f_nt_os}. The operators $L_{4,2,3}:\mathcal{L}_{S_{{3}}}\rightarrow{A}\times\mathcal{L}_{S_{{4}}}$ and $L_{4,3}:\mathcal{L}_{S_{{3}}}\rightarrow{A}\times\mathcal{L}_{S_{{4}}}$ defined as
\begin{align*}&[L_{4,2,3}m](a_{4},x_{4})=\int\frac{f_{A_{4}A_{3}A_{2}A_{1}X_{4}X_{3}X_{2}|X_{1}}(a_{4},a_3,a_{2},a_{1},x_{4},x_{3},x_{2},x_{1})}{F_{x_{4}}(x_{4}|x_{3},a_3)F_{x_{3}}(x_{3}|x_{2},a_{2})F_{x_{2}}(x_{2}|x_{1},a_{1})}m(x_{3})dx_{3}\\
&[L_{4,3}m](a_{4},x_{4})=\int\sum_{a_{2}\in{A}}\frac{f_{A_{4}A_{3}A_{2}A_{1}X_{4}X_{3}X_{2}|X_{1}}(a_{4},a_3,a_{2},a_{1},x_{4},x_{3},x_{2},x_{1})}{{F_{x_{4}}(x_{4}|x_{3},a_3)F_{x_{3}}(x_{3}|x_{2},a_{2})F_{x_{2}}(x_{2}|x_{1},a_{1})}}m(x_{3})dx_{3}
\end{align*}
are well-defined and observed for $x_{2}\in{S}_{2}$. Define the following operators:
\begin{align*}
L_{4,\beta}:\mathcal{L}_{S_{\beta}}\rightarrow{A}\times\mathcal{L}_{S_{{4}}}\qquad&[L_{4,\beta}m](a_{4},x_{4})=\int{P}_{4}(a_{4},x_{4},b)m(b)db\\
D^{2}_{\beta}:\mathcal{L}_{S_{\beta}}\rightarrow\mathcal{L}_{S_{\beta}}\qquad&[D^{2}_{\beta}m](b)=P_{2}(a_{2},x_{2},b)m(b)\\
D_{\beta}:\mathcal{L}_{S_{\beta}}\rightarrow\mathcal{L}_{S_{\beta}}\qquad&[D_{\beta}m](b)=P_{1}(a_{1},x_{1},b)f_{\beta|X_{1}}(b,x_{1})m(b)\\
L_{\beta,3}:\mathcal{L}_{S_{{3}}}\rightarrow\mathcal{L}_{S_{\beta}}\qquad&[L_{\beta,3}m](b)=\int{P}_{3}(a_3,x_{3},b)m(x_{3})dx_{3}
\end{align*}
It is straightforward to show $L_{4,2,3}=L_{4,\beta}D^{2}_{\beta}D_{\beta}L_{\beta,3,}$ and $L_{4,3}=L_{4,\beta}D_{\beta}L_{\beta,3,}$. We begin by showing injectivity of $L_{4,\beta}$ and $L_{\beta,3}^*$. Notice 
\begin{equation*}
P_t(a_t,x_t,b)=\frac{\exp\left(x^\intercal(b_{a},\;\gamma_{t,a}^{\intercal})^{\intercal}+\rho\int{v_{t+1}}(x';b)dF_{x_{t}}(x'|x,a)\right)}{\sum_{\tilde{a}\in{A}}\exp\left(x^\intercal(b_{\tilde{a}},\;\gamma_{t\tilde{a}}^{\intercal})^{\intercal}+\rho\int{v_{t+1}}(x';b)dF_{x_{t}}(x'|x,\tilde{a})\right)}
\end{equation*}
differs from equation \eqref{eq:ccp} only by the time-dependence of $\gamma_t,v_t$ and $F_{x_t}$. Since Assumption \ref{As:f_nt_puh} places restrictions on $(\gamma_t,F_{x_t})$ that are analagous to restrictions placed by Assumption \ref{As:inf_puh} on $(\gamma,F_{x})$ in the stationary model, injectivity will result from the arguments of Lemmas \ref{lem:inf_ccp} and \ref{lemma_cosapprox}. The arguments of Lemma \ref{lem:inf_ccp} apply directly. The arguments of Lemma \ref{lemma_cosapprox} do not directly apply since in the non-stationary model the value function $v_t$ is defined recursively, so we cannot use the uniform bound on $v_t$ from Lemma \ref{lemma_cosapprox}. To develop the uniform bound on $v_t$, I proceed recursively.

First define $e(a,x)=E[\epsilon_{t,a}|x,a\text{ is optimal strategy}]$. Under Assumption \ref{As:f_bl}, the function $e(a,x)$ is known and bounded \parencite{aguirregabiria2007sequential}. Now consider the terminal value function (i.e., $t=T_1$),
\begin{equation*}
{v}_{T_1}(x;b)=\sum_{a\in{A}}P_{T_1}(a,x,b)\left(x^\intercal\left(\beta_{a},\;\gamma_{T,a}^{\intercal}\right)^{\intercal}+e(a,x)\right),
\end{equation*}
which is bounded because the CCP functions are. For $t<T_1$, suppose that $v_{t+1}$ is finite. Since
\begin{equation*}
{v}_{t}(x;b)=\sum_{a\in{A}}P_{t}(a,x,b)\left(x^\intercal\left(\beta_{a},\;\gamma_{t,a}^{\intercal}\right)^{\intercal}+\rho\int{v}_{t+1}(x';b)dF_{x_t}(x'|x,a)\right),
\end{equation*}
${v}_{t}(x;b)$ is finite also. So for any $t$, ${v}_{t}(x;b)$ is finite for any $(x,b)$ and a uniform bound is given by the supremum over the support. Therefore the remaining steps in Lemma \ref{lemma_cosapprox} go through directly.

The arguments in the proof to Theorem \ref{thm:inf} imply that $L_{4,3,2}=L_{4,\beta}D_{\beta}^{2}D_{\beta}L_{\beta,3}$ and $L_{4,3}=L_{4,\beta}D_{\beta}L_{\beta,3}$, and that the spectral decomposition
\begin{equation*}
L_{4,2,3}L_{4,3}^{-1}=L_{4,\beta}D_{\beta}^{2}L_{4,\beta}^{-1}
\end{equation*}
identifies $P_4(a,x,b)$ and thus $\gamma_{4}$. Exchanging the role of $L_{4,\beta}$ and $L_{3,\beta}$ yields identification of $P_3(a,x,b)$ and thus $\gamma_3$. Given identification of $D_\beta^2$, $\gamma_4$ and $\gamma_3$, $\gamma_2$ is identified under Assumption \ref{As:f_nt_os}. Finally, given $D_\beta=L_{4,\beta}^{-1}L_{4,3}L_{\beta,3}^{-1}$, $f_{\beta|X_1}$ and $P_1(a,x,b)$ (and thus $\gamma_1)$ are identified.
\end{proof}

\subsubsection{Proof of Corollary \ref{thm:fh_fd}}\label{sec:pf_fd}

The result uses the following support condition:
\begin{manualassumption}{F3${}^{\prime\prime}$}\label{As:f_nt_os_fd}
For each $x\in\mathrm{Supp}(X_1)$, $\exists a\in A$ such that (i) $\mathrm{Supp}(X_1)$, $S_2\equiv\mathrm{Supp}(X_2\mid X_{1}=x,A_1 =a)$, $S_3=\cap_{a_2\in A}\mathrm{Supp}(X_{3}\mid{X_2}\in{S_2},A_2=a_2)$ and $\mathrm{Supp}(X_4\mid X_3\in S_3, A_3=1)$  span $\mathbb{R}^k$. (ii) $\cap_{a_2\in A}\mathrm{Supp}(X_{3}\mid{X_2}\in{S_2},A_2=a_2)$ and $\mathrm{Supp}(X_{4}\mid X_3\in S_3, A_3=1)$ contain a non-empty open set.
\end{manualassumption}

\begin{proof}
Define the following operators:
\begin{align*}
L_{3,4,2}:\mathcal{L}_{S_{{2}}}\rightarrow\mathcal{L}_{S_{{3}}}\qquad&[L_{4,3,2}m](x_{3})=\int\frac{f_{A_{4}A_{3}A_{2}A_{1}X_{4}X_{3}X_{2}|X_{1}}(1,1,1,a_{1},x_{4},x_{3},x_{2},x_{1})}{F_{x_{4}}(x_{4}|x_{3},1)F_{x_{3}}(x_{3}|x_{2},1)F_{x_{1}}(x_{2}|x_{1},a_{1} ) }m(x_{2})dx_{2}\\
L_{3,2}:\mathcal{L}_{S_{{2}}}\rightarrow\mathcal{L}_{S_{{3}}}\qquad&[L_{4,3}m](x_{3})=\int\sum_{a_{2}\in{A}}\frac{f_{A_{4}A_{3}A_{2}A_{1}X_{4}X_{3}X_{2}|X_{1}}(1,1,a_{2},a_{1},x_{4},x_{3},x_{2},x_{1})}{F_{x_{4}}(x_{4}|x_{3},1)F_{x_{3}}(x_{3}|x_{2},1)F_{x_{1}}(x_{2}|x_{1},a_{1} ) }m(x_{2})dx_{2}\\
L_{3,\beta}:\mathcal{L}_{S_{\beta}}\rightarrow\mathcal{L}_{S_{{3}}}\qquad&[L_{3,\beta}m](x_{3})=\int{P}_{3}(1,x_{3},b)m(b)db\\
D^{4}_{\beta}:\mathcal{L}_{S_{\beta}}\rightarrow\mathcal{L}_{S_{\beta}}\qquad&[D^{4}_{\beta}m](b)=P_{4}(1,x_{4},b)m(b)\\
D_{\beta}:\mathcal{L}_{S_{\beta}}\rightarrow\mathcal{L}_{S_{\beta}}\qquad&[D_{\beta}m](b)=P_{1}(a_{1},x_{1},b)f_{\beta|X_{1}}(b,x_{1})m(b)\\
L_{\beta,2}:\mathcal{L}_{S_{{2}}}\rightarrow\mathcal{L}_{S_{\beta}}\qquad&[L_{\beta,2}m](b)=\int{P}_{2}(1,x_{2},b)m(x_{2})dx_{2}
\end{align*}

Under Assumptions \ref{As:f_bl} and \ref{As:f_nt_os_fd} these operators are well-defined and observed. One can show $L_{4,3,2}=L_{3,\beta}D_{\beta}^{4}D_{\beta}L_{\beta,2}$ and $L_{3,2}=L_{3,\beta}D_{\beta}L_{\beta,2}$. Under Assumptions \ref{As:f_bl}, \ref{As:f_nt_puh} and \ref{As:f_nt_os_fd}, $L_{3,\beta}$ and $L_{\beta,2}^*$ are injective and thus the observed operator $L_{4,3,2}L_{3,2}^{-1}$ has the eigendecomposition $L_{3,\beta}D_{\beta}^{4}L_{3,\beta}^{-1}$. I now show the eigenvalue-eigenfunction representation is unique. Since the model is binary choice with real valued $\beta$, the function $P_{4}(1,x_{4},b)$ is injective in $b$. It follows that the eigenvalues are unique, and, up to the ordering function $R$, $P_{4}(1,x_{4},R(b))$ is identified. The eigenfunctions of the decomposition identify $P_{3}(1,x_{3},R(b))$, which equal
\begin{equation*}
\Lambda\left(x_{3}^\intercal(R(b),\gamma_{3}^{\intercal})^{\intercal}+\int{v_4(x';R(b))}\left(dF_{x_{3}}(x'|x_{3},1)-dF_{x_{3}}(x'|x_{3},0)\right)\right).
\end{equation*}
Under Assumption \ref{ass:fh_fd}, $v_4(x';R(b))$ can be expressed in terms of $P_{4}(1,x_{4},R(b))$, and is therefore identified. Therefore identification consists of showing that $(R(b),\gamma_{3})$ can be identified from $x_{3}^\intercal(R(b),\gamma_{3})$, which follows from the support assumption.  Given identification of $P_4(a,x,b)$, identification of $\gamma$ and $f_{\beta|X_1}$ are attained under Assumption \ref{As:f_nt_os_fd} by a sequential argument as in Corollary \ref{thm:f_nt}.
\end{proof}

\subsubsection{Proof of Corollary \ref{thm:inf_fe}}\label{pf:inf_fe}

\begin{proof}
The proof follows closely the structure of the proof to Theorem \ref{thm:inf}. As in that proof, Assumptions \ref{As:inf_bl} and \ref{As:inf_os_fe} enable the decompositions $L_{3,4,2}=L_{3,\beta}D_{\beta}^{4}D_{\beta}L_{\beta,2}$ and $L_{3,2}=L_{3,\beta}D_{\beta}L_{\beta,2}$ where the operators are defined in proof to Theorem \ref{thm:inf}. I first show injectivity of $L_{3,\beta}$ and $L_{\beta,2}^{*}$. By Assumption \ref{As:inf_os_fe}, for $t=2,3$, the conditional supports of $X_{t}$ contains a non-empty open set for which 
\begin{equation*}
P(a,x,b)=\frac{\exp\left(\beta_{a}+x^{\intercal}\gamma_{a}\right)}{\sum_{\tilde{a}\in{A}}\exp\left(\beta_{\tilde{a}}+x^{\intercal}\gamma_{\tilde{a}}\right)}.
\end{equation*}
Given this functional form, the arguments of Theorem \ref{thm:inject_finite} give that
\begin{equation*}
\int\tilde{P}(x;b)^{\alpha}d\mu(b)=0
\end{equation*}
for all multi-indices $\alpha\in\mathbb{N}^{J}$ where $\tilde{P}(x;b)=\{P(a,x,b)\colon{a}=1,2\ldots,J\}$. It follows that the measure induced by the mapping $\beta\rightarrow\tilde{P}(x;\beta)$ is identically zero. Because this mapping is injective, the measure $\mu(b)$ is identically zero and thus $L_{3,\beta}$ and $L_{2,\beta}^{*}$ are injective. Then, under Assumption \ref{As:inf_os_fe}, identification follows from the proof to Theorem \ref{thm:inf}.
\end{proof}

\subsubsection{Proof of Corollary \ref{thm:kwon}}\label{pf:kwon}
\begin{proof}
From the definitions in the proof to Theorem \ref{thm:inf} and Corollary \ref{thm:f_nt}, it is immediate that $L=L_{3,\beta}D_{\beta}L_{\beta,2}$. By assumption, $D_\beta$ has rank $R$. We now argue that $L_{3,\beta}$ and $L_{\beta,2}^*$ are injective and therefore have rank $R$. Given that $\beta$ has $R<\infty$ points of support, $L^*_{\beta,2}:\mathbb{R}^{R}\rightarrow\mathcal{L}_{S_2}$. From the approximation result in Theorem \ref{thm:inject}, for each $r$, a sequence with elements ${x}_{r,n}\in\mathbb{R}^k$ can be found such that $\lim{P}(0,{x}_{r,n},b_{r_{+}})=1$ for $r_{+}\ge{r}$ and $\lim{P}(0,{x}_{r,n},b_{r_{-}})=0$ for $r_{-}<r$. Define a sequence of $R\times R$ matrices whose $r$th row is $\tilde{P}(x_{r,n})\equiv(P(0,x_{r,n},b_{\tilde{r}})\colon{{\tilde{r}}=1,\ldots,R})$. Since the limit of the sequence of matrices is full rank, for any $m\in\mathbb{R}^R$, for $n$ large enough $\tilde{P}(x_{r,n})^\intercal m=0$ for all $r=1,\ldots,R$ implies $m=0$. We conclude $L_{3,\beta}$ and $L_{\beta,2}^*$ are injective. The result then follows from \textcite{kwon2019estimation}, p. 32.
\end{proof}

\subsection{Appendix to Section \ref{sec:estimation}}

\subsubsection{ Theorem \ref{thm:est}}\label{ssec:est}

This section details the assumptions of Theorem \ref{thm:est} that provide for consistent estimation of $\theta_{0}=(F_{x},\gamma,F_{\beta|X_{1}})\in\Theta=\mathcal{F}\times\Gamma\times\mathcal{M}$ where $\mathcal{F}$ is the space of state transitions, $\Gamma\subseteq\mathbb{R}^{\dim\gamma}$, and $\mathcal{M}=\{F:S_\beta\times\mathrm{Supp}(X_1)\rightarrow [0,1]~\colon~ b\mapsto F(b,x)\text{ is c\`{a}dl\`{a}g}\}$. The first assumption supposes the existence of a consistent estimator for the state transition $F_{x}$\footnote{With some abuse of notation, we allow $F_{x}$ to be either the time-invariant state transition, or the set of time-varying state transitions $F_{x_t}:t=2,\ldots$, and the marginal distribution of the initial observed state $X_{1}$.}:
\begin{manualassumption}{E1}\label{As:est_fs}
There exists an estimator $\hat{F}_{x,n}$ that satisfies $\left|\left|\hat{F}_{x,n}-F_{x}\right|\right|_{\mathcal{F}}=o_{p}(1)$, where $\|\cdot\|_{\mathcal{F}}$ is a norm on $\mathcal{F}$.
\end{manualassumption}

One such estimator that satisfies Assumption \ref{As:est_fs} is the kernel estimator of the conditional density, for any $t>1$
\begin{equation}\label{Eq:cdke}
\hat{F}_{x_{t},n}(x'|x,a)=\frac{\sum_{i=1}^{N}K_{X',h_{X'}}(x'-x_{i,t})K_{X,h_{X}}(x-x_{i,t-1})1\{a_{i,t-1}=a\}}{\sum_{i=1}^{N}K_{X,h_{X}}(x-x_{i,t-1})1\{a_{i,t-1}=a\}}
\end{equation}
where $K_{Z,h_{Z}}$ are multivariate kernel functions with bandwidth $h_{Z}$. Let $\mathcal{M}_{n}$ be a sieve space that approximates $\mathcal{M}$, and denote $d_{\mathcal{M}}(\cdot,\cdot)$ as the Prokhorov metric. The Prokhorov distance between two measures $f,\tilde{f}$ on $S_{\beta}$ is
\begin{equation*}
\inf\left\{\delta>0\colon\forall{B}\in\mathcal{B}(S_{\beta}),\;({f(B)}\le\tilde{f}(B_{\delta})+\delta)\vee({\tilde{f}(B)}\le{f}(B_{\delta})+\delta)\right\},
\end{equation*}
where $B_{\delta}$ is the $\delta$ neighborhood of $B\subseteq{S}_{\beta}$ and $\mathcal{B}(S_{\beta})$ is the Borel sigma field. Let $Y=(A_{t},X_{t})_{t=1}^{T}$. The next assumption requires that the true parameter is a well-separated maximum.
\begin{manualassumption}{E2}\label{As:est_ws}
For all $\epsilon>0$ there exists some decreasing sequence of positive numbers $c_{n}(\epsilon)$ satisfying $\lim\inf{c}_{n}(\epsilon)>0$ such that
\begin{equation*}
E[\psi(Y,F_{x},\gamma,F_{\beta|X_{1}})]-\sup_{\{(\tilde{\gamma},\tilde{f})\in\Gamma\times\mathcal{M}_{n}:\|\tilde{\gamma}-\gamma\|+d_{\mathcal{M}}(\tilde{F},{F}_{\beta|X_{1}})\ge\epsilon\}}E[\psi(Y,F_{x},\tilde{\gamma},\tilde{F})]\ge{c_{n}(\epsilon)}.
\end{equation*}
\end{manualassumption}
Assumption  \ref{As:est_ws} is the condition of Remark 3.1(2) in \textcite{chen2007large} that strengthens their Condition 3.1. If the strict inequality restriction on $c_{n}$ were replaced by a weak inequality, then the assumption would be implied by the identification result.

\begin{manualassumption}{E3}\label{As:est_ss}
The sieve space (i) satisfies
$\mathcal{M}_{n}\subseteq\mathcal{M}_{n+1}\subseteq\mathcal{M}$ and (ii) is such
that there exists a sequence $F_{n}\in\mathcal{M}_{n}$ that converges to
$F_{\beta|X_{1}}$ and satisfies
\begin{equation*}
\left|E[\psi(Y,F_{x},\gamma,F_{n})]-E[\psi(Y,F_{x},\gamma,F_{\beta|X_{1}})]\right|=o(1).
\end{equation*}
\end{manualassumption}
These are standard restrictions on the sieve space and the population criterion function \parencite[Condition 3.2, 3.3(ii)]{chen2007large}. The second condition is a local continuity assumption. As per \textcite[Remark 2.1]{chen2007large}, it is implied by compactness of the sieve space and continuity of the population criterion function on $\mathcal{M}_{n}$.

Define $\mathcal{F}_{n}$ to be the set of possible values that the estimator $\hat{F}_{x,n}$ can take. For example, if the conditional density kernel estimator is chosen, then an element of the set $\mathcal{F}_{n}$ takes the form in equation \eqref{Eq:cdke} and the set $\mathcal{F}_{n}$ is defined by ranging $(X_{t+1},X_{t},A_{t})$ over its support. Define the neighborhood $\mathcal{N}_{F_{x},n}=\{\tilde{F}_{x}\in\mathcal{F}_{n}\colon\|\tilde{F}_{x}-F_{x}\|_{\cal{F}}\le\epsilon_{1,n}\}$ where $\|\cdot\|_{\mathcal{F}}$ is the norm in Assumption \ref{As:est_fs}.
\begin{manualassumption}{E4}\label{As:est_uc}
The following two conditions hold
\begin{equation*}
\sup_{(\tilde{F}_{x},\tilde{\gamma},\tilde{F})\in\mathcal{N}_{F_{x},n}\times\Gamma\times\mathcal{M}_{n}}\left|\frac{1}{n}\sum_{i=1}^{n}\psi(y_{i},\tilde{F}_{x},\tilde{\gamma},\tilde{F})-E[\psi(Y,\tilde{F}_{x},\tilde{\gamma},\tilde{F})]\right|=o_{p}(1),
\end{equation*}
\begin{equation*}
\sup_{(\tilde{F}_{x},\tilde{\gamma},\tilde{F})\in\mathcal{N}_{F_{x},n}\times\Gamma\times\mathcal{M}_{n}}\left|E[\psi(Y,\tilde{F}_{x},\tilde{\gamma},\tilde{F})]-E[\psi(Y,{F}_{x},\tilde{\gamma},\tilde{F})]\right|=o(1).
\end{equation*}
\end{manualassumption}
This is similar to \textcite[Assumption 5.3]{hahn2018nonparametric}, which is based on \textcite[Condition 3.5]{chen2007large} but includes an additional condition to account for the presence of a first-step estimator.

Theorem \ref{thm:est} is a direct consequence of \textcite[Theorem 5.1]{hahn2018nonparametric}, so the proof is omitted. In the proof, by consistency it is meant that $\|\hat{\gamma}-\gamma\|+d_{\mathcal{M}}(\hat{F}_{\beta|X_{1}},F_{\beta|X_{1}})=o_{p}(1)$.


\subsubsection{Theorem \ref{thm:est_fomss}}\label{ssec:fge}

The choice of tuning parameters must satisfy the following condition:

\begin{manualassumption}{E3${}^\prime$}\label{As:est_fomss}
$\mathcal{M}_n$ defined in equation \eqref{eq:fomss} is such that (i) $\mathcal{M}_{n}\subseteq\mathcal{M}_{n+1}$ and as $n\rightarrow\infty$, (ii) $\mathcal{B}_{n}\times\mathcal{X}_{n}$ becomes dense in $S_{\beta}\times\mathrm{Supp}(X_1)$ and (iii) $I(n)\log{I}(n)=o(n)$ for $I(n)=B(n)X(n)$.
\end{manualassumption}
We also place some restrictions on the complexity of $\mathcal{N}_{F_{x},n}$, the neighborhood to which the estimator $\hat{F}_{x,n}$ belongs with probability approaching one. For this purpose define ${N}(w,\mathcal{G},\|\cdot\|_{\mathcal{G}})$ as the covering number of set $\mathcal{G}$ with balls of radius $w$ under the norm $\|\cdot\|_{\mathcal{G}}$.
\begin{manualassumption}{E4${}^\prime$}\label{As:est_compl}
(i)  $(\mathcal{N}_{F_{x},n},\|\cdot\|_{\mathcal{F}})$ and $\Gamma$ are compact. (ii) $P_{t}$ is Lipschitz continuous in $\gamma\in\Gamma$ and continuous in $F_{x}\in\mathcal{N}_{F_{x},n}$. (iii) $\log{N}(w/\sqrt{I(n)},\mathcal{N}_{F_x,n},\|\cdot\|_{\mathcal{F}})=o(n)$ with $I(n)$ as in Assumption \ref{As:est_fomss}.
\end{manualassumption}

\begin{proof}[Proof of Theorem \ref{thm:est_fomss}]
The proof consists of verifying the assumptions of Theorem \ref{thm:est_fomss} imply those of Theorem \ref{thm:est}. Assumption \ref{As:est_fs} is assumed. To verify assumption \ref{As:est_ws}, suppose that (i) $\mathcal{M}_{n}$ and $\mathcal{M}$ are compact in the weak topology and (ii) that $E[(Y,{F}_{x},\gamma,F_{\beta|X_{1}})]$ is continuous in $F_{\beta|X_{1}}\in\mathcal{M}\supset\mathcal{M}_{n}$ in the weak topology and $\gamma\in\Gamma$. Then, since $\theta_{0}=(\gamma,F_{\beta|X_1},F_x)$ is identified, for any $(\tilde\gamma,\tilde{F}_{\beta|X_{1}})\neq(\gamma,F_{\beta|X_{1}})$,
\begin{equation*}
E[\psi(Y,{F}_{x},\gamma,F_{\beta|X_{1}})]-E[\psi(Y,{F}_{x},\tilde\gamma,\tilde{F}_{\beta|X_{1}})]>0
\end{equation*}
Because $\{(\tilde{\gamma},\tilde{F})\in\Gamma\times\mathcal{M}_{n}:\|\tilde{\gamma}-\gamma\|+d_{\mathcal{M}}(\tilde{F},{F}_{\beta|X_{1}})\ge\epsilon\}$ is closed in the compact set $\mathcal{M}_{n}\times\Gamma$, it is compact and the infinum
\begin{equation*}
E[\psi(Y,{F}_{x},\gamma,F_{\beta|X_{1}})]-\sup_{\{(\tilde{\gamma},\tilde{F})\in\Gamma\times\mathcal{M}_{n}:\|\tilde{\gamma}-\gamma\|+d_{\mathcal{M}}(\tilde{F},{F}_{\beta|X_{1}})\ge\epsilon\}}E[\psi(Y,{F}_{x},\tilde\gamma,\tilde{F})]
\end{equation*}
is attained for each $(\epsilon,n)$. Set this difference to $c_{n}(\epsilon)$. It remains to show that $\lim\inf{c}_{n}(\epsilon)>0$. Consider that
\begin{align*}
  c_{n}(\epsilon)=&E[\psi(Y,{F}_{x},\gamma,F_{\beta|X_{1}})]-\sup_{\{(\tilde{\gamma},\tilde{F})\in\Gamma\times\mathcal{M}_{n}:\|\tilde{\gamma}-\gamma\|+d_{\mathcal{M}}(\tilde{F},{F}_{\beta|X_{1}})\ge\epsilon\}}E[\psi(Y,{F}_{x},\tilde\gamma,\tilde{F}_{\beta|X_{1}})]\\\ge&
E[\psi(Y,{F}_{x},\gamma,F_{\beta|X_{1}})]-\sup_{\{(\tilde{\gamma},\tilde{F})\in\Gamma\times\mathcal{M}:\|\tilde{\gamma}-\gamma\|+d_{\mathcal{M}}(\tilde{F},{F}_{\beta|X_{1}})\ge\epsilon\}}E[\psi(Y,{F}_{x},\tilde\gamma,\tilde{F}_{\beta|X_{1}})]\\>&0
\end{align*}
The weak inequality is because $\mathcal{M}_{n}\subseteq\mathcal{M}$. The strict inequality is because the set ${\{(\tilde{\gamma},\tilde{F})\in\Gamma\times\mathcal{M}_{n}:\|\tilde{\gamma}-\gamma\|+d_{\mathcal{M}}(\tilde{F},{F}_{\beta|X_{1}})\ge\epsilon\}}$ is compact and $E[(Y,{F}_{x},\gamma,F_{\beta|X_{1}})]$ is continuous. Since $c_{n}(\epsilon)$ is bounded away from zero uniformly in $n$, its limit inferior is strictly positive.

To complete the argument, it must be shown that (i) $\mathcal{M}_{n}$ and $\mathcal{M}$ are compact in the weak topology and (ii) that $E[\psi(Y,{F}_{x},\gamma,F_{\beta|X_{1}})]$ is continuous on $\mathcal{M}\supset\mathcal{M}_{n}$ in the weak topology and $\gamma\in\Gamma$. Compactness of $\mathcal{M}$ and $\mathcal{M}_{n}$ in the weak topology is shown in \textcite[pp. 240, 247]{fox2016simple}. Since the CCP functions $P_{t}$ are continuous in $(b,\gamma)$ \parencite{norets2010continuity}, the argument of \textcite[Remark 2]{fox2016simple} implies the function $F_{\beta|X_{1}}\mapsto\int\prod_{t=1}^{t_{1}}P_{t}(a_{t},x_{},b;F_{x},\gamma)dF_{\beta|X_{1}}(b,x_{1})$ is continuous. Since it is bounded away from zero, $F_{\beta|X_{1}}\mapsto\log\int\prod_{t=1}^{t_{1}}P_{t}(a_{t},x_{t},b;F_{x},\gamma)dF_{\beta|X_{1}}(b,x_{1})$ is also continuous. And since this function is bounded away from negative infinity, $F_{\beta|X_{1}}\mapsto{E}[\log\int\prod_{t=1}^{t_{1}}P_{t}(A_{t},X_{t},b;F_{x},\gamma)dF_{\beta|X_{1}}(b,X_{1})]$
is continuous by the bounded convergence theorem.

Assumption \ref{As:est_ss}(i) is guaranteed by Assumption \ref{As:est_fomss}(i). For Assumption \ref{As:est_ss}(ii), \textcite[p. 247]{fox2016simple} show the existence of a sequence $(F_n)_{n\in\mathbb{N}}\subseteq\mathcal{M}$ that converges to $F_{\beta|X_{1}}\in\mathcal{M}$.  Since the sequence $(F_n)_{n\in\mathbb{N}}$ takes values in $\mathcal{M}$ and $E[\psi(Y,{F}_{x},\gamma,F_{\beta|X_{1}})]$ is continuous on $\mathcal{M}$, we have that
\begin{equation*}
\left|E[\psi(Y,{F}_{x},\gamma,F_{n})]-E[\psi(Y,{F}_{x},\gamma,F_{\beta|X_{1}})]\right|=o(1).
\end{equation*}

For the first part of Assumption \ref{As:est_uc}, note that 
\begin{align*}
\left|E[\psi(Y,{F}_{x},{\gamma},F_{\beta|X_{1}})]\right|&\le{E}[\left|\psi(Y,{F}_{x},{\gamma},{F}_{\beta|X_{1}})\right|]\\&=E\left[\left|\log\int\prod_{t=1}^{T}P_{t}(A_{t},X_{t},b;F_{x},\gamma)dF_{\beta|X_{1}}(b,x_{1})\right|\right]<\infty,
\end{align*}
because $P_{t}$ is uniformly bounded away from zero since $\mathcal{N}_{F_x,n}\times\Gamma\times{S_{\beta}}$ is compact and $P_{t}$ is strictly positive for each $(b,F_{x},\gamma)$. Then by \parencite[p. 5592]{chen2007large}, $\log{N}(w,\{\psi(\cdot,{F}_{x},{\gamma},F_{\beta|X_{1}})\colon({F}_{x},{\gamma},F_{\beta|X_{1}})\in\mathcal{N}_{F_x,n}\times\Gamma\times\mathcal{M}_{n}\},\|\cdot\|_{1})=o_{p}(n)$ implies the first part of Assumption \ref{As:est_uc}. This entropy is bounded above by the sum of the entropies associated with $\mathcal{N}_{F_{x},n}$, $\Gamma$ and $\mathcal{M}_{n}$. \textcite[p. 248]{fox2016simple} show the entropies associated with $\Gamma$ and $\mathcal{M}_{n}$ are $o_{p}(n)$ under Assumption
\ref{As:est_fomss}(iii). By Assumption \ref{As:est_compl}(iii), the entropy
associated with $\mathcal{N}_{F_{x},n}$ is $o_p(n)$.
The second part of Assumption \ref{As:est_uc} follows easily from the continuity of the population
criterion function on the compact set $\mathcal{N}_{F_{x},n}\times\Gamma\times\mathcal{M}_{n}$.
\end{proof}

\subsection{Appendix to Section \ref{sec:sims}}\label{sec:sims_addl}

This subsection contains several additional simulation results. First, Figures \ref{fig:sim1}-\ref{fig:sim3}, contain the empirical quantiles for the estimator of $F_{\beta}$ for each of DGP 1, DGP 2 and DGP 3. For each sample size the median estimate (the black curve) falls close to the true distribution (the blue curve). The empirical pointwise confidence bands are substantially narrower for the larger sample sizes.

\begin{figure}[H] \centering \includegraphics[width=\textwidth]{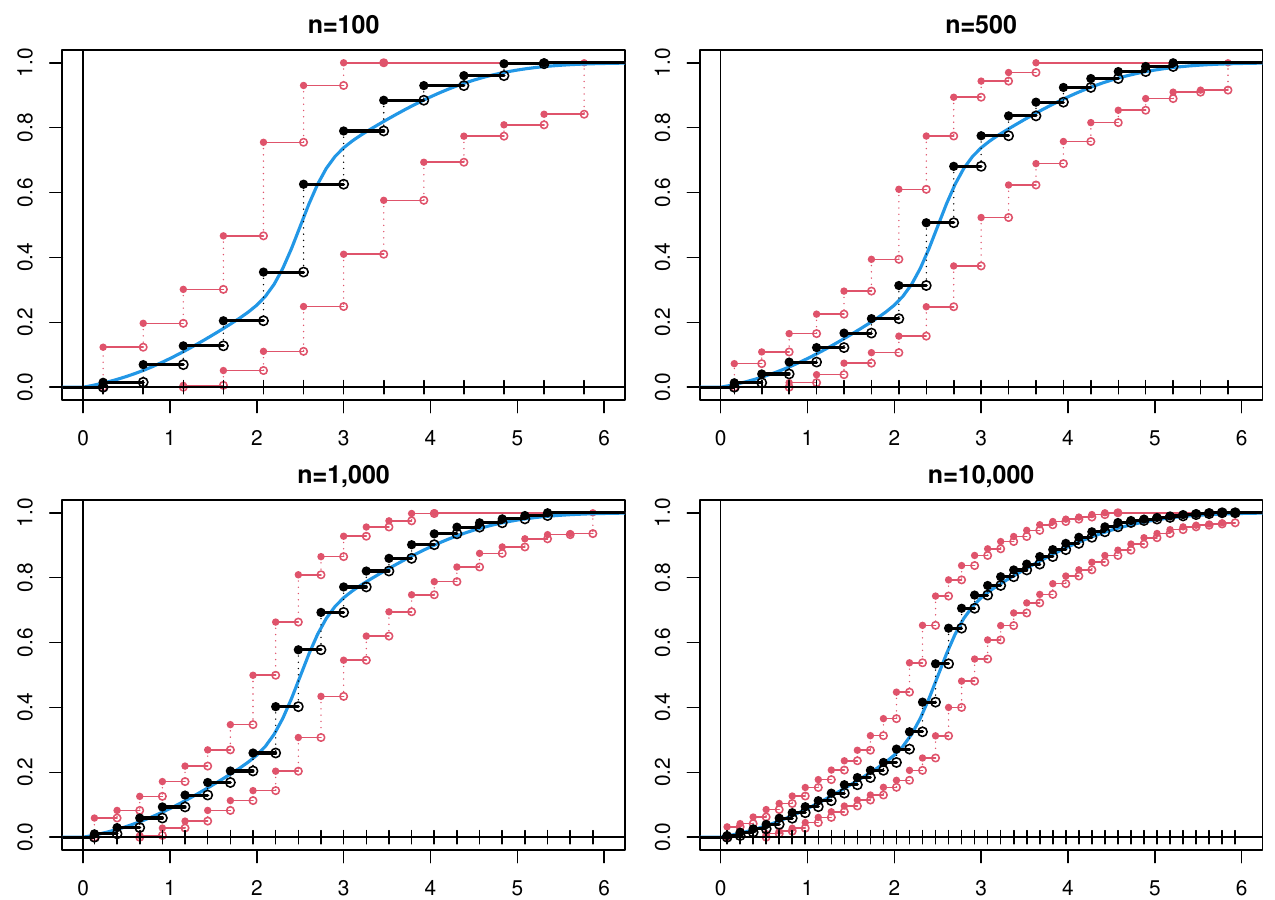}
\vspace*{-10mm} \caption{\small{Simulation results for estimation of $F_{\beta}$ for each sample size in DGP 1. The black curve represents the median estimate, the red curves pointwise  97.5\%, 2.5\% quantiles, and the blue curve the true distribution. The ticks on the x-axis represent the grid points.}}\label{fig:sim1}\vspace*{-0mm}
\end{figure}

\begin{figure}[H] \centering \includegraphics[width=\textwidth]{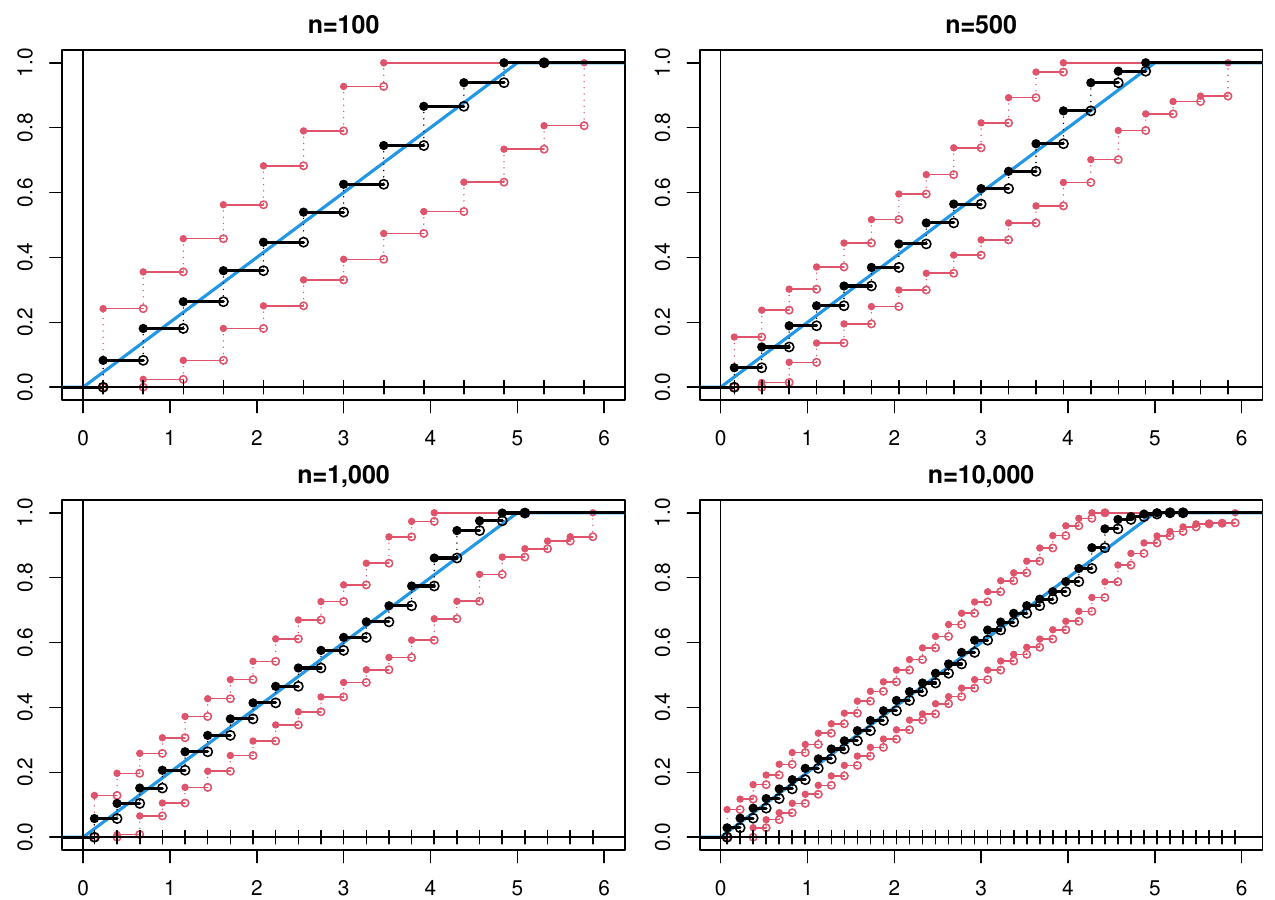}
\vspace*{-10mm} \caption{\small{Simulation results for estimation of $F_{\beta}$ for each sample size in DGP 2. The black curve represents the median estimate, the red curves pointwise  97.5\%, 2.5\% quantiles, and the blue curve the true distribution. The ticks on the x-axis represent the grid points.}}\label{fig:sim2}\vspace*{-0mm}
\end{figure}

\begin{figure}[H] \centering \includegraphics[width=\textwidth]{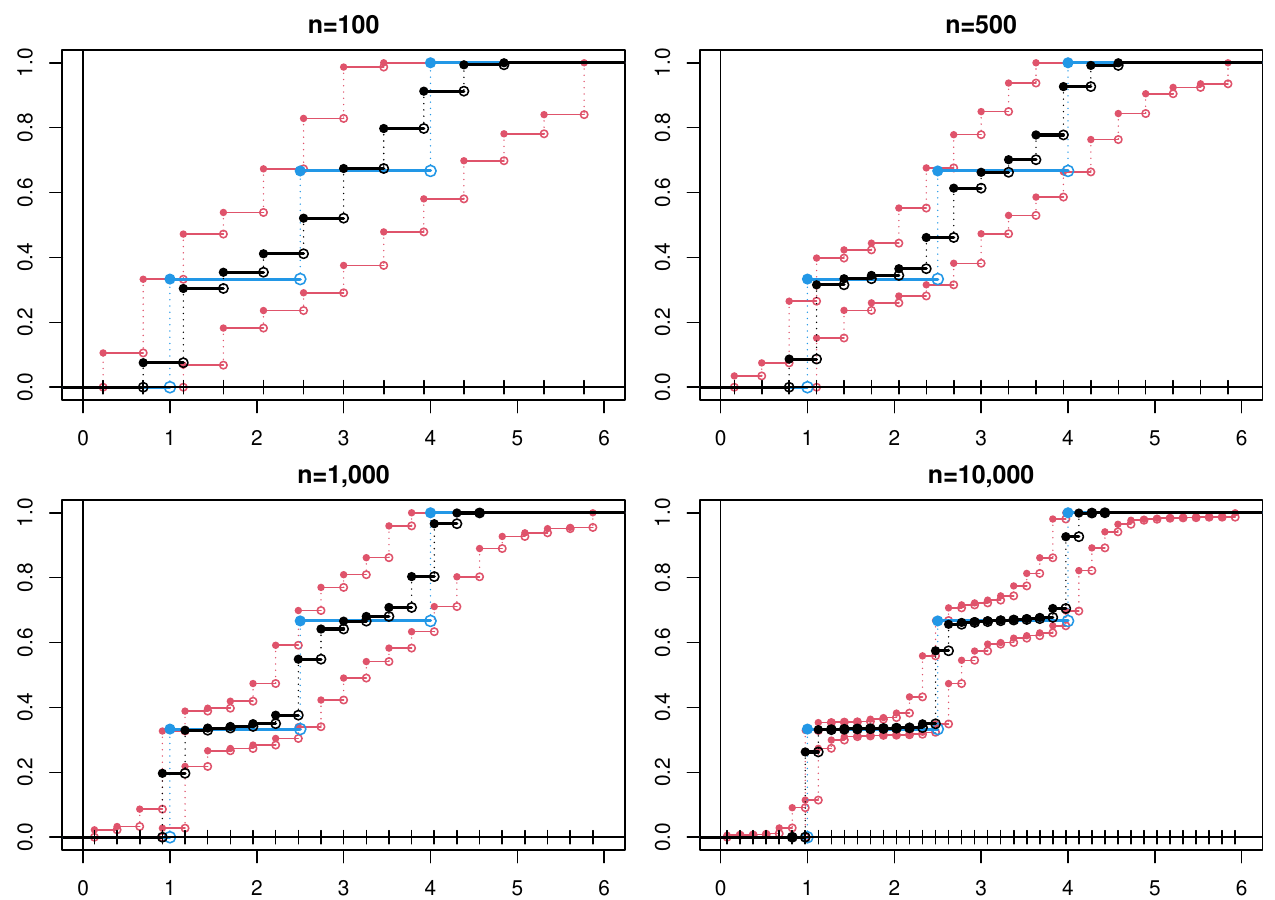}
\vspace*{-10mm} \caption{\small{Simulation results for estimation of $F_{\beta}$ for each sample size in DGP 3. The black curve represents the median estimate, the red curves pointwise  97.5\%, 2.5\% quantiles, and the blue curve the true distribution. The ticks on the x-axis represent the grid points.}}\label{fig:sim3}\vspace*{-0mm}
\end{figure}

Finally, Table \ref{table:res_bs_coverage} contains empirical coverage probabilities of the pointwise bootstrap confidence intervals for $\gamma$ and the c.d.f. of $\beta$ at each decile ($q_{d}\equiv F_\beta^{-1}(d)$ for $d=0.1,\ldots,0.9$), evaluated for the sample size $n=100,500, 1{,}000$. For the largest sample size ($n=1{,}000$), the minimum, median and maximum coverage probabilities of $F_\beta$ over the 9 evaluation points are 0.85, 0.89 and 0.93 respectively, and the standard deviation is 0.024. The empirical coverage probabilities are computed as follows. First, for each draw $m=1,\ldots,100$ of size $n$ from DGP 1, the estimator is computed on 100 bootstrap samples of size $n$, generated by sampling $i=1,2,\ldots,n$ uniformly with replacement. Then, a 90\% confidence interval $CI_{m,n}(\xi)$ for $\xi=(F_\beta(q_{0.1}),\ldots,F_\beta(q_{0.9}),\gamma)$ is computed as the interval between the 0.05 and 0.95 percentiles of the 100 bootstrap estimates of $\xi$. Finally, the empirical coverage probabilities are computed as the average number of times that the bootstrapped confidence interval $CI_{m,n}(\xi)$ contains the true parameter $\xi$.

\begin{table}[ht]
\centering
\begin{tabular}{c|ccc}
  \hline   \hline   \multicolumn{1}{l|}{Parameter}   & $n=100$ & $n=500$ & $n=1{,}000$  \\\hline  \hline  $F_\beta(q_{0.1})$ & 0.78 & 0.83 & 0.89 \\ 
  $F_\beta(q_{0.2})$ & 0.92 & 0.76 & 0.91 \\ 
  $F_\beta(q_{0.3})$ & 0.92 & 0.87 & 0.88 \\ 
  $F_\beta(q_{0.4})$ & 0.87 & 0.81 & 0.85 \\ 
  $F_\beta(q_{0.5})$ & 0.80 & 0.90 & 0.86 \\ 
  $F_\beta(q_{0.6})$ & 0.95 & 0.86 & 0.90 \\ 
  $F_\beta(q_{0.7})$ & 0.95 & 0.93 & 0.93 \\ 
  $F_\beta(q_{0.8})$ & 0.92 & 0.87 & 0.89 \\ 
  $F_\beta(q_{0.9})$ & 0.94 & 0.86 & 0.90 \\ 
   \hline
$\gamma$ & 0.89 & 0.90 & 0.89 \\ 
   \hline
\hline
\end{tabular}
\caption{Empirical coverage probabilities of the bootstrap $90\%$ confidence interval for different model parameters and sample sizes. For each $n=100,500, 1{,}000$ and parameter $\xi=F_\beta(q_{0.1}),\ldots,F_\beta(q_{0.9}),\gamma$ (where $q_{d}\equiv F_\beta^{-1}(d)$), the coverage probability is computed as $\sum_{m=1}^{100} 1\{\xi \in CI_{m,n}(\xi)\}/100$,  where $CI_{m,n}(\xi)$ is the 90\% bootstrap confidence interval for $\xi$ evaluated on the $m$th draw of sample size $n$ from DGP 1.}

\label{table:res_bs_coverage}
\end{table}

\subsection{Appendix to Sections \ref{sec:applic}}\label{sec:emp_app}

\subsubsection{Data construction}\label{app:data_construct}

The model is estimated using a subset of data from the Panel Study of Income Dynamics \parencite{PSID_Public_Use} from survey years 1973 to 1986. Our subset of wives with working husbands is constructed following the description in \textcite{altuug1998effect}, Appendix B. Wives are identified using the `Relationship to Head’ variable, with an additional check to ensure consistency between the `Age of Individual’ and `Age of Wife’ variables. The demographic variables are extracted directly from the raw data as the `Age of Individual', `\# Children in Family Unit', and `Highest Grade'/`Completed Education' variables. Similarly, head-of-household and wife income is extracted as the `Head labor income' and `Wife labor income', respectively. We also extract wife's hours worked variable, and household size.  Following \textcite{altuug1998effect}, the consumption variable is defined as a measure of food consumption. I construct this variable in line with their approach, which they describe as follows: the consumption variable ``for a given year is obtained by summing the values of annual food expenditures for meals at home, annual food expenditures for eating out, and the value of food stamps received for that year. We then measured consumption expenditures for year $t$ by taking $0.25$ of the value of this variable for year $t-1$ and $0.75$ of its value for year $t$.'' Each of the monetary variables are adjusted for inflation using FRED's Personal Consumption Expenditures implicit price deflator \parencite{BEA_PCE_Deflator}. Wife wages are constructed as labor income divided by hours worked, and is thus undefined when hours worked is zero.

Filtering is applied as follows. I keep only wives that are observed for (at least) four consecutive periods and aged between 17 and 64 years, require positive head-of-household labor income, and drop any records with missing fields (wife/husband labor income, age, children, household size, hours, education).

Construction of the $Z_{i,t}$ variable follows the description of \textcite{altuug1998effect}, and thus requires log consumption ($c_{i,t})$ and log wage ($y_{i,t}$) regressions. Specifically, in the identity $Z_{i,t}=\eta_i\lambda_t\omega_t\exp(\gamma_3^\intercal x_{Wi,t})l_{i,t}$, I set $\eta_i$ and $\lambda_t$ as the coefficients from the log consumption regression
$$
\log c_{i,t} = \log \eta_i + \log \lambda_t + (hhn_{i,t}, age_{i,t}, educ_{i,t}, age_{i,t}^2 ) \gamma_C + \tilde{\eta}_{i,t},
$$
where $hhn_{i,t}, age_{i,t}, educ_{i,t}$ are the household size, age and education variables, respectively. Next, I set $\omega_t$, $\gamma_3^\intercal x_{Wi,t}$ based upon the log wage regression
$$
\log y_{i,t} = \log \zeta_i + \log \omega_t +  x_{Wi,t}^\intercal\gamma_3 + \tilde{\epsilon}_{i,t},
$$
for $x_{Wi,t}=(age^2_{i,t}, age_{i,t} \cdot educ_{i,t}, hours_{i,t-1}, hours_{i,t-2}, 1\{hours_{i,t-1}{>}0\}, 1\{hours_{i,t-2}>0\} )$, where $hours_{i,t}$ indicates the hours worked by wife $i$ in period $t$. Finally, I set $l_{i,t}$, the number of hours a woman chooses to spend at work conditional on participating, as the fitted values from the regression
$$
hours_{i,t} =    \tilde{\omega}_t + \tilde{\xi}_t \cdot \{hours_{i,t-1}>0\} +  x^\intercal_{Li,t}\gamma_L + \varepsilon_{it},
$$
for $x_{Li,t}=(age_{i,t},educ_{i,t},hhn_{i,t},kidsn_{i,t},hours_{i,t-1})$.

For use in estimating the DDC model, I normalize the continuous variables $Z_{i,t}$ and $hinc_{i,t}$ to have unit standard deviation, and remove the 2.5\% of observations that have very large values of $Z_{i,t}$ or $hinc_{i,t}$ (larger than 6.5 and 7.3, respectively).

\subsubsection{Model fit}\label{sec:model_fit}

Table \ref{table:model_fit} compares model-implied and empirical summary statistics for some key variables in the empirical model of Section \ref{sec:applic}. Specifically, the table presents first and second moments for the variables $(A_{t},Z_t,Hinc_t)$, which I refer to as the choice variable, the wage variable, and spouse earnings, respectively. The empirical moments are calculated directly from the data, whereas the model-implied moments are averages computed over 1{,}000{,}000 draws from the estimated model as described in footnote \ref{cf_desc}.

\begin{table}[H]
\centering
\begin{tabular}{lcl|ccccc}
   \hline
\hline
 &  &  & $A_{1}$ & $A_{2}$ & $A_{3}$ & $A_{4}$ & $A_{5}$ \\ 
   \hline
\multirow{2}{*}{Mean} & \multirow{2}{*}{} & Est. & 0.6555 & 0.6531 & 0.6498 & 0.6473 & 0.6431 \\ 
   &  & Data & 0.6575 & 0.6578 & 0.6788 & 0.6900 & 0.6847 \\ 
   \hline
\multirow{8}{*}{Corr} & \multirow{2}{*}{$A_{1}$} & Est. &  & 0.4440 & 0.4408 & 0.4378 & 0.4347 \\ 
   &  & Data &  & 0.6389 & 0.5177 & 0.4437 & 0.3904 \\ 
   & \multirow{2}{*}{$A_{2}$} & Est. &  &  & 0.4584 & 0.4545 & 0.4519 \\ 
    &  & Data &  &  & 0.6488 & 0.5422 & 0.4846 \\ 
   & \multirow{2}{*}{$A_{3}$} & Est. &  &  &  & 0.4722 & 0.4689 \\ 
   &  & Data &  &  &  & 0.6800 & 0.5765 \\ 
   & \multirow{2}{*}{$A_{4}$} & Est. &  &  &  &  & 0.4833 \\ 
   &  & Data &  &  &  &  & 0.6444 \\ 
   \hline
\hline
 &  &  & $Z_{1}$ & $Z_{2}$ & $Z_{3}$ & $Z_{4}$ & $Z_{5}$ \\ 
   \hline
\multirow{2}{*}{Mean} & \multirow{2}{*}{} & Est. & 0.5094 & 0.5575 & 0.6077 & 0.6601 & 0.7132 \\ 
   &  & Data & 0.5106 & 0.5715 & 0.6101 & 0.6495 & 0.6929 \\ 
   \hline
\multirow{2}{*}{Std} & \multirow{2}{*}{} & Est. & 0.7554 & 0.8766 & 0.9896 & 1.0985 & 1.2002 \\ 
   &  & Data & 0.7579 & 0.8509 & 0.8781 & 0.9117 & 0.9512 \\ 
   \hline
\multirow{8}{*}{Corr} & \multirow{2}{*}{$Z_{1}$} & Est. &  & 0.9411 & 0.8921 & 0.8483 & 0.8080 \\ 
   &  & Data &  & 0.8625 & 0.8080 & 0.7339 & 0.6903 \\ 
   & \multirow{2}{*}{$Z_{2}$} & Est. &  &  & 0.9457 & 0.8973 & 0.8528 \\ 
    &  & Data &  &  & 0.8920 & 0.7815 & 0.7151 \\ 
   & \multirow{2}{*}{$Z_{3}$} & Est. &  &  &  & 0.9493 & 0.9019 \\ 
   &  & Data &  &  &  & 0.8499 & 0.7779 \\ 
   & \multirow{2}{*}{$Z_{4}$} & Est. &  &  &  &  & 0.9509 \\ 
   &  & Data &  &  &  &  & 0.9051 \\ 
   \hline
\hline
 &  &  & $Hinc_{1}$ & $Hinc_{2}$ & $Hinc_{3}$ & $Hinc_{4}$ & $Hinc_{5}$ \\ 
   \hline
\multirow{2}{*}{Mean} & \multirow{2}{*}{} & Est. & 1.5548 & 1.5861 & 1.6206 & 1.6573 & 1.6967 \\ 
   &  & Data & 1.5563 & 1.6159 & 1.6529 & 1.6309 & 1.6644 \\ 
   \hline
\multirow{2}{*}{Std} & \multirow{2}{*}{} & Est. & 0.8985 & 0.9838 & 1.0568 & 1.1203 & 1.1762 \\ 
   &  & Data & 0.9005 & 0.9175 & 0.9489 & 0.9400 & 1.0056 \\ 
   \hline
\multirow{8}{*}{Corr} & \multirow{2}{*}{$Hinc_{1}$} & Est. &  & 0.8951 & 0.8178 & 0.7558 & 0.7047 \\ 
   &  & Data &  & 0.8372 & 0.7674 & 0.7274 & 0.7174 \\ 
   & \multirow{2}{*}{$Hinc_{2}$} & Est. &  &  & 0.9116 & 0.8413 & 0.7836 \\ 
    &  & Data &  &  & 0.8272 & 0.7686 & 0.7342 \\ 
   & \multirow{2}{*}{$Hinc_{3}$} & Est. &  &  &  & 0.9222 & 0.8582 \\ 
   &  & Data &  &  &  & 0.8518 & 0.7949 \\ 
   & \multirow{2}{*}{$Hinc_{4}$} & Est. &  &  &  &  & 0.9300 \\ 
   &  & Data &  &  &  &  & 0.8353 \\ 
   \hline
\hline
\end{tabular}
\caption{Mean, standard deviation (``Std'') and correlation matrix for each of the labor force participation variable $A_t$, wage variable variable $Z_t$, and head-of-household earnings variable $Hinc_t$. ``Data'' refers to the sample moments, ``Est.'' refers to the model-implied moments based on 1{,}000{,}000 draws from the estimated model.} 
\label{table:model_fit}
\end{table}

\subsubsection{Estimation with finite types}\label{app:appl_AJ}

For comparison, I estimate the model under the assumption that $\beta_i$ has three points of support using the iterative method of \textcite{arcidiacono2003finite}. I initialize the algorithm at $(\tilde{\beta}-1,\tilde{\beta},\tilde{\beta}+1,\tilde{\gamma}^\intercal)^\intercal$, where $(\tilde{\beta},\tilde{\gamma}^\intercal)^\intercal$ is the estimate of the parametric model (i.e., with $\beta_i$ assumed to be degenerate with unknown support), and continue the iterative steps until the average (over the parameter vector) percent change in the absolute value of the parameter is less than 0.025\%.

\begin{table}[ht]
\centering
\begin{tabular}{ccccc}
  \hline\hline Intercept & $hinc_{i,t}$ & $kids_{i,t}$ & $age_{i,t}$ & $educ_{i,t}$   \\  \hline
-2.473 & -0.298 &  0.087 & -0.626 &  0.304 \\ 
  {\small(0.1293)} & {\small(0.0276)} & {\small(0.0780)} & {\small(0.0785)} & {\small(0.0750)} \\ 
   \hline
\hline
\end{tabular}
\caption{Point estimates of $\gamma$ for the participation model of Section \ref{sec:applic} under the assumption that $\beta_i$ has three points of support, using the estimator of \textcite{arcidiacono2003finite}. Standard errors are in parentheses, calculated as the standard deviation of the estimator over 1{,}000 bootstrap samples.} 
\label{table:altug_AJ}
\end{table}

\begin{figure}[H] \centering \includegraphics[width=\textwidth]{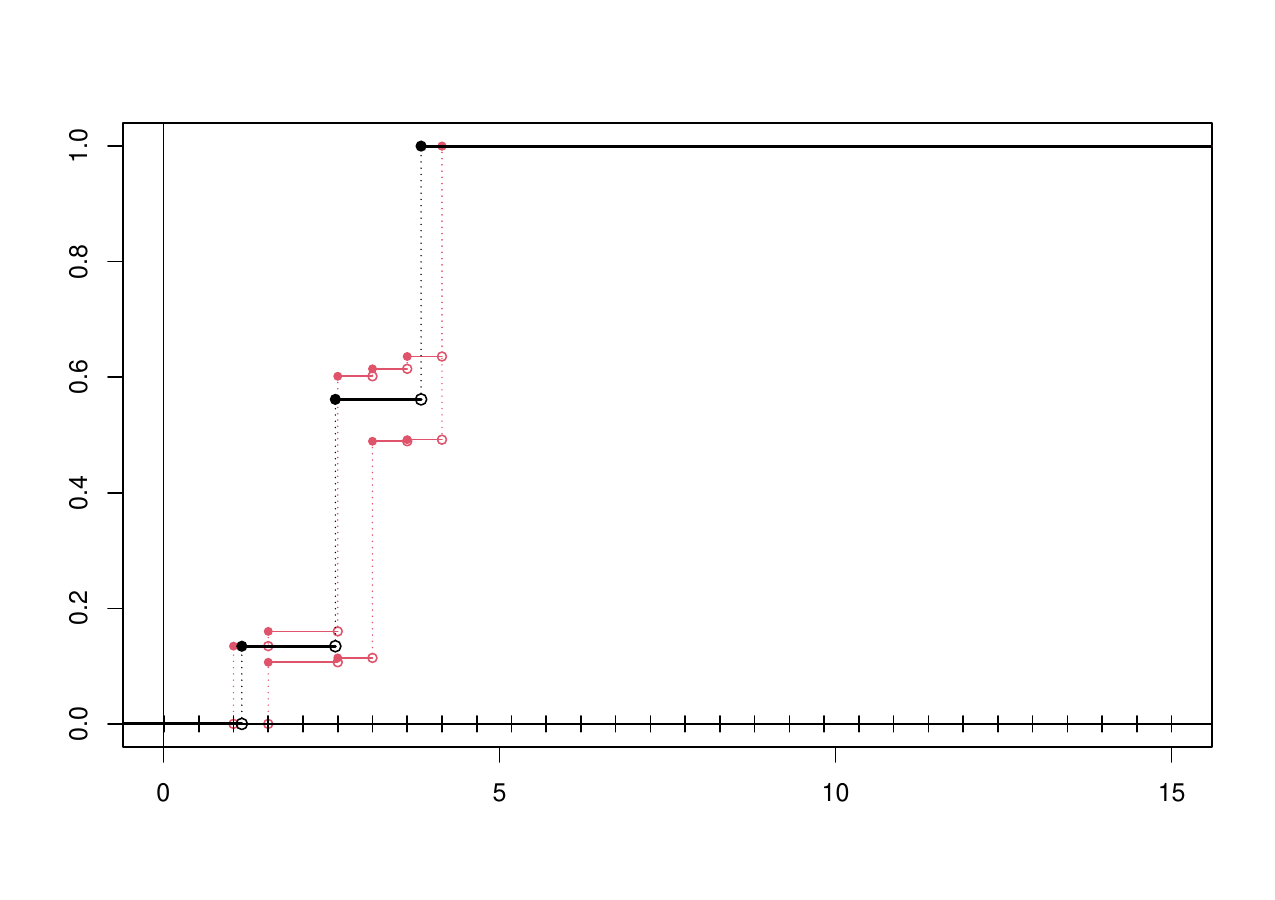}
\vspace*{-20mm} 	\caption{\small{Estimated distribution of
  $\beta_{i}$ for the participation model of Section \ref{sec:applic} under the assumption that $\beta_i$ has three points of support, using the estimator of \textcite{arcidiacono2003finite}. The black curve represents the point estimate. The red curves represent  bootstrapped 95\% pointwise confidence intervals for the c.d.f evaluated at the knots of the sieve space used for the estimator of Section \ref{sec:applic} (indicated by the ticks on the x-axis).}}\label{fig:altug_AJ}
\end{figure}

\subsubsection{Standard errors for the counterfactual estimates}

Table \ref{table:cf_std_errors} presents standard errors for the counterfactual results in Table \ref{table:cf_data}.

\begin{table}[ht]
\centering
\begin{tabular}{l|cccccc}
  \hline\hline & \multicolumn{6}{c}{Quantile of labor productivity $\beta$}  \\ Wage increase & $q_{0.01}$ & $q_{0.2}$ & $q_{0.4}$ & $q_{0.6}$ & $q_{0.8}$ & $q_{0.99}$ \\ 
   \hline
0\% & 0.0230 & 0.0488 & 0.0606 & 0.0238 & 0.0119 & 0.0098 \\ 
  5\% & 0.0247 & 0.0483 & 0.0602 & 0.0236 & 0.0117 & 0.0096 \\ 
  10\%  & 0.0264 & 0.0478 & 0.0597 & 0.0234 & 0.0116 & 0.0095 \\ 
  15\%  & 0.0282 & 0.0473 & 0.0593 & 0.0232 & 0.0114 & 0.0093 \\ 
  20\%  & 0.0298 & 0.0468 & 0.0589 & 0.0229 & 0.0112 & 0.0092 \\ 
  25\%  & 0.0314 & 0.0463 & 0.0584 & 0.0227 & 0.0111 & 0.0090 \\ 
   \hline
Elasticity: & 0.2321 & 0.0478 & 0.0269 & 0.0098 & 0.0042 & 0.0035 \\ 
   \hline
\hline
\end{tabular}
\caption{Bootstrapped standard errors for counterfactual labor force participation rates in Table \ref{table:cf_data}. Each cell presents bootstrapped standard errors for the corresponding cell in Table \ref{table:cf_data}, computed as the standard deviation of 1{,}000{,}000 bootstrap estimates.} 
\label{table:cf_std_errors}
\end{table}

\subsection{Additional state variables}\label{sec:discrete_state}

As claimed in Remark \ref{rem:discrete}, the results of the paper apply immediately to the case that there are additional state variables. This section states conditions that are sufficient for Theorems \ref{thm:inject} and \ref{thm:inf} for the $k\ge \dim(\beta)+ 1$ case. Intuitively, the assumptions require that conditions Assumptions \ref{As:inf_puh} and \ref{As:inf_os} apply to the first $\dim(\beta)+1$ elements of the state vector, leaving the remaining elements largely unrestricted. For instance, the additional variables may be discrete or binary. Analogous conditions can be provided for the models in Section \ref{sec:extensions}.

In this section, denote the observed state vector as $X_t=(Z_t^\intercal,W_t^\intercal)^\intercal$ for $Z_t\in\mathbb{R}^{\dim(\beta)+1}$.

\begin{manualassumption}{I2$^{\bm{add}}$}\label{As:inf_puh_add}
\begin{enumerate*}[label=(\roman*)]
    \item $S_t=(X_t^\intercal,\beta^\intercal)^\intercal\in\mathbb{R}^{k+J}$, and $k\ge J+1$. Denote $X_t=(Z_t^\intercal,W_t)^\intercal$ with $\dim(Z_t)=J+1$. For each $x\in \mathrm{Supp}(X_1)$, $\beta\mid X_1=x$ admits a bounded density $f_{\beta|X_{1}}$.
    \item $u(s,a)=x^\intercal\left(\beta_{a},\;\gamma_{a}^\intercal,\;\delta_a^\intercal\right)^\intercal,$ for $\gamma_a\in\mathbb{R}^{J}$.
\item
The probability distribution of $X_{t+1}$ conditional upon $(A_t,X_t)=(a,x)$ has no singular components, and the associated probability density and mass functions are real analytic functions of $z$ with bounded analytic
continuations to $\mathbb{R}^{J+1}$.
\item Assumptions \ref{As:inf_puh}(iii) and (iv).
\end{enumerate*} 
\end{manualassumption}

\begin{corollary}[Injectivity with additional state variables]
Assume \ref{As:inf_bl} and \ref{As:inf_puh_add}. Let $\mathcal{X}\subset \mathrm{Supp}(X_t)$ be such that $\{z~\colon~ (z^\intercal,w^\intercal)^\intercal\in\mathcal{X}\}$ contains a non-empty open set. Also, let $\mu$ be an absolutely continuous finite signed measure over set $\mathrm{Supp}(\beta)$. If
    \begin{equation*}
        \int P(a,x,b)d\mu(b)=0 \quad \text{for almost every } \forall (a,x)\in A \times\mathcal{X},
    \end{equation*}
    then $\mu=0$.
\end{corollary}

\begin{manualassumption}{I3$^{\bm{add}}$}\label{As:inf_os_add} For all $x\in\mathrm{Supp}(X_1)$, $\exists ~a\in A$ such that: (i) $\mathrm{Supp}(Z_{2}\mid{X}_1=x,A_1=a)$ and $\mathrm{Supp}(Z_{3}\mid{X_2}\in{\mathrm{Supp}(X_2\mid X_{1}=x,A_1 =a)},A_2=0)$ contain a non-empty open set; (ii) Assumption \ref{As:inf_os} (ii).
\end{manualassumption}

\begin{corollary}[Identification with additional state variables]
Assume the distribution of $(X_{t},A_{t})_{t=1}^{T}$ is observed for $T\ge4$, generated from agents solving the model of equation (\ref{eq:inf_prob}) satisfying assumptions \ref{As:inf_bl}, \ref{As:inf_puh_add} and \ref{As:inf_os_add}. Then $(\gamma,\delta,f_{\beta|X_{1}})$ is point identified.
\end{corollary}

\printbibliography[title={References for Online Appendix}]
\end{refsection}

\end{document}